\documentclass[a4paper,onecolumn,superscriptaddress,10pt,shorttitle=papers]{compositionalityarticle}
 \pdfoutput=1
\usepackage[utf8]{inputenc}
\usepackage[english]{babel}
\usepackage[T1]{fontenc}
\usepackage{etex}
\usepackage{hyperref}
\usepackage{float}

\usepackage{subfigure}
\usepackage{graphicx,subfigure}

\usepackage{comment}

\usepackage{tikz-network}
\usepackage{tikz}
	\usetikzlibrary{arrows}
	\usetikzlibrary{arrows.meta} 
 \usetikzlibrary{sbgn}
	\usetikzlibrary{calc}
	\usetikzlibrary{cd}
	\usetikzlibrary{positioning}

\usepackage[all]{xy}
\usepackage{lineno} 

\usepackage{enumerate}

\usepackage[thinc]{esdiff}

\usepackage[normalem]{ulem}
\usepackage{pdfpages}
\usepackage{graphicx}

\usepackage{url}

\usepackage[numbers,sort&compress]{natbib}
\usepackage[nameinlink]{cleveref}

\usepackage{xcolor}

\usepackage{mathtools}
\usepackage{mathrsfs}

\usepackage{amssymb}
\usepackage[mathscr]{eucal}
\usepackage{amsmath}
\usepackage{amscd}
\usepackage{amsthm}
\theoremstyle{definition}
\newtheorem{theorem}{Theorem}[section]
\newtheorem{definition}[theorem]{Definition}
\newtheorem{lemma}[theorem]{Lemma}
\newtheorem{proposition}[theorem]{Proposition}

\newtheorem{remark}[theorem]{Remark}
\newtheorem{example}[theorem]{Example}

\renewcommand{\arraystretch}{-1.7}

\numberwithin{equation}{section}

\newcommand{\lbr}{\lbrace}
\newcommand{\rbr}{\rbrace}
\newcommand{\mc}{\mathcal}
\newcommand{\mb}{\mathbb}

\newcommand{\ra}{\rightarrow}

\def\og{\leavevmode\raise.3ex\hbox{$\scriptscriptstyle\langle\!\langle$~}}
\def\fg{\leavevmode\raise.3ex\hbox{~$\!\scriptscriptstyle\,\rangle\!\rangle$}}

\title{A mathematical framework to study organising principles in graphical representations of biochemical processes}

\author[1]{Adittya Chaudhuri}
\email{adittya.chaudhuri@uni-rostock.de, chaudhuriadittya@gmail.com}
\orcid{0000-0002-1703-5889}
\affiliation[1]{University of Rostock, Institute of Computer Science, Rostock, Germany.}

\author[2,3]{Ralf Köhl}
\email{koehl@math.uni-kiel.de}
\orcid{0000-0003-0105-0029}
\affiliation[2]{Christian-Albrechts-University Mathematics Seminar, Kiel, Germany.}

\affiliation[3]{Kiel Nano, Surface and Interface Science, Christian-Albrechts-University, Kiel, Germany}

\author[4,5,6]{Olaf Wolkenhauer}
\email{olaf.wolkenhauer@uni-rostock.de}
\orcid{0000-0001-6105-2937}
\affiliation[4]{University of Rostock, Department of Systems Biology \& Bioinformatics, Rostock, Germany.}
\affiliation[5]{Leibniz-Institute for Food Systems Biology; Technical University of Munich, Freising; Germany}
\affiliation[6]{Stellenbosch Institute for Advanced Study, South Africa}

\begin{document}

\maketitle

\begin{abstract}
The complexity of molecular and cellular processes forces experimental studies to focus on subsystems. To study the functioning of biological systems across levels of structural and functional organisation, we require tools to compose and organise networks with different levels of detail and abstraction. Systems Biology Graphical Notation (SBGN) is a standardised notational system that visualises biochemical processes as networks. Despite their widespread adoption, SBGN languages remain purely visual and lack an underlying mathematical framework, limiting their compositional analysis, abstraction, and integration with formal modelling approaches. SBGN comprises three complementary visual languages—Process Description (SBGN-PD), Activity Flow (SBGN-AF), and Entity Relationship (SBGN-ER)-each operating at a different level of abstraction. 

In this manuscript, we introduce a category-theoretic formalism for SBGN-PD,  a visual language to describe biochemical processes as biochemical reaction networks. Using the theory of structured cospans, we construct a symmetric monoidal double category whose horizontal 1-morphisms correspond to SBGN-PD diagrams. We also analyse how a designated subnetwork influences the surrounding network and how external entities, in turn, affect the internal reactions of the subnetwork. Our work addresses a key gap between biological visualisation and mathematical structure. It provides precise organising principles for SBGN-PD, including compositionality, enabling the construction of large biochemical reaction networks from smaller ones, and zooming out, allowing the abstraction of detailed biochemical mechanisms while preserving their functional interfaces.  Throughout the paper, the proposed framework is illustrated using standard SBGN-PD examples, demonstrating its applicability to large-scale biochemical reaction networks.

\end{abstract}

\section{Introduction}\label{Sec:1}

Information about biological processes is available in databases, for which KEGG \cite{kanehisa_kegg_2024} is a widely known example. Various markup languages, like KEGG ML \cite{kyoto1995} or BioPax \cite{demir_biopax_2010}, have been developed to represent biological processes graphically. Encoding of biological processes as networks allows linking a graphical representation with molecular information, including references to the literature, links to gene and disease ontologies, and links to databases containing chemical and structural information. For quantitative analyses of biological processes, we need to translate the network representation to a mathematical model. To this end, markup languages like CellML \cite{CellMLproject}, PharmML \cite{https://doi.org/10.1002/psp4.57}, and SBML \cite{hucka_systems_2003, finney_systems_2003, keating_sbml_2020} allow the encoding of mathematical models in a standardized computational format. The BioModels database \cite{malik-sheriff_biomodels15_2019} is a repository providing over 1000 models of biological processes encoded in SBML. 

\textit{Systems Biology Graphical Notation} (SBGN) \cite{le_novere_systems_2009} has been developed
to promote an efficient, unambiguous exchange and reuse of biological information related to signalling pathways, metabolic networks, and gene regulatory networks within the scientific community . Over the years, SBGN has become a widely used standardised graphical notational system to visualise biological processes at different level of detail. 
SBGN offers three different but complementary visual languages, namely \textit{Process Description} (SBGN-PD) \cite{rougny_systems_2019}, \textit{Activity Flow} (SBGN-AF) \cite{mi_systems_2015}  and \textit{Entity Relationship} (SBGN-ER) \cite{sorokin_systems_2015}, each focussing on a different level of abstraction. SBGN-PD works at the detailed biochemical reaction level,
SBGN-AF represents the connections and interactions between biochemical entities in terms of information flow, and SBGN-ER shows how entities influence each other's actions and behaviours. However, SBGN is only a visual tool and is not a mathematical representation of biological processes.  There have been a few attempts to fill the gap between the SBGN visualisations and mathematical representations, including \cite{rougny2016qualitative}, which formalised SBGN-PD diagrams using asynchronous automata networks, or \cite{Loewe2011} using Hybrid Functional Petri Net (HFPN) and \cite{cherdal2018sbgn2hfpn} using textual representations (SBGNtext). However, all of these attempts fail to support the study of \textit{organising principles}, specifically,

\begin{itemize}
\item[(a)] formal ways to compose a collection of biochemical molecular/cellular networks into a composite network and decompose a large network into smaller subcomponents,
\item[(b)] formal ways to zoom in and zoom out details in a biochemical molecular/cellular network,
\item[(c)] formal compatibility features between the said compositionality and the said process of zooming-out and zooming-in details,
\item[(d)] formal ways to study how a particular portion of a network depends on the entities produced outside it, and conversely, how the network outside this particular portion depends on the entities produced inside the said portion.
\end{itemize}
To this day, most experimental studies will only be able to address subsystems, parts of a larger whole. There is, thus, a need to compose networks, and ideally, we must have tools available to study the organisation of large networks independent of the simulation framework chosen. Taking human diseases as an example, virtually all processes linked to a disease phenotype involve various cell types and many molecule types. The Atlas of Inflammation Resolution (AIR) \cite{SERHAN2020100894}  is an example where information for over twenty thousand reactions involved in the resolution of acute inflammation is gathered. Processes of this size are never studied as a whole. Experimental studies are usually focusing on, and are practically limited to, networks of relatively small sizes. Thus, especially when the target network size is very large, formal organisational principles, as stated above, and a formal framework of graphical representations, which abstracts from concrete modelling and simulation formalisms, are expected to be helpful for quantitative and qualitative analyses. 

  Our present manuscript uses Applied Category Theory (ACT), especially Baez et al.'s theory of structured cospans in the framework of symmetric monoidal double categories \cite{baez2020structuredcospans}  to develop the previously mentioned organising principles for biochemical molecular and cellular networks admitting SBGN-PD visualisations. The framework of symmetric monoidal double categories has been previously successfully used to study the composition of networks, where horizontal 1-morphisms represent subsystems with interfaces. Examples close to our goal are reaction networks modeled with Petri nets \cite{MR4085076,MR4483767,compositionality:13637}. For example, in \cite{compositionality:13637}, Aduddell et al. introduced a compositional framework for Petri nets with signed links to model regulatory biochemical networks. These works and the availability of large numbers of SBGN-encoded networks motivated our present effort. Since SBGN also allows features like compartments and submaps, we hope our formal compositional framework would provide us with a new perspective to study multilevelness in biological systems. With the increasing availability of SBGN-PD visualisations in biological databases, formal organisational principles for a generic SBGN-PD would provide biologists with generic methods to analyse the behaviour of a generic biochemical reaction network at multi scales, irrespective of the network size.
 
Before we move on to the paper's organisation, we say a few words about our choice to use Applied Category Theory for our goal. For the last decade or so, ACT has established itself as a Category Theory-based discipline in mathematics for studying the behaviour of large-scale systems by composing the behaviour of its subsystems. It has been successfully applied to a wide range of areas, including biochemical regulatory networks \cite{compositionality:13637}, chemical reaction networks \cite{MR3694082, MR4483767}, Markov processes \cite{MR3478745}, epidemiological modelling \cite{MR4489365, Baez_2023}, data structures \cite{MR4528606, MR4610076, althaus2023compositionalalgorithmscompositionaldata}, game theory \cite{MR3883754}, deterministic dynamical system \cite{Libkind_2022} etc., to name a few. In fact, ACT-based frameworks allow a level of abstraction or generalisation that encompasses a range of concrete modelling approaches like Petri Nets \cite{MR4085076,MR4483767, MR3694082}, ODEs \cite{MR4085076, MR3694082}, stochastic processes \cite{MR3478745}, graphs \cite{master2021composingbehaviorsnetworks, compositionality:13637}, to name a few, and their \textit{functorial interrelationships} (interrelationships which respect compositionality in a suitable way). Many of the ACT-based formalisations have also been successfully translated into user-friendly software for the purpose of computational studies via platforms like Algebraic Julia \cite{noauthor_algebraicjulia_nodate}. The framework presented in this paper contributes to these efforts, by focusing on graphical visualizations of biochemical processes.

\subsection{Structure of the Paper}
The paper is organised as follows. In Section \ref{Sec:2}, we illustrate our main results informally through some standard examples of biochemical reaction networks visualised in SBGN Process Description. We begin formally developing our theory in Section \ref{sec:3}. Subsection \ref{subsection:Process-network-and-process-species} introduces the notion of \textit{process networks} and \textit{process species}, which model biochemical reaction networks and biochemical reactions,  respectively, and illustrate these notions through some standard SBGN Process Descriptions. In Subsection \ref{Subsection:Transforming a process network into another: Morphisms of process networks}, we introduce the notion of a \textit{morphism of process networks} and derive some of their properties. Using SBGN-PD examples, we illustrate their interpretations, such as \textit{zoom-out} and \textit{zoom-in} details within a biochemical reaction network, and distinguishing networks of different types. Section \ref{Section4} forms the heart of our formal development. In Subsection \ref{subsection:category-of-process-networks}, we construct a category of process networks, and prove that it contains all finite colimits. We also compare our framework with the notion of a Petri net with link, a notion recently introduced by Aduddell et al in \cite{compositionality:13637} for modelling regulatory networks. We start Subsection \ref{Subsection:Process networks with interfaces as structured cospans} by introducing the notion of an \textit{open process network} or a \textit{process network with an interface} using the Baez et al.'s theory of structured cospans \cite{baez2020structuredcospans}. Then, using the results from Subsection \ref{subsection:category-of-process-networks},  we construct a symmetric monoidal double category whose horizontal 1-morphisms are open process networks. Furthermore, we make a pair of observations that enable us to illustrate various elements in the above symmetric monoidal double category with concrete examples visualised in SBGN-PD. In Subsection \ref{Subsection:Macroscope of a process network: A tool to study the influence of a process network on its environment and vice versa}, we introduce the notion of a \textit{process subnetwork of a process network} and its \textit{environment}. Next, we introduce a tool that we call a \textit{macroscope of a process subnetwork with respect to a network}. This notion allows us to study the influence of the process subnetwork on its environment and vice versa. We end the Section \ref{Section4} by applying our macroscope on a standard SBGN-PD. In Section \ref{Section5}, we illustrate via examples the  translation of concrete SBGN-PD digrams into process networks. Finally, in Section \ref{Sec:discussion-and-conclusion}, we summarize our achievements and discuss some future directions and limitations of the current state of our work.

\section{An illustration of our results via examples}\label{Sec:2}
This section attempts to illustrate the main achievements of this manuscript through some standard examples of biochemical reaction networks visualised in  SBGN Process Description. We postpone our concrete theoretical development until Section \ref{sec:3}.

We begin our treatment with an example of a biochemical reaction visualised as an SBGN Process Description (Figure \ref{fig:process-species} (top)), and the visualisation of its formal abstraction (Figure \ref{fig:process-species} (bottom)), as done in our framework in the later sections. The reaction in Figure \ref{fig:process-species} (top) describes the activation of the molecule ERK (\textit{Extracellular Signal-Regulated Kinase}) through phosphorylation (attaching a \textit{phosphate group} P). The process requires energy generated by breaking down ATP (\textit{Adenosine Triphosphate}), thereby releasing ADP (\textit{Adenosine Diphosphate}). The activation of ERK is facilitated by the phosphorylated MEK (\textit{Methyl Ethyl Ketone}). In SBGN Process Descriptions, \textit{entity pool nodes} like  macromolecules and simple chemicals  are visualised by rectangular glyphs with rounded corners and circles respectively. The small inserted circles visualize covalent modifications (like the state of phosphorylation), and the small squares represent process nodes describing biochemical processes. A connecting arc between a biochemical entity and the process node denotes \textit{consumption}, a connecting arc (from the process node to a biochemical entity) with a black arrowhead represents\textit{ production} and a connecting arc (from a biochemical entity to the process node) with a small white circular head denotes \textit{catalysis}. 

\begin{figure}[htb]
\centering
\includegraphics[width=14cm]{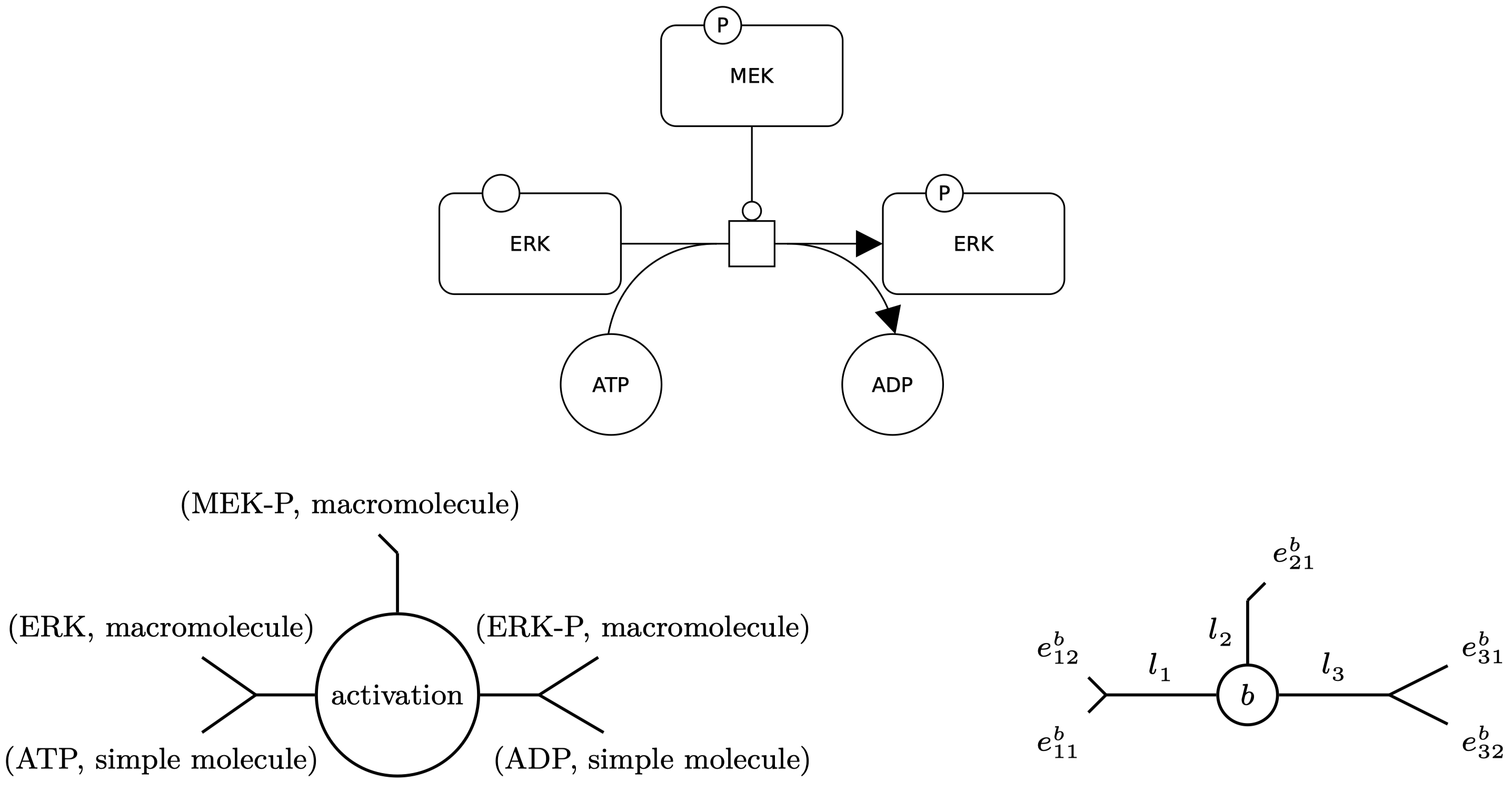}
\caption{Example of a SBGN Process Description  describing a biochemical reaction (top), and the visualization of its formal abstraction (bottom right) as done in our framework (Definition \ref{ARDefinition: Biochemical process networks}). In the bottom left, we illustrate a  translation of the SBGN-PD visualization to its abstraction, whose detail is explained  Example \ref{Example:MAPK}.}
\label{fig:process-species}
\end{figure}

The formal abstraction of SBGN-PD diagram, which we will define mathematically as a \textit{process species} in Definition \ref{ARDefinition: Biochemical process networks}, is visualised in the bottom right of Figure \ref{fig:process-species}. The process is drawn as a circle, labelled $b$, and the arcs labelled $l_1,l_2$ and $l_3$ model consumption, catalysis and production, respectively, of the process $b$. Symbols $e_{11}^{b}, e_{12}^{b}$ model respectively the simple chemical ATP and the macromolecule ERK, and are attached to the arc labelled $l_1$. The symbol $e_{21}^{b}$ models the macromolecule phosphorylated MEK which is attached to $l_2$. Symbols $e_{31}^{b}$ and  $e_{32}^{b}$ model the macromolecule phosphorylated ERK and the simple chemical ADP, respectively, and are attached to $l_3$. Furthermore, observe that the attachment of biochemical entities to arcs like $l_1$,$l_2$ and $l_3$ are shown with additional arcs. In the bottom left of Figure \ref{fig:process-species}, we illustrate the translation of the SBGN-PD (Figure \ref{fig:process-species} top) to our formal abstraction (bottom right of Figure \ref{fig:process-species}). Section \ref{Section5} discuss such translations in details.

Our theoretical framework (Section \ref{sec:3} and Section \ref{Section4}) provides us with organisational principles (as discussed in Section \ref{Sec:1}) for biochemical reaction networks admitting SBGN-PD visualisations.  We achieve such organisational principles through our two main results in this manuscript. Precisely, through Theorem \ref{ARMain Theorem1} we get (i) formal ways to compose a collection of biochemical molecular/cellular networks into a composite network and decompose a large network into smaller subcomponents, (ii) formal ways to zoom in and zoom out details in a biochemical molecular/cellular network, (iii) formal compatibility features between the said compositionality and the said process of zooming-out and zooming-in details. The formal study of how a particular portion of a network depends on the entities produced outside it, and conversely, how the network outside this particular portion depends on the entities produced inside the said portion, we obtain through Definition \ref{Defn:macroscope}.

\subsection{An illustration of Theorem \ref{ARMain Theorem1}}

In technical terms, Theorem \ref{ARMain Theorem1} produces a symmetric monoidal double category whose horizontal 1-morphisms can be interpreted as SBGN Process Descriptions,  and 2-morphisms can be interpreted as zoom-out or zoom-in operations on the details of a biochemical network visualised using a SBGN-PD. The composition laws and the monoidal product laws in the constructed symmetric monoidal double category provide us with a compositional framework for SBGN-PD  which remain compatible with our zoom-in and zoom-out operations. 

MAPK (\textit{Mitogen-Activated Protein Kinase}) is an intracellular signaling pathway activated by potentially multiple phosphorylations. It transduces external signals into specific cellular responses such as gene expression, proliferation, differentiation, or stress responses. In Figure \ref{fig:construction-of-cascade}, we illustrate how to build an SBGN-PD visualisation of the MAPK cascade in two different ways, viz.  by composing three reaction networks  ((a), (b) and (c)) and composing two reaction networks ((d) and (e)) using our Theorem \ref{ARMain Theorem1}. The dotted black lines denote compositions (interconnections). We can also formally compose (using the Theorem \ref{ARMain Theorem1}) the SBGN visualisation of the whole MAPK cascade with other biochemical networks to form a larger network, as visualized in Figure \ref{fig:construction-of-insulinsignalling}, where we demonstrate how to build insulin-like Growth Factor (IGF) signalling by composing the SBGN-PD visualisation of the MAPK cascade with two other biochemical molecular networks (marked with blue and grey) using our Theorem \ref{ARMain Theorem1}.

\begin{figure}[H]
\centering
\includegraphics[width=10cm]{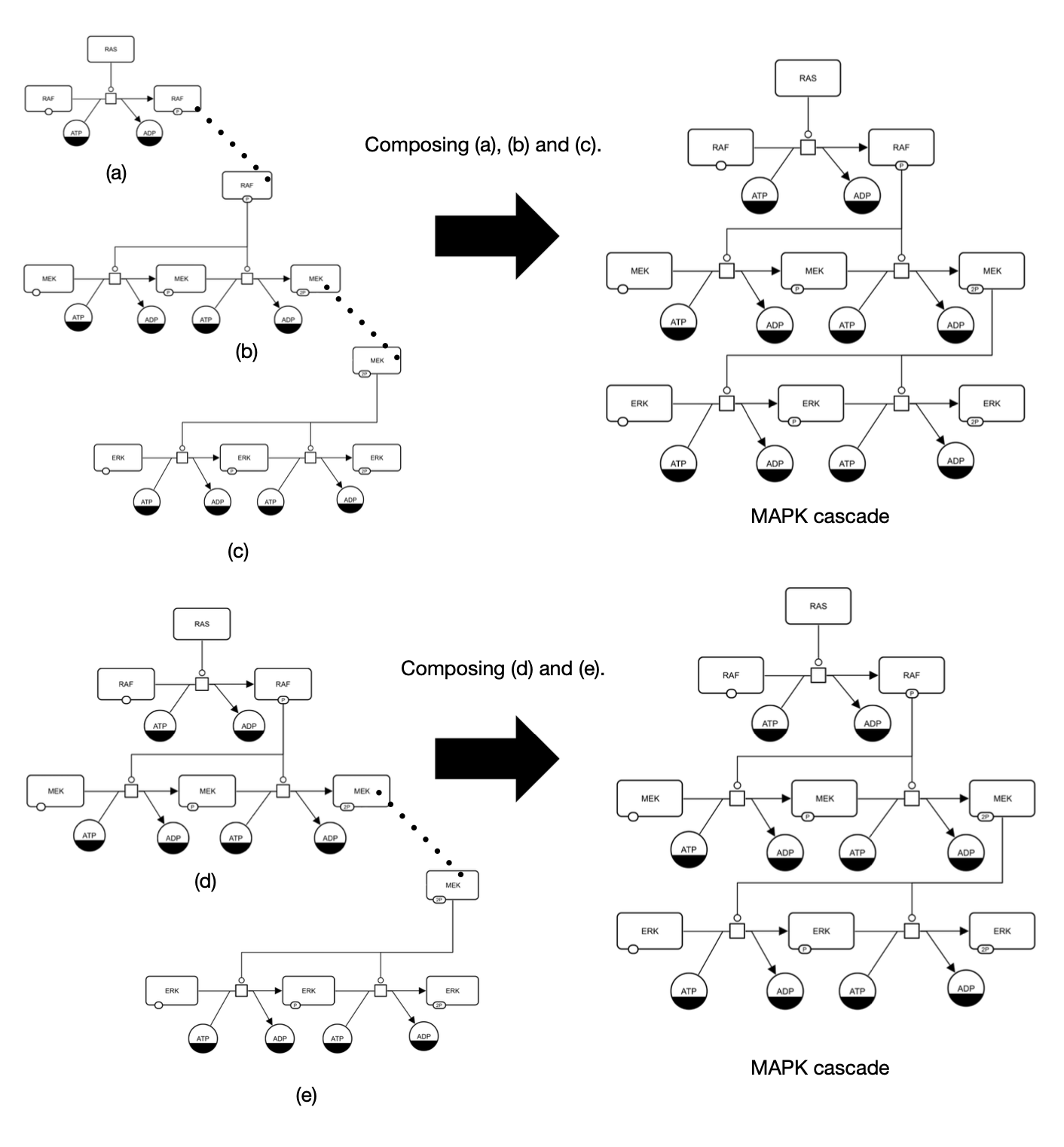}
\vspace*{-.5cm}
\caption{Illustration of building an SBGN-PD visualisation of the MAPK cascade by composing three reaction networks (a), (b) and (c), and two reaction networks (d) and (e), using the Theorem \ref{ARMain Theorem1}. SBGN images are derived from the MAPK cascade example on Page 65 in \cite{rougny_systems_2019}.}
\label{fig:construction-of-cascade}
\end{figure}

The dotted black lines denote compositions (interconnections). However, while visualizing the IGF signalling  using SBGN, usually an \textit{encapsulation node} called the \textit{submap} is used to hide the details of the MAPK cascade in the pathway, as shown in Figure \ref{fig:cascade-as-submap}. Here, the \textit{reference nodes} \textit{tags} show how the submap is connected to the rest part of the IGF signalling via the macromolecules RAS and ERK through \textit{equivalence arcs} (connecting the RAS and the tag RAS, and the ERK and the tag ERK).

\begin{figure}[htb]
\centering
\begin{center}
\includegraphics[width=10cm]{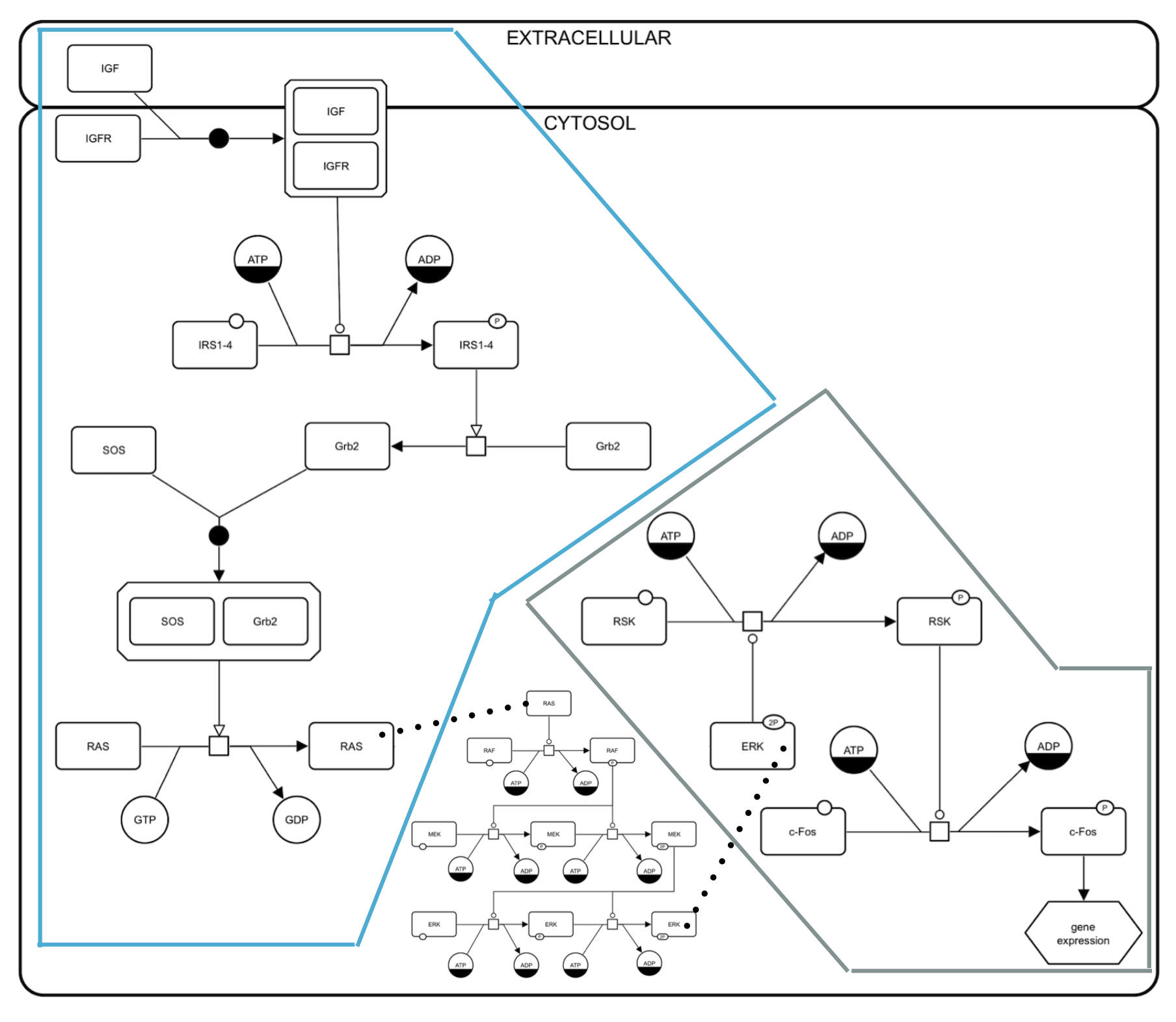}
\caption{An illustration of building an SBGN-PD visualisation of the Insulin-like Growth Factor (IGF) signalling by composing the MAPK cascade with two other biochemical molecular networks (marked in blue and grey) using the Theorem \ref{ARMain Theorem1}. SBGN images are derived from IGF signalling and MAPK cascade examples on Pages 64 and 65 in \cite{rougny_systems_2019}.}
\label{fig:construction-of-insulinsignalling}
\end{center}
\end{figure}

Observe that in Figure \ref{fig:construction-of-insulinsignalling}, we encounter some geometric shapes such as the ones highlighted in Figure \ref{fig:details-MAPK-insert}. According to SBGN PD language Level 1 Version 2.0 as in \cite{rougny_systems_2019}, we now briefly explain their meanings. The geometric shape Figure \ref{fig:details-MAPK-insert}(a) represents the process node \textit{association}. For example, in Figure \ref{fig:construction-of-insulinsignalling}, the macromolecules IGF and IGFR combine to form a \textit{complex} through association.

\begin{figure}[htb]
\centering
\begin{center}
\includegraphics[width=15cm]{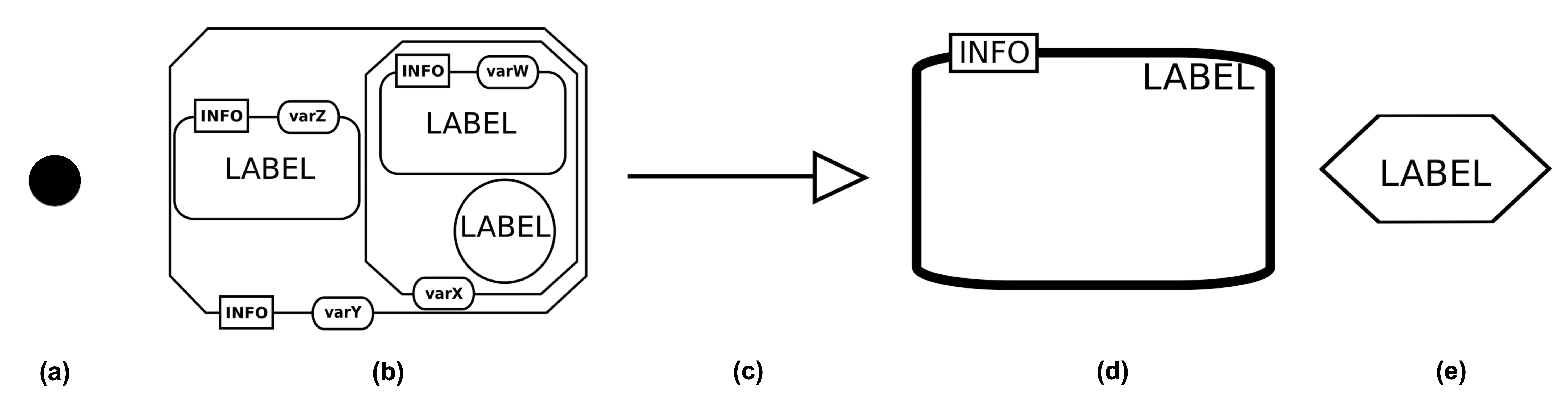}
\caption{Some geometric shapes used in Figure \ref{fig:construction-of-insulinsignalling}. Here, (a), (b), (c), (d) and (e), respectively, denote association, complex node, stimulation arc, compartment node and phenotype. SBGN images are derived from the reference card on Page 75 in \cite{rougny_systems_2019}.}
\label{fig:details-MAPK-insert}
\end{center}
\end{figure}

The geometric shape Figure \ref{fig:details-MAPK-insert}(b) denotes a \textit{complex node}, a biochemical entity comprising other biochemical entities like macromolecules, simple chemicals, multimers, or other complexes, connected by non-covalent bonds. In SBGN Process Descriptions, the stimulation  of a process by a biochemical entity is denoted by a connecting arc (from an entity to the process node) with a white arrowhead (see Figure \ref{fig:details-MAPK-insert}(c)). For example, in Figure \ref{fig:construction-of-insulinsignalling}, the macromolecule IRS1-4 stimulates the activation of the macromolecule Grb2. In SBGN, the geometric shape Figure \ref{fig:details-MAPK-insert}(d) denotes a \textit{compartment node}, representing a logical or physical structure. Every entity pool node, such as macromolecule, simple chemical, complex, etc., belongs to a compartment. Two identical entity pool nodes located in different compartments are considered different entities. In Figure \ref{fig:construction-of-insulinsignalling}, two compartments are used: extracellular and cytosol. In SBGN, geometric shape Figure \ref{fig:details-MAPK-insert}(e) denotes a \textit{phenotype}. In Figure \ref{fig:construction-of-insulinsignalling}, we have the phenotype \textit{gene expression}.

\begin{figure}[htb]
\begin{center}
\includegraphics[width=10 cm]{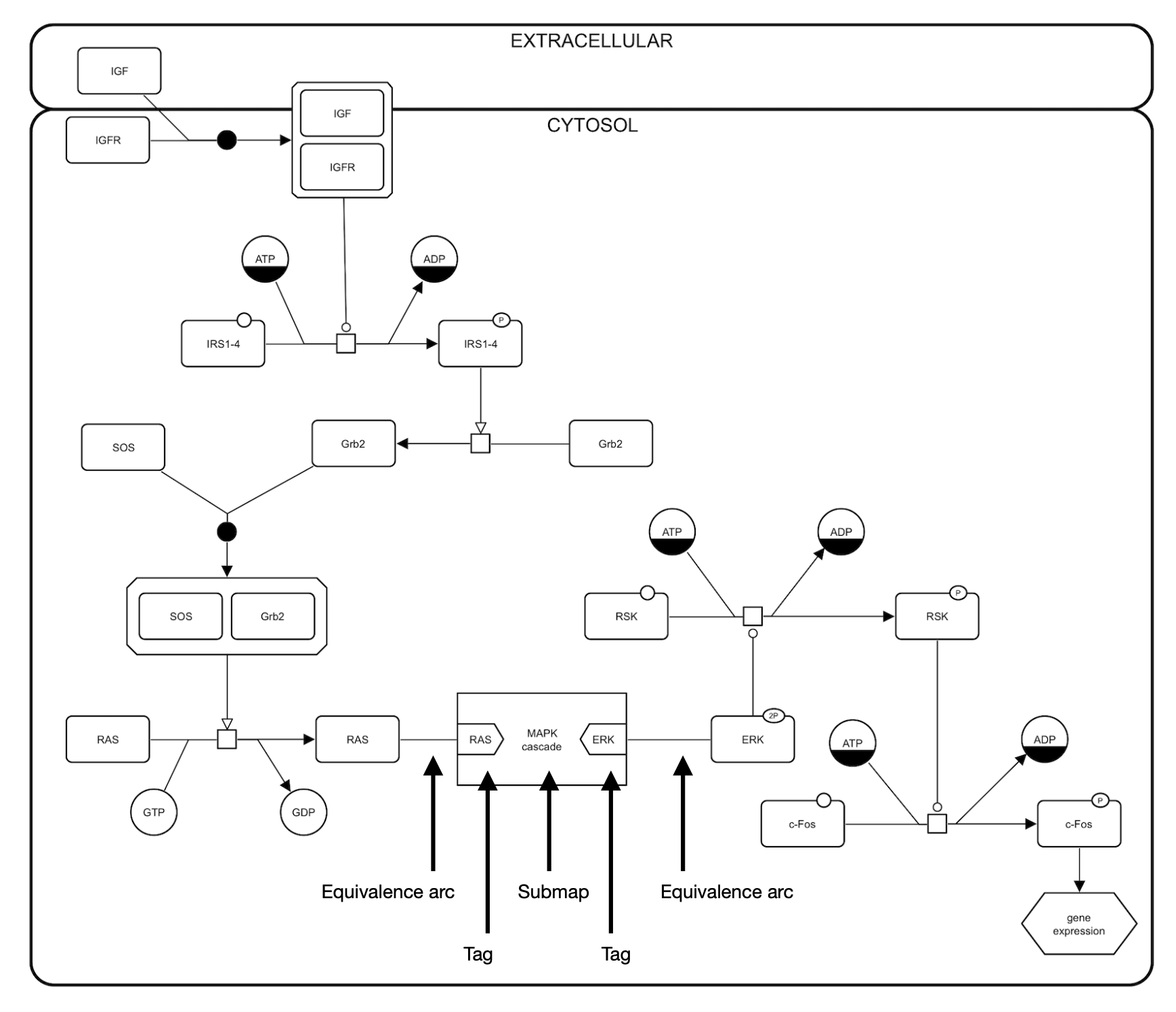}
\caption{An illustration of the encapsulation node submap, the reference node tag and the equivalence arc in the SBGN-PD visualisation of the IGF signalling. The SBGN image is taken from the example of IGF signalling on Page 64, \cite{rougny_systems_2019}.}
\label{fig:cascade-as-submap}
\end{center}
\end{figure}
Often to study complicated biochemical reaction networks, we purposefully omit details to obtain a broader view of the whole reaction network. Our Theorem \ref{ARMain Theorem1} provides us with a formal way to forget details from biochemical reaction networks such that the forgetting procedure (zooming-out procedure) is compatible with our formal compositional framework, as is illustrated in Figure \ref{fig:abstraction-of-structural-details}.  Here, we start with the SBGN-PD visualisations of two biochemical reactions numbered Figure \ref{fig:abstraction-of-structural-details}(1) and Figure \ref{fig:abstraction-of-structural-details}(2). Our Theorem \ref{ARMain Theorem1} allows us to zoom-out (shown with the thin black arrows) by forgetting ADP's and ATP's from the reactions, and in turn, we obtain the reactions Figure \ref{fig:abstraction-of-structural-details}(3) from Figure \ref{fig:abstraction-of-structural-details}(1), and Figure \ref{fig:abstraction-of-structural-details}(4) from Figure \ref{fig:abstraction-of-structural-details}(2). Then, again the Theorem \ref{ARMain Theorem1} let us combine Figure \ref{fig:abstraction-of-structural-details}(3) and Figure \ref{fig:abstraction-of-structural-details}(4)  to obtain the reaction network Figure \ref{fig:abstraction-of-structural-details}(6), and combine the reactions Figure \ref{fig:abstraction-of-structural-details}(1) and Figure \ref{fig:abstraction-of-structural-details}(2) to get the reaction network Figure \ref{fig:abstraction-of-structural-details}(5). More interestingly, Theorem \ref{ARMain Theorem1}   provides us with a canonical way of combining two zooming-out procedures $\big($ Figure \ref{fig:abstraction-of-structural-details}(1) to Figure \ref{fig:abstraction-of-structural-details}(3) and Figure \ref{fig:abstraction-of-structural-details}(2) to Figure \ref{fig:abstraction-of-structural-details}(4)$ \big)$ such that the combined zoom-out procedure is compatible with the composition process. More precisely, here, the combined zoom-out process takes the reaction network Figure \ref{fig:abstraction-of-structural-details}(5) to the reaction network Figure \ref{fig:abstraction-of-structural-details}(6). Here, the dotted lines denote compositions. We provide the technical details of the above-mentioned procedure in Example \ref{Example:horizontal-composition-of-2-morphisms}.


\begin{figure}[H]
\begin{center}
		\includegraphics[width=15.5cm]{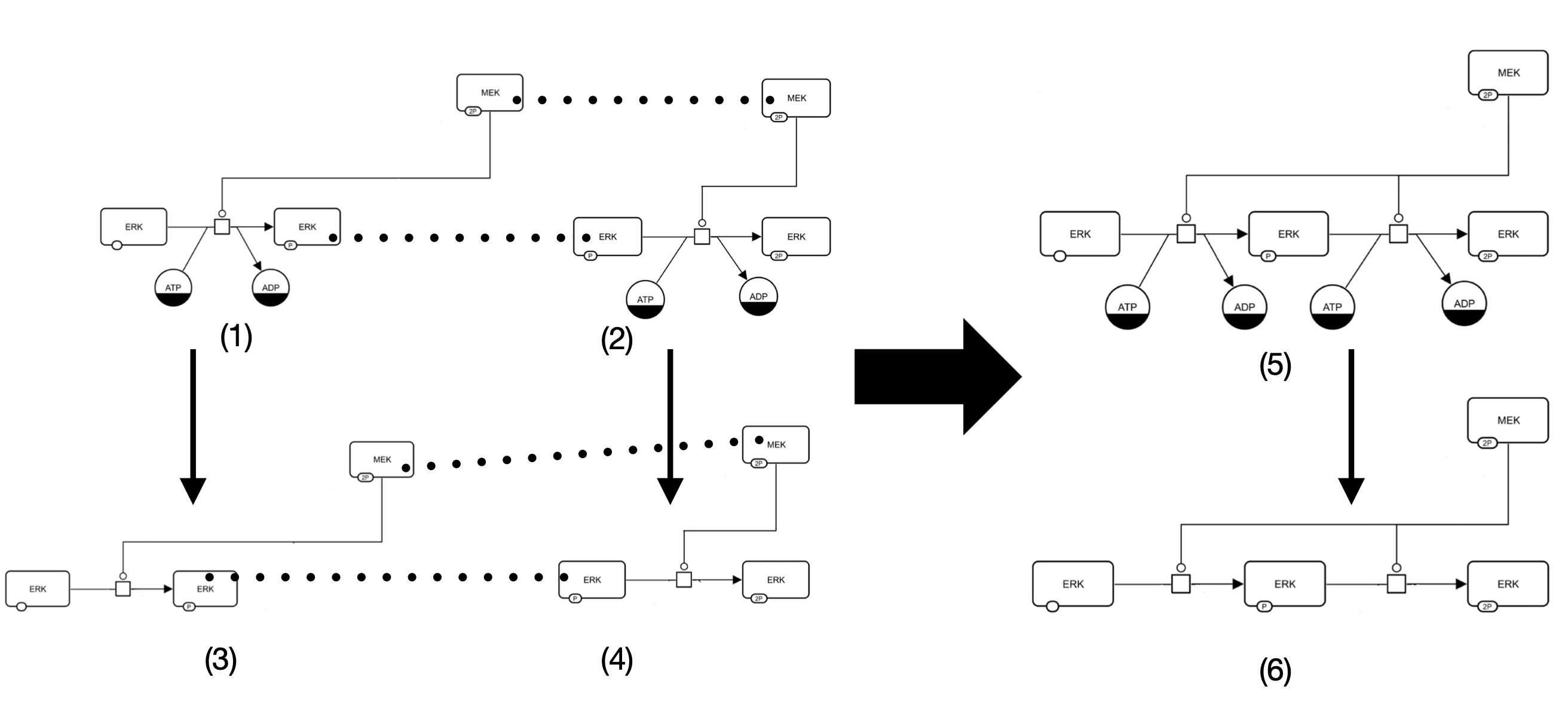}
  \caption{An illustration of formally zooming-out details in a biochemical reaction network using the Theorem \ref{ARMain Theorem1}. This figure, in particular, demonstrates how Theorem \ref{ARMain Theorem1} provides a compatibility between the formal zooming-out methods and the compositionality. SBGN images are derived from the MAPK cascade example on Page 65, \cite{rougny_systems_2019}.}
  \label{fig:abstraction-of-structural-details}
\end{center}
\end{figure}

\subsection{An illustration of Definition \ref{Defn:macroscope}}

Often, especially for disease-related purposes, it is essential to see how a particular portion of a biochemical network gets affected by the remaining part of the network and the converse, i.e. how that particular portion affects the rest of the network. We introduce a mathematical technique called a \textit{macroscope} (Definition \ref{Defn:macroscope}), which allows us to formalise such effects. We choose the name macroscope because it allows us to see the overall effect of a biochemical reaction network on its  surrounding network and vice versa. We illustrate the macroscope in Figure \ref{Figure:Macroscopeoverview}.

\begin{figure}[H]
\begin{center}
\includegraphics[width=5cm]{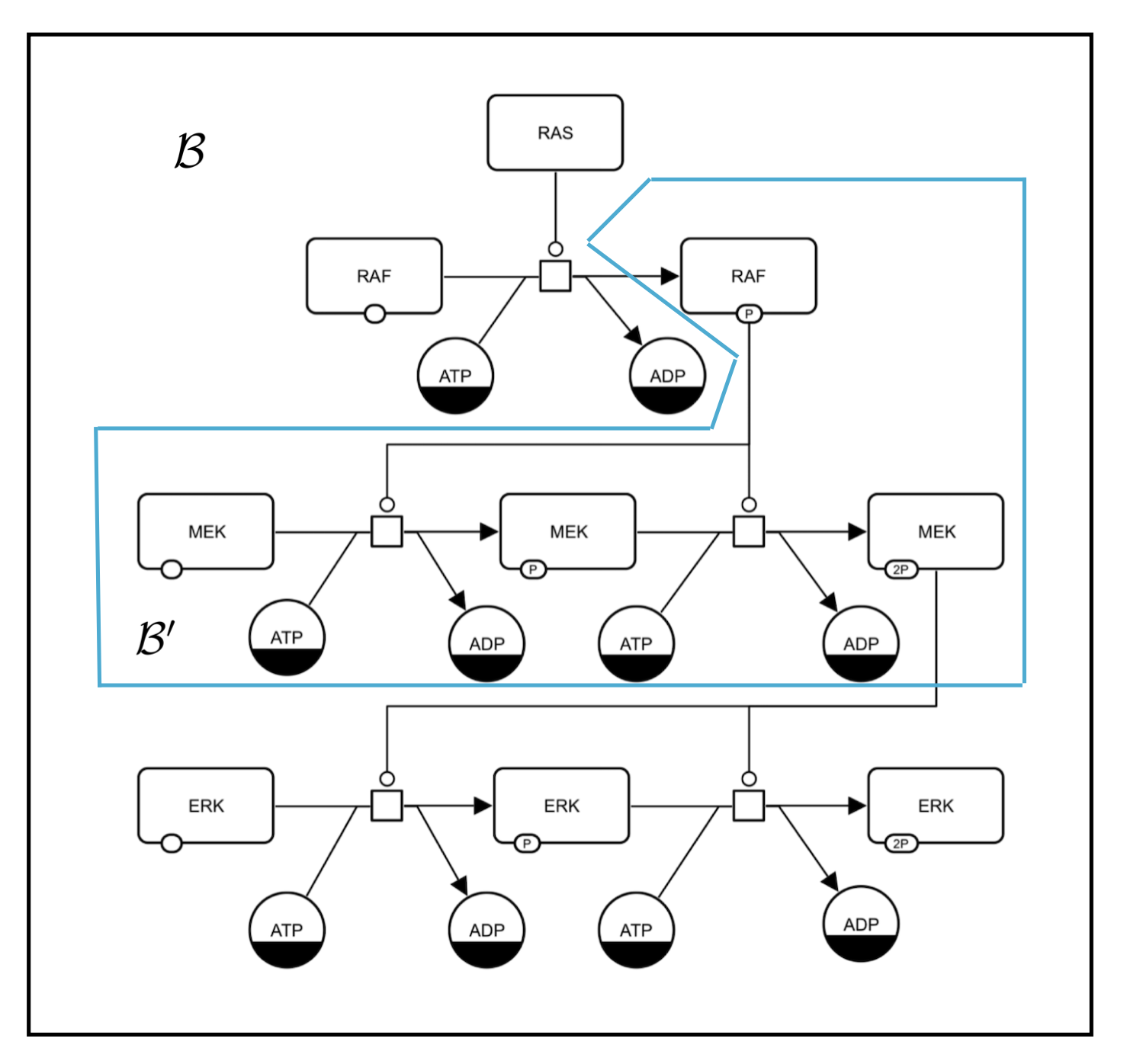}
\caption{An illustration of the macroscope, formally introduced in Definition \ref{Defn:macroscope}. Applying the macroscope in the blue marked portion, we can conclude that the outer portion is catalysing the blue marked portion through the phosphorylated RAF, and in turn, the blue marked portion is catalysing the outer portion through the double phosphorylated MEK. The SBGN image is taken from the MAPK cascade example on Page 65 in \cite{rougny_systems_2019}.}
\label{Figure:Macroscopeoverview}
\end{center}
\end{figure}

\section{Process species, process networks, and their transformations}\label{sec:3}
This section introduces mathematical structures called \textit{process species}, \textit{process network} and \textit{morphisms of process networks}, which model, respectively, a biochemical reaction, biochemical reaction network and a process of transforming a biochemical reaction network into another.

\subsection{Process networks and process species}\label{subsection:Process-network-and-process-species}

\begin{definition}[{Process network and process species}]\label{ARDefinition: Biochemical process networks}
Let $n \in \mathbb{N}$. A \textit{process network} $\mc{B}= (E, B, \lbr l_i \rbr_{n} )$ on a finite set $E$ of \textit{entities} \textit{with $n$ legs} consists of the following: 
	\begin{itemize}
 \item A finite set $B$, whose elements $b$ are called \textit{process species}.
		\item A finite set of functions $\mc{L}:= \lbrace l_i \colon B \ra \mathbb{B}[E] \rbrace_{i \in \lbr 1,2, \ldots, n \rbr}$, where $\mathbb{B}[E]$ is the free commutative monoid on the set $E$ with the coefficients from the \textit{additive Boolean monoid} $\mathbb{B}= \big(\lbr 0,1 \rbr, + \big)$, whose multiplication table is given as follows:

\vskip 0.5em
\renewcommand\arraystretch{1.3}
\setlength\doublerulesep{0pt}
\begin{center}
\begin{tabular}{c||c|c}
+ & 0 & 1 \\
\hline\hline 
0 & 0 & 1 \\ 
\hline
1 & 1 & 1 \\ 
\end{tabular}
\end{center}
	\end{itemize}
  We call $l_i(b)$ as the \textit{$i$-th leg of the process species $b$}. 
\end{definition}

\begin{definition}[{Evaluation function of a process species}]\label{Defn:Evaluation-function-of-a-process-species}
Given a process network $(E, B, \lbr l_i \rbr_{n} )$, for any process species $b \in B$, we will call the function $b_{{\rm{legs}}} \colon \mc{L} \ra \mathbb{B}[E], l_i \mapsto l_i(b)$, the \textit{evaluation function of the process species $b$}, where $\mc{L}= \lbr l_1, l_2, \ldots, l_n \rbr$.
\end{definition}
We say that the \textit{$i$-th leg of the process species $b$ is missing} if $b_{{\rm{legs}}}(l_i)=0$. 

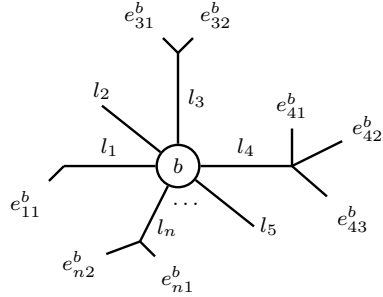
\begin{figure}[htb]
\begin{center}
\begin{tikzpicture}[font={\footnotesize}, line width = 0.9pt]
%
%
\node (b) at (3.5,3) [circle, draw] {$b$};
\coordinate (l1) at (2,3);
\node (e11) at (1.5,2.5) {$e^{b}_{11}$};
\draw[-] (e11) -- (l1);
\draw[-] (l1) -- (b) node[pos=0.5, above] {$l_1$};
\coordinate (l2) at (2.5,3.8);
\draw[-] (b) -- (l2) node[xshift=0mm, yshift=2mm] {$l_2$};
\coordinate (l3) at (3.5,4.5);
\draw[-] (l3) -- (b) node[pos=0.5, right] {$l_3$};
\node (e31) at (3,5) {$e^{b}_{31}$};
\node (e32) at (4,5) {$e^{b}_{32}$};
\draw[-] (e31) -- (l3);
\draw[-] (e32) -- (l3);
\coordinate (l4) at (5,3);
\draw[-] (l4) -- (b) node[pos=0.5, above] {$l_4$};
\node (e41) at (5,3.8) {$e^{b}_{41}$};
\node (e42) at (6,3.5) {$e^{b}_{42}$};
\node (e43) at (5.8,2.3) {$e^{b}_{43}$};
\draw[-] (l4) -- (e41);
\draw[-] (l4) -- (e42);
\draw[-] (l4) -- (e43);
\coordinate (l5) at (4.5,2.2);
\draw[-] (b) -- (l5) node[xshift=2mm] {$l_5$};
\node (dots) at (3.6,2.5) {$\ldots$};
\coordinate (ln) at (3,2);
\draw[-] (b) -- (ln) node[pos=0.75, right] {$l_n$};
\node (en1) at (3.5,1.5) {$e^{b}_{n1}$};
\draw[-] (ln) -- (en1);
\node (en2) at (2.2,1.7) {$e^{b}_{n2}$};
\draw[-] (ln) -- (en2);
\end{tikzpicture}
\end{center}
\caption{Evaluation function $b_{{\rm{legs}}} \colon \mc{L} \ra \mathbb{B}[E]$ for a process species $b.$ Observe that we draw the legs $l_2,l_5, l_6, \ldots, l_{n-1},$ which are evaluated zero at $b$, but we do not attach entities to them.}
\label{fig:notation-process-species}
\end{figure}

\begin{remark}\label{ARRemark: Relevant entities}
Consider a process network $\mc{B}=(E, B, \lbr l_i \rbr_{n} )$. Observe that for any $b \in B$ and $i \in \lbr 1,2, \ldots, n \rbr$ such that  $l_i(b) \neq 0$, there is a unique way to write $l_i(b)= \Sigma^{j=m^{i}_b}_{j=1}e^{b}_{ij}$ with distinct non-zero summands, where $m^{i}_b \in \mb{N}$ and $e^{b}_{ij} \in E$. Let us denote the set $\lbr e^{b}_{ij}\colon l_i(b)= \Sigma^{j=m^{i}_b}_{j=1}e^{b}_{ij}  \rbr$ as $\bar{e}_{i,b}$. Note that $|\bar{e}_{i,b}|= m^{i}_b$, where $|\bar{e}_{i,b}|$ denotes the cardinality of the set $\bar{e}_{i,b}$.   When $l_i(b)=0$, we define  $\bar{e}_{i,b}$ as the empty set $\emptyset$. 
\end{remark}
A general element $x$ in $\mathbb{B}[E]$ is a formal linear combination of the form $x= \alpha_1e_1 + \alpha_2e_2 + \cdots + \alpha_ne_n$, where $\alpha_i \in \mathbb{B}$ and $e_i \in E$ for all $i=1,2, \ldots, n$ and $n \in \mb{N}$ .
Thus, we define the evaluation function $b_{{\rm{legs}}}$ for Figure \ref{fig:notation-process-species} as
\begin{equation}\label{eq:notation-process-species}
\begin{split}
&b_{{\rm{legs}}}(l_1)= 1 \cdot e^{b}_{11} + \big(\sum_{e \in \big(E-\lbr e^{b}_{11}\rbr\big)} 0 \cdot e \big)= e^{b}_{11} \\
&b_{{\rm{legs}}}(l_2)= \big(\sum_{e \in E} 0 \cdot e \big)=0 \\
&b_{{\rm{legs}}}(l_3) = 1 \cdot e^{b}_{31} + 1 \cdot e^{b}_{32} +  \big(\sum_{e \in \big( E-\lbr e^{b}_{31}, e^{b}_{32}\rbr \big)} 0 \cdot e \big)= e^{b}_{31} + e^{b}_{32}\\
& b_{{\rm{legs}}}(l_4)= 1 \cdot e^{b}_{41} + 1 \cdot e^{b}_{42} +  1 \cdot e^{b}_{43} + \big(\sum_{e \in \big(E-\lbr e^{b}_{41}, e^{b}_{42}, e^{b}_{43} \rbr \big)} 0 \cdot e \big)= e^{b}_{41} + e^{b}_{42} + e^{b}_{43}\\
& b_{{\rm{legs}}}(l_k)= \big(\sum_{e \in E} 0 \cdot e \big)=0\,\, \text{for all}\,\, k= 5, 6, \ldots, n-1\\
& b_{{\rm{legs}}}(l_n)= 1 \cdot e^{b}_{n1} + 1 \cdot e^{b}_{n2} +  \big(\sum_{e \in \big(E-\lbr e^{b}_{n1}, e^{b}_{n2}\rbr \big)} 0 \cdot e \big)= e^{b}_{n1} + e^{b}_{n2}
\end{split}
\end{equation}
In \Cref{eq:notation-process-species}, the coefficient $1$ and $0$ expresses, respectively, the presence and absence of the entities. For example, the expression $b_{{\rm{legs}}}(l_3)= e^{b}_{31} + e^{b}_{32}$ says that the only entities that are associated to the leg $l_3$ of the process species $b$ are $e^{b}_{31}$ and $e^{b}_{32}$. Our definition of process network is motivated by the definition of a Petri net as in \cite{MR4085076}. However, the above definition is tailored to our purpose. There are two main differences in our definition. First, instead of just a pair of maps (source and target), we have room for more maps $\lbr l_i\rbr_{n}$ to consider modulators like stimulation, catalysis, inhibition, necessary stimulation, etc. Secondly, the maps $\lbr l_i\rbr_{n}$ are valued in $\mathbb{B}[E]$ instead of $\mb{N}[E]$. This emphasises the fact that often one cannot quantify the exact number of molecules involved in a reaction. In that case, one only measures the presence or absence, and assumes that the presence means presence in abundance. To take account of it, in SBGN visualisation of a biochemical reaction network, precise stoichiometry is absent. With this motivation, in our framework, we consider only the presence or absence of entities given by the coefficients $1$ and $0$, respectively, coming from the additive Boolean monoid $\mathbb{B}$.


\subsection{Examples of process networks}\label{subsection:Example-of-process-networks}
In Figure \ref{fig:process-species}, we consider the SBGN-PD visualisation of the biochemical reaction, where a molecule MEK-P modulates the activation of ERK into ERK-P. Using notations from Definition \ref{ARDefinition: Biochemical process networks}, we construct an associated process network $\mc{B}=(E, B, \lbr l_i \rbr_{3} )$ as follows:
\begin{itemize}
\item $B := \lbr b \rbr$, the singleton set containing the process species.

\item $\mc{L}$ has three elements $l_1, l_2, l_3$, representing consumption arc, modulation arc and production arc, respectively, in Figure \ref{fig:process-species}.

\item $E:=\lbr e^{b}_{11},e^{b}_{12}, e^{b}_{21}, e^{b}_{31}, e^{b}_{32}  \rbr$, where 
\begin{itemize}
\item $e^{b}_{11}=(\text{ATP, simple chemical})$,
\item $e^{b}_{12}=(\text{ERK, macromolecule})$,
\item $e^{b}_{21}=(\text{MEK-P, macromolecule})$,
\item  $e^{b}_{31}=(\text{ERK-P, macromolecule})$,
\item $e^{b}_{32}=(\text{ADP, simple chemical}).$
\end{itemize}
 If we represent the \textit{set of  molecules} $\lbr\text{MEK-P, ERK, ERK-P, ATP, ADP}\rbr$ by $M$ and the \textit{set of types of molecules involved} $\lbr\text{macromolecule, simple chemical}\rbr$ by $T$, then the set $E= M \times T$.
 
\item Legs $l_1,l_2,l_3 \colon B \ra \mathbb{B}[E]$ are defined as $b \mapsto e^{b}_{11} + e^{b}_{12}$, $b \mapsto e^{b}_{21}$ and $b \mapsto e^{b}_{31} + e^{b}_{32}$, respectively.
\end{itemize}
Hence, $\mc{B}=(E, B, \lbr l_i \rbr_{3} )$ defines a process network on the set $E$ with three legs. Comparing with Figure \ref{fig:notation-process-species}, observe that the bottom right diagram in Figure \ref{fig:process-species} describes the evaluation function $b_{{\rm{legs}}} \colon \mc{L} \ra \mathbb{B}[E]$ of the process species $b$.

In Figure \ref{Figure:Processnetworkwithtwoprocessspecies},  we demonstrate an example of a process network with two process species.

\begin{figure}[htb]
\begin{center}
\includegraphics[width=10cm]{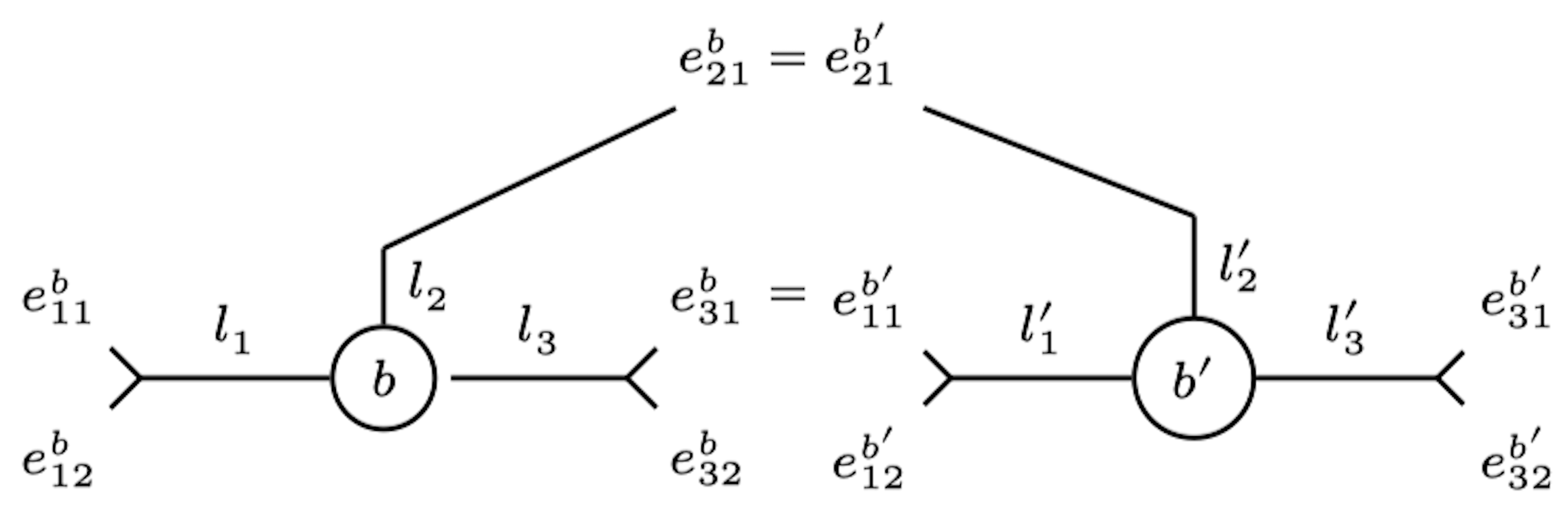}
\caption{A process network with two process species}
\label{Figure:Processnetworkwithtwoprocessspecies}
\end{center}
\end{figure}

\subsection{Transforming a process network into another: Morphisms of process networks}\label{Subsection:Transforming a process network into another: Morphisms of process networks}
Here, we introduce the notion of a morphism of process networks and illustrate how it models the process of transforming a biochemical reaction network into another.

\begin{definition}[Morphism of process networks]\label{Defn:morphism-of-process-networks}
For a fixed $n \in \mb{N}$, a \textit{morphism from a process network $\mc{B}=(E, B, \lbr l_i \rbr_{n})$ to another process network $\mc{B}'=(E', B', \lbr l_i' \rbr_{n})$} is given by a pair of functions $F:=(\alpha \colon E \ra E', \beta \colon B \ra B' )$, such that following diagram commutes for all $i \in \lbr 1,2,3, \ldots, n \rbr$,	where $\mathbb{B} {\rm{[}}\alpha {\rm{]}} \colon  \mathbb{B}[E] \ra  \mathbb{B}[E'] $ is the unique monoid homomorphism extending $\alpha$.
          \[
			\begin{tikzcd}[sep=small]
				B \arrow[rr, "l_i"] \arrow[dd,xshift=0.75ex,"\beta"]&  & \mathbb{B} [E] \arrow[dd,xshift=0.75ex,"\mathbb{B} {\rm{[}}\alpha {\rm{]}}"] \\
				&  &                \\
				B' \arrow[rr, "l'_i"]            &  & \mathbb{B}[E']   
			\end{tikzcd}\]
\end{definition}
Next, we show how the evaluation functions of process species (Definition \ref{Defn:Evaluation-function-of-a-process-species}) behaves with the morphisms of process networks. We will see how their behaviour provides us with formal criteria for distinguishing between reactions with the modulator's influence and the ones without such influences.  

\begin{lemma}\label{Lemma:behaviour-of-evaluation-functions-of-a-process-network}
Let $\mc{B}=(E, B, \lbr l_i \rbr_{n})$ and $\mc{B}'= (E', B', \lbr l'_i \rbr_{n})$ be a pair of process networks. Then, for any morphism $F= (\alpha, \beta) \colon \mc{B} \ra \mc{B}'$ the following hold:
\begin{itemize}
\item[(a)] if $b_{\rm{legs}}(l_i)=0$, then $\beta(b)_{{\rm{legs}}}(l'_i)=0$,
\item[(b)] if  $b_{\rm{legs}}(l_i) \neq 0$, then $\beta(b)_{{\rm{legs}}}(l'_i) \neq 0$,
\item[(c)] if $b_{\rm{legs}}(l_i) \neq 0$, then using the notations introduced in Remark \ref{ARRemark: Relevant entities}, we have  $|\bar{e}_{i,b}| \geq|\bar{e}_{i, \beta(b)}|$,
\end{itemize}
for each $b \in B$ and $i \in \lbr 1,2, \ldots, n \rbr$, where $b_{{\rm{legs}}}$ and $\beta(b)_{{\rm{legs}}}$ are the evaluation functions (see Definition \ref{Defn:Evaluation-function-of-a-process-species}) of the process species $b$ and $\beta(b)$ respectively.
\end{lemma}
\begin{proof}
The proof follows directly from the commutative diagram in the Definition \ref{Defn:morphism-of-process-networks}.
\end{proof}
Now, let us consider a process network $\mc{B}=(E, \lbr b\rbr, \lbr l_i \rbr_{3})$ with 3 legs, where 
\begin{itemize}
\item $l_1$ denotes the \textit{input leg}, i.e.  $l_1(b) \in \mathbb{B}[E]$ contains the information of the substrates that goes into the process species (reaction) $b$,
\item $l_2$ denotes a \textit{modulation leg} (eg. activation or  inhibition), i.e. $l_2(b) \in \mathbb{B}[E]$ contains the information of the biomolecules  that modulate (activate or inhibit) the reaction $b$,
\item $l_3$ denotes the \textit{production leg}, i.e.  $l_3(b) \in \mathbb{B}[E]$ contains the information of the metabolites produced in the reaction $b$.
\end{itemize}
 Now, let us assume $l_{2}(b)=0$. This means that the reaction $b$ occurs without any modulator's influence. Now, consider another process network $\mc{B}'=(E', \lbr b' \rbr, \lbr l'_i \rbr_{3})$ with 3 legs such that $l'_{2}(b') \neq 0$, that is the reaction $b'$ occurs with modulators' influence. Then, by the condition (a) and (b) in Lemma \ref{Lemma:behaviour-of-evaluation-functions-of-a-process-network} respectively, it is obvious we can not expect a morphism  $\mc{B}$ to $\mc{B}'$ and from $\mc{B}'$ to $\mc{B}$. See Figure \ref{Figure:Existence-of-morphisms-of-process-networks}.


\begin{figure}[htb]
	\begin{center}
	\resizebox{15cm}{!}{%
	\begin{tikzpicture}[font={\footnotesize}, line width = 0.9pt]
		%
		\node at (2.5,0.5) [] {Process network $\mathcal{B}$};
		\node (b2l) at (2.5,2) [circle, draw] {$b_2$};
		\coordinate (l1l) at (1,2);
		\node (e11) at (0.5,1.5) {$e_{5}$};
		\draw[-] (e11) -- (l1l);
		\draw[-] (l1l) -- (b2l) node[pos=0.5, above] {$l_1$};
		\coordinate (l2l) at (2.5,2.8);
		\draw[-] (b2l) -- (l2l) node[above] {$l_2$};
		\coordinate (l3l) at (4,2);
		\draw[-] (l3l) -- (b2l) node[pos=0.5, above] {$l_3$};
		\node (e6l) at (4.5,1.5) {$e_{6}$};
		\draw[-] (e6l) -- (l3l);
		\node (b1l) at (2.5,4) [circle, draw] {$b_1$};
		\coordinate (l1lu) at (1,4);
		\node (e4l) at (0.5,3.5) {$e_{4}$};
		\draw[-] (e4l) -- (l1lu);
		\node (e3l) at (0.5,4.5) {$e_{3}$};
		\draw[-] (e3l) -- (l1lu);
		\draw[-] (l1lu) -- (b1l) node[pos=0.5, above] {$l_1$};
		\coordinate (l3lu) at (2.5,4.8);
		\draw[-] (b1l) -- (l3lu) node[above] {$l_2$};
		\coordinate (l3lu) at (4,4);
		\draw[-] (l3lu) -- (b1l) node[pos=0.5, above] {$l_3$};
		\node (e1l) at (4.5,4.5) {$e_{1}$};
		\draw[-] (e1l) -- (l3lu);
		\node (e2l) at (4.5,3) {$e_{2}$};
		\draw[-] (e2l) -- (l3lu);
		\draw[-] (e2l) -- (l3l);
		%
		\node at (8,0.5) [] {Process network $\mathcal{B}'$};
		\node (bp) at (8,3) [circle, draw] {$b'$};
		\coordinate (l1m) at (6.5,3);
		\node (e2m) at (6,2.5) {$e'_{2}$};
		\draw[-] (e2m) -- (l1m);
		\draw[-] (l1m) -- (bp) node[pos=0.5, above] {$l'_1$};
		\node (e3m) at (6,3.5) {$e'_{1}$};
		\draw[-] (e3m) -- (l1m);
		\coordinate (l2m) at (8,4);
		\draw[-] (bp) -- (l2m) node[pos=0.5,right] {$l'_2$};
		\node (e4m) at (8.5,4.5) {$e'_4$};
		\draw[-] (e4m) -- (l2m);
		\coordinate (l3m) at (9.5,3);
		\draw[-] (bp) -- (l3m) node[pos=0.5,above] {$l'_3$};
		\node (e3m) at (10,3.5) {$e'_3$};
		\draw[-] (e3m) -- (l3m);
		%
		\node at (13.5,0.5) [] {Process network $\mathcal{\bar{\mathcal{B}}}$};
		\node (bb) at (13.5,3) [circle, draw] {$\bar{b}$};
		\coordinate (l1b) at (12,3);
		\draw[-] (l1b) -- (bb) node[pos=0.5, above] {$\bar{l}_1$};
		\node (e3m) at (11.5,2.5) {$\bar{e}_1$};
		\draw[-] (e3m) -- (l1b);
		\coordinate (l2b) at (13.5,4);
		\draw[-] (bb) -- (l2b) node[above] {$\bar{l}_2$};
		\coordinate (l3b) at (15,3);
		\draw[-] (bb) -- (l3b) node[pos=0.5, above] {$\bar{l}_3$};
		\node (e2b) at (15.5,3.5) {$\bar{e}_2$};
		\draw[-] (e2b) -- (l3b);
	\end{tikzpicture}
	} 
	\end{center}
	\vspace{-0.37cm}
	\caption{An illustration of condition (a) and condition (b) in Lemma \ref{Lemma:behaviour-of-evaluation-functions-of-a-process-network}. Due to (a), there cannot exist a morphism of process networks from $\mathcal{B}$ to $\mathcal{B}'$ and from $\bar{\mathcal{B}}$ to $\mathcal{B}'$, whereas (b) ensures that there cannot exist any morphism from $\mathcal{B}'$ to $\mathcal{B}$ and from $\mathcal{B}'$ to $\bar{\mathcal{B}}$.}
	\label{Figure:Existence-of-morphisms-of-process-networks}
\end{figure}
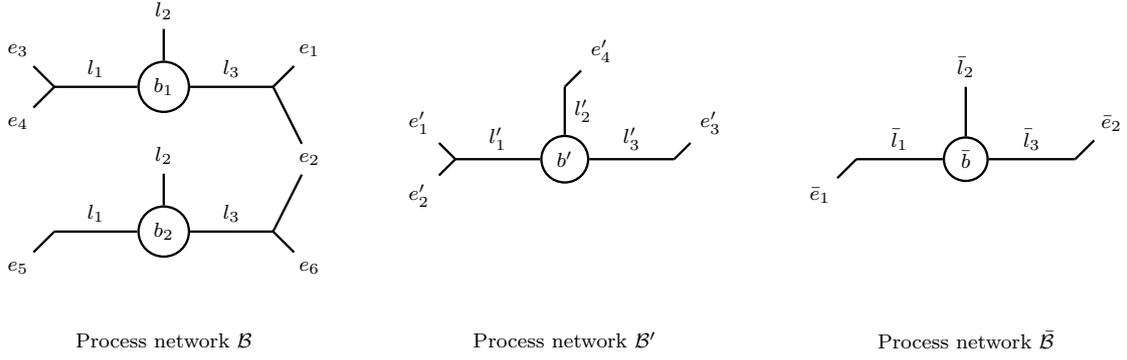

\begin{example}\label{Example:Existence-of-morphisms-of-process-networks}
Let us consider process networks $\mc{B} =(E, B, \lbr l_i \rbr_{3})$, $\bar{\mc{B}} =(\bar{E}, \bar{B}, \lbr \bar{l_i}  \rbr_{_{3}}), \mc{B}' =(  E', B', \lbr l'_i \rbr_{3})$ and  $\bar{\mc{B}'} =(\bar{E'}, \bar{B'}, \lbr \bar{l'_i} \rbr_{3})$, where $E:= \lbr e_1,e_2,e_3,e_4, e_5,e_6 \rbr$, $\bar{E}:= \lbr e_2, e_5, e_6 \rbr$, $E':= \lbr e_1', e_2', e, e_3', e_4' \rbr$, $\bar{E'}:= \lbr e_1', e, e_3' \rbr$, $B:= \lbr b_1, b_2 \rbr$, $\bar{B}:= \lbr b_2 \rbr$, $B':= \lbr b' \rbr$, $\bar{B'}:= \lbr b' \rbr$ (see Figure \ref{fig:morphism-of-process-networks}).
Define $F:=(\alpha \colon E \ra \bar{E}, \beta \colon B \ra \bar{B})$, where $\alpha \colon E \ra \bar{E}$ is given by $e_1 \mapsto e_6, e_2 \mapsto e_2, e_3 \mapsto e_5, e_4 \mapsto e_5, e_5 \mapsto e_5, e_6 \mapsto e_6$ and $\beta \colon B \ra \bar{B}$ is given as $b_1 \mapsto b_2, b_2 \mapsto b_2$.  Now, define $\bar{F}:=(\bar{\alpha} \colon \bar{E} \ra E, \bar{\beta} \colon \bar{B} \ra B)$, where $\bar{\alpha} \colon \bar{E} \ra E$ is given as $e_2 \mapsto e_2, e_6 \mapsto e_6,  e_5 \mapsto e_5$ and $\bar{\beta} \colon  \bar{B} \ra B$ is given by $b_2 \mapsto b_2$. Defining $F':=(\alpha' \colon E' \ra \bar{E'}, \beta' \colon B' \ra \bar{B'})$, where $\alpha' \colon  E' \ra \bar{E'}$ is defined as $e_1' \mapsto e_1', e_2' \mapsto e_1', e_3' \mapsto e_3', e_4' \mapsto e_3', e \mapsto e$ and $\beta' \colon  B' \ra \bar{B'}$ is defined by $b' \mapsto b'$. From the evaluation functions shown in Figure \ref{fig:morphism-of-process-networks}, it is straightforward to verify that $F, \bar{F}$ and $F'$ are morphisms of process networks. Observe in Figure \ref{fig:morphism-of-process-networks} that one can interpret morphisms $F \colon \mc{B} \ra \bar{\mc{B}}$ and $F' \colon \mc{B}' \ra \bar{\mc{B}'}$ as ways of zooming-out details, and the morphism $\bar{F} \colon \bar{\mc{B}} \ra \mc{B}$ as a way to zoom-in details in biochemical reaction networks. An SBGN-PD visualisation of the zooming-out procedure is modelled at the bottom of the Figure \ref{fig:morphism-of-process-networks}, where $F'$ models the process of forgetting ADP and ATP in the reaction describing the activation of the ERK molecule through phosphorylation.  In this regard, it is worth mentioning that the roles of ATP and ADP in biochemical reactions are often regarded as standard and, for the sake of clarity, may therefore be omitted from schematic representations of biochemical processes. In particular, the reaction shown in the bottom right of Figure~\ref{fig:morphism-of-process-networks} represents the phosphorylation of ERK, with the conventional involvement of ATP and ADP implicitly assumed.
Also, one can check that as a consequence of the condition (c) in the Lemma \ref{Lemma:behaviour-of-evaluation-functions-of-a-process-network}, there can not exist a morphism of process network from $\bar{\mc{B}'}$ to $\mc{B}'$.
\end{example}

\begin{figure}[H] 
\begin{center}
\includegraphics[width=12cm]{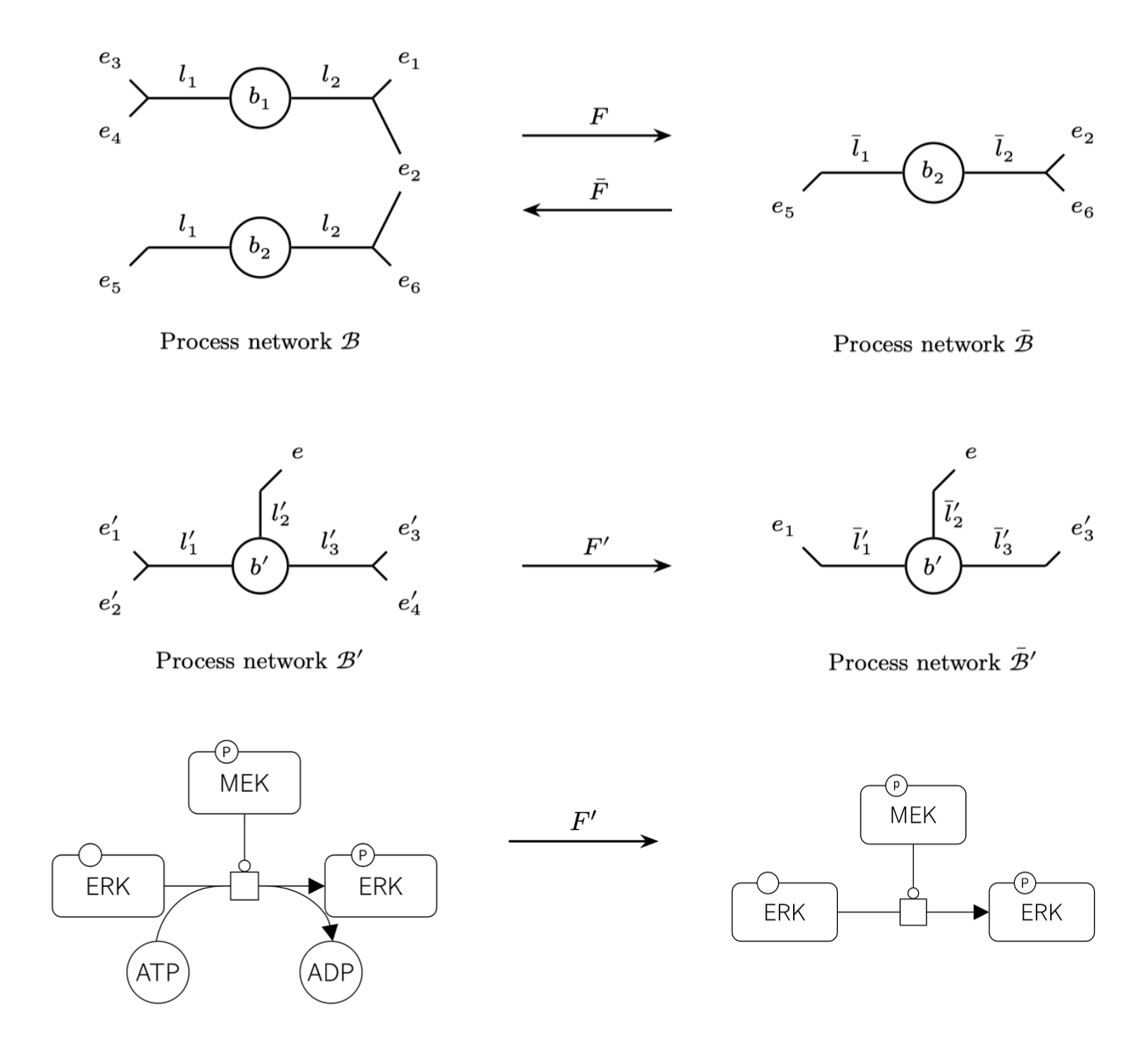}
\caption{An illustration of Example \ref{Example:Existence-of-morphisms-of-process-networks}, showing how to interpret morphisms of process networks as ways of `zooming-out and zooming-in details' in biochemical reaction networks.} 
\label{fig:morphism-of-process-networks}
\end{center}
\end{figure}

\section{An organising principle for process networks}\label{Section4}
 This section develops a mathematical theory that enables us to
\begin{itemize}
\item[(a)] introduce a notion of a \textit{process network with an interface},

\item[(b)] derive a formal way to combine process networks with interfaces into a composite process network with an interface,

 \item[(c)] derive a formal way to transform a process network with an interface to another process network with an interface,
 
  \item[(d)] find a formal way to study the \textit{influence of a process subnetwork on the remaining portion of the network} and vice versa.
    \end{itemize}

    By an \textbf{organizational principle for process networks}, we mean (a), (b), (c) and (d), and their interrelationships. To formalise (a), (b) and (c), first we will construct a finitely cocomplete category whose objects are process networks (Definition \ref{ARDefinition: Biochemical process networks}) and morphisms are morphisms of process networks (Definition \ref{Defn:morphism-of-process-networks}) and then, we will build a symmetric monoidal double category whose horizontal 1-morphisms are structured cospans in the category of process networks. However, (d) will be formalised using a combinatorial argument.

    \subsection{Category of process networks}\label{subsection:category-of-process-networks}
    For a fixed number of legs, the collection of process networks and their morphisms forms a category as we see next.

\begin{proposition}\label{Prop:category-of-process-networks}
		For any $n \in \mb{N}$, the collection of process networks forms a category \textbf{Process}$_{n}$ whose
		\begin{itemize}
			\item objects are process networks $\mc{B}=(E, B, \lbr l_i \rbr_{n} )$;
			\item a morphism from a process network $\mc{B}=(E, B, \lbr l_i \rbr_{n} )$ to another process network $\mc{B}'=(E', B', \lbr l_i' \rbr_{n} )$ is given by a pair of functions $F:=(\alpha \colon E \ra E', \beta \colon B \ra B' )$, such that 
          \[
			\begin{tikzcd}[sep=small]
				B \arrow[rr, "l_i"] \arrow[dd,xshift=0.75ex,"\beta"]&  & \mathbb{B} [E] \arrow[dd,xshift=0.75ex,"\mathbb{B} {\rm{[}}\alpha {\rm{]}}"] \\
				&  &                \\
				\mc{B}' \arrow[rr, "l'_i"]            &  & \mathbb{B}[E']    
			\end{tikzcd}\]
   commutes for all $i \in \lbr 1,2,3, \ldots, n \rbr$.
		\end{itemize}	
\end{proposition}	

\begin{proof}
Let us consider a pair of morphisms $F \colon \mc{B} \ra \mc{B}'$ and $G \colon \mc{B}' \ra \mc{B}''$, defined as $F:=(\alpha \colon E \ra E', \beta \colon B \ra B')$ and $G:=(\alpha' \colon E' \ra E'', \beta \colon B' \ra B'' )$ where, $\mc{B} =(E, B, \lbr l_i \rbr_{n} )$, $\mc{B}' =(E', B', \lbr l_i' \rbr_{n} )$ and $\mc{B}'' =(E'', B'', \lbr l_i'' \rbr_{n} )$.  We define $G \circ F:= \big( \alpha' \circ \alpha, \beta' \circ \beta \big)$. To see if the definition of $G \circ F$  makes sense, consider the  diagrams
\[
	\begin{tikzcd}
		B \arrow[d, "\beta"'] \arrow[r, "l_i"]  & \mathbb{B} [E] \arrow[d, "\mathbb{B} {\rm{[}}\alpha {\rm{]}}"] \\
		B' \arrow[r, "l_i'"'] \arrow[d, "\beta'"'] & \mathbb{B} [E'] \arrow[d, "\mathbb{B} {\rm{[}}\alpha' {\rm{]}}"] \\
		B'' \arrow[r, "l_i''"']                 & \mathbb{B} [E'']               
	\end{tikzcd}
 \]
 that commutes for all $i= \lbr 1,2, 
 \ldots, n\rbr$. Observe that from the commutativity of small squares, we have $l_i'' \circ (\beta' \circ \beta)= \mathbb{B} {\rm{[}}\alpha' {\rm{]}} \circ \mathbb{B} {\rm{[}}\alpha {\rm{]}} \circ l_i$. Then, the observation $\mathbb{B} {\rm{[}}\alpha' \circ \alpha{\rm{]}}= \mathbb{B} {\rm{[}}\alpha' {\rm{]}} \circ \mathbb{B} {\rm{[}}\alpha {\rm{]}}$ concludes $G \circ F$ is well defined. For each object $\mc{B} =(E, B, \lbr l_i \rbr_{n} )$, it is easy to see that the identity morphism  is given by ${\rm{id}}_{\mc{B}} =({\rm{id}} \colon E \ra E, {\rm{id}} \colon B \ra B ).$ Finally, the associativity of composition follows from the associativity of the composition of functions. Hence, we proved \textbf{Process}$_n$ is a category.
\end{proof}

    Now, since we will be representing process networks with interfaces as cospans in \textbf{Process}$_{n}$ (Proposition \ref{Prop:category-of-process-networks}), for composing cospans, we need the category \textbf{Process}$_{n}$ to have all finite pushouts. In fact, we will show \textbf{Process}$_{n}$ contains all finite colimits. For this, we proceed in two steps:
\begin{itemize}
\item [Step 1:] We will show \textbf{Process}$_{n}$ is equivalent to the comma category $\rm{Id}/ \mathbb{B}[-]^{n}$, where $\rm{Id} \colon$ \textbf{Set} $\ra$ \textbf{Set} is the identity functor and $\mathbb{B}[-]^{n} \colon {\rm{\textbf{Set}}} \ra {\rm{\textbf{Set}}}$ is the functor that takes a set $E$ to the underlying set of the commutative monoid $ \mathbb{B}^{n}[E]:= \underbrace{\mathbb{B}[E] \times \mathbb{B}[E] \times \cdots  \times \mathbb{B}[E]}_{n-times}$, and takes a function $\alpha \colon  E \ra E'$  to the function  $ \mathbb{B}^{n}[\alpha]  \colon \mathbb{B}^{n}[E] \ra \mathbb{B}^{n}[E'].$ In more concrete terms, objects of the category $\rm{Id}/ \mathbb{B}[-]^{n}$ are triples $(D, h, E)$, where $h \colon D \to \mathbb{B}^{n}[E]$ is a function from the set $D$ to the underlying set of $\mathbb{B}^{n}[E]$, and a morphism from $(D, h, E)$ to $(D', h', E')$ is given by a pair of functions $\phi \colon D \to D'$ and $\psi \colon E \to E'$ such that the following diagram commutes:
\[
\begin{tikzcd}
D \arrow[d, "h"'] \arrow[r, "\phi"] & D'  \arrow[d, "h'"] \\
\mathbb{B}^{n}[E] \arrow[r, "\mathbb{B}^{n}{\rm{[}}\psi {\rm{]}}"']                & \mathbb{B}^{n}[E']               
\end{tikzcd}
\]

\item[Step 2:] We will show $\rm{Id}/\mathbb{B}[-]^{n}$ has all finite colimits.
\end{itemize}

\begin{lemma}\label{ARLemma Bio coma}
	\textbf{Process}$_{n}$ is equivalent to the comma category $\rm{Id}/\mathbb{B}[-]^{n}$.
\end{lemma}

\begin{proof}
First we show that the following functor is well defined:
\begin{equation}\label{AREquation: definition of equivalence}
		\begin{split}
			F \colon  {\rm{\textbf{Process}}}_{n} & \ra {\rm{Id}}/\mathbb{B}[-]^{n}\\
			(E, B, \lbr l_i \rbr_{n} ) & \mapsto (B, \prod_{i=1}^{i=n} l_i, E)\\
			 \Big((\alpha, \beta) \colon (E, B, \lbr l_i \rbr_{n} ) \ra (E', B', \lbr l_i' \rbr_{n} ) \Big) & \mapsto \big( \beta \colon B \ra B', \alpha \colon E \ra E'\big).  
		\end{split}
\end{equation}
Let us denote the maps induced by $F$ on objects and morphism using the notation $F_0$ and $F_1$, respectively. From the definition of the function 
$$
\prod_{i=1}^{i=n} l_i \colon B \ra \underbrace{\mathbb{B}[E] \times \mathbb{B}[E]\times \cdots  \times \mathbb{B}[E]}_{n-times}, b \mapsto \big(l_1(b), l_2(b), \ldots, l_n(b) \big), 
$$ 
it follows $F_0$ is well-defined. To see $F_1$ is well-defined, we need to show that 
$$
(\beta,\alpha) \in \hom_{{\rm{Id}}/\mathbb{B}[-]^{n}} \Big((B, \prod_{i=1}^{i=n} l_i, E), (B', \prod_{i=1}^{i=n} l'_i, E') \Big).
$$ 
Now, observe that for our purpose it is sufficient to prove the commutativity of the following diagram
\[
	\begin{tikzcd}
		B \arrow[d, "\beta"'] \arrow[r, "\prod_{i=1}^{i=n} l_i"] & \mathbb{B}[E]^{n} \arrow[d, "\mathbb{B}{\rm{[}}- {\rm{]}}^{n}(\alpha)"] \\
		B' \arrow[r, "\prod_{i=1}^{i=n} l_i'"']                & \mathbb{B}[E']^{n}            
	\end{tikzcd}
\]
for all $i \in \lbr 1,2, \ldots, n \rbr$.	Then, for each $i$, the commutativity of the following diagrams
\[
	\begin{tikzcd}[sep=small]
		B \arrow[rr, "l_i"] \arrow[dd,xshift=0.75ex,"\beta"]&  & \mathbb{B} [E] \arrow[dd,xshift=0.75ex,"\mathbb{B} {\rm{[}}\alpha {\rm{]}}"] \\
		&  &                \\
		B' \arrow[rr, "l_i'"]            &  & \mathbb{B}[E']    
	\end{tikzcd}
 \]
	implies
\begin{equation}\label{AREquation: well defined equivalence}
		\prod_{i=1}^{i=n}(\mathbb{B}[\alpha] \circ l_i)= \prod_{i=1}^{i=n} (l_i' \circ \beta) 
\end{equation}
It is easy to see that the left hand side of \ref{AREquation: well defined equivalence} is same as $\mathbb{B}{\rm{[}}- {\rm{]}}^{n}(\alpha) \circ \prod_{i=1}^{i=n} l_i$ and the right hand side of \ref{AREquation: well defined equivalence} is same as $\prod_{i=1}^{i=n} l_i' \circ \beta$. Hence, $F_1$ is well-defined. Observe that from the definition of $F$ itself (\ref{AREquation: definition of equivalence}), it is clear that $F$ is a faithful functor. The fullness follows directly from the pointwise definition of \ref{AREquation: well defined equivalence}. To show that $F$ is essentially surjective, consider an object
$(B, f, E)$ in  ${\rm{Id}}/\mb{Z}_{2}[-]^{n}$. Since $s(f)= B$ and $t(f)=  \mathbb{B}[E] \times \mathbb{B}[E]\times \cdots  \times \mathbb{B}[E]$, it is clear that $f=\prod_{i=1}^{i=n} l_i$ for some functions $l_1, l_2, \ldots, l_n \colon B \ra  \mathbb{B}[E]$. Observe that  $(E, B, \lbr l_i \rbr_{n} )$ is an object of \textbf{Process}$_{n}$ by the pointwise definition in \ref{AREquation: well defined equivalence}. Also, it is clear from the definition of $F$ (\ref{AREquation: definition of equivalence}), that $F \big((E, B, \lbr l_i \rbr_{n} ) \big)= (B, f, E)$. Thus, we have shown $F$ is surjective, and hence, essentially surjective. So, we proved that \textbf{Process}$_{n}$ is equivalent to the category ${\rm{Id}}/\mathbb{B}[-]^{n}$.
\end{proof}

\begin{proposition}\label{Proposition: Bio is finitely cocomplete}
	\textbf{Process}$_{n}$ contains all finite colimits.
\end{proposition}

\begin{proof}
	It follows from the \textit{Theorem 3, Section 5.2 of} \cite{MR999925} that whenever the categories $A$ and $B$ have finite colimits, $F \colon A  \ra C$ is a functor preserving such colimits and $G \colon B \ra C$	is any functor, then the comma category $F/G$ has finite colimits. Hence, the category ${\rm{Id}}/\mathbb{B}[-]^{n}$ has all finite colimits.  Since by Lemma \ref{ARLemma Bio coma}, \textbf{Process}$_{n}$ is equivalent to  ${\rm{Id}}/\mathbb{B}[-]^{n}$,   \textbf{Process}$_{n}$ contains also all finite colimits.
\end{proof}

The following easily verifiable lemma demonstrates the  coproducts in \textbf{Process}$_{n}$.

\begin{lemma}\label{ARLemma:coproduct Biochemical process species}
	Consider two process networks $\mc{B}_1 =(E_1, B_1, \lbr l_{i_{1}} \rbr_{n} ) $ and  $\mc{B}_2 =(E_2, B_2, \lbr l_{i_{2}} \rbr_{n} )$ in \textbf{Process}$_{n}$. Then, the coproduct of $\mc{B}_1$ and $\mc{B}_2$ is defined upto an isomorphism as a process network $\mc{B}_1 + \mc{B}_2= (E, B, \lbr l_i \rbr_{n} )$  given as follows:
		\begin{itemize}
			\item $E:= E_1 \sqcup E_2$.
			\item $B:= B_1 \sqcup B_2$.
			\item For each i $\in \lbr 1,2, \ldots, n \rbr$, $l_i \colon B \ra  \mathbb{B}[E]$ is defined as 
			\begin{equation*}
				\begin{split}
					l_i \colon  B & \ra \mathbb{B}[E]\\
					 (b_1, B_1) & \mapsto \mathbb{B}[j_1]\big(l_{i {_1}}(b_1) \big)\\
					 (b_2, B_2) & \mapsto \mathbb{B}[j_2]\big(l_{i {_2}}(b_2) \big),
				\end{split}
			\end{equation*}
   \end{itemize}
		where $j_1 \colon E_1 \ra E, x \mapsto (x, E_1)$ and $j_2 \colon E_2 \ra E, x \mapsto (x, E_2)$.
	\end{lemma}

We demonstrate the computation of pushouts  in \textbf{Process}$_{n}$ later in Lemma \ref{ARLemma:Pushout Biochemical process species} for the special case relevant to the goal of the current manuscript.

\begin{proposition}\label{Proposition:embedding}
 Let $m,n \in \mb{N}$, such that $m < n$. Then, there is a full embedding from the category \textbf{Process}$_{m}$ to the category \textbf{Process}$_{n}$.
\end{proposition}
 
\begin{proof}
 Proof directly follows from the fullness and faithfulness of the functor $$i \colon \textbf{Process}_{m} \ra \textbf{Process}_{n}$$ which takes a biochemical process network $(E,B, \lbr l_i \rbr_{m})$ with $m$ legs to a  biochemical process network $(E,B, \lbr l_i \rbr_{n})$ with $n$ legs such that $l_{k}=0$ for $k > m$.
\end{proof}
Hence, if $m < n \in \mb{N}$, then the Proposition \ref{Proposition:embedding} identifies a process network having $m$ legs $\mc{B} = (E, B, \lbr l_i \rbr_{m})$ with a process network having $n$ legs $\mc{B}' = (E, B, \lbr l_i \rbr_{n})$, where $l_k \colon B \ra \mathbb{B}[E]$ are zero functions for $k>m$, see Figure \ref{fig:embedding}.
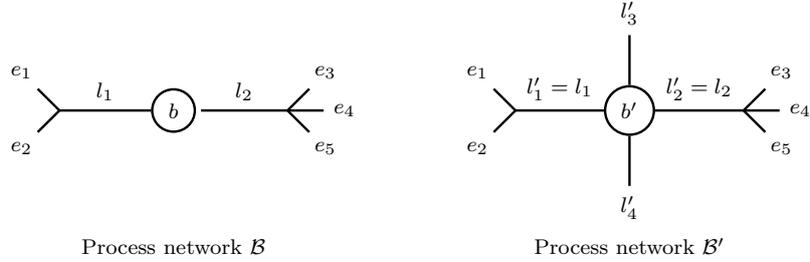
\begin{figure}[htb]
	\begin{center}
	\begin{tikzpicture}[font={\footnotesize}, line width = 0.9pt]
		%
		%
		\node at (2.5,0.2) [] {Process network $\mathcal{B}$};
		\node (b) at (2.5,2) [circle, draw] {$b$};
		\coordinate (l1l) at (1,2);
		\node (e2l) at (0.5,1.5) {$e_{2}$};
		\draw[-] (e2l) -- (l1l);
		\draw[-] (l1l) -- (b) node[pos=0.5, above] {$l_1$};
		\node (e1l) at (0.5,2.5) {$e_{1}$};
		\draw[-] (e1l) -- (l1l);
		\coordinate (l2l) at (4,2);
		\draw[-] (l2l) -- (b2l) node[pos=0.5, above] {$l_2$};
		\node (e32) at (4.5,1.5) {$e_{5}$};
		\draw[-] (e32) -- (l2l);
		\node (e31) at (4.5,2.5) {$e_{3}$};
		\draw[-] (e31) -- (l2l);
		\node (e4) at (4.75,2) {$e_{4}$};
		\draw[-] (e4) -- (l2l);
		%
		\node at (8.5,0.2) [] {Process network $\mathcal{B}'$};
		\node (bp) at (8.5,2) [circle, draw] {$b'$};
		\coordinate (l1r) at (7,2);
		\node (e2r) at (6.5,2.5) {$e_{1}$};
		\draw[-] (e2r) -- (l1r);
		\draw[-] (l1r) -- (bp) node[pos=0.5, above] {$l'_1=l_1$};
		\node (e12p) at (6.5,1.5) {$e_{2}$};
		\draw[-] (e12p) -- (l1r);
		\coordinate (l3r) at (8.5,3);
		\draw[-] (bp) -- (l3r) node[above] {$l'_3$};
		\coordinate (l4r) at (8.5,1);
		\draw[-] (bp) -- (l4r) node[below] {$l'_4$};
		\coordinate (l2r) at (10,2);
		\draw[-] (bp) -- (l2r) node[pos=0.5,above] {$l'_2=l_2$};
		\node (e3r) at (10.5,2.5) {$e_{3}$};
		\draw[-] (e3r) -- (l2r);
		\node (e5r) at (10.5,1.5) {$e_{5}$};
		\draw[-] (e5r) -- (l2r);
		\node (e4r) at (10.75,2) {$e_{4}$};
		\draw[-] (e4r) -- (l2r);
	\end{tikzpicture}
	\end{center}
	\vspace{-0.4cm}
	\caption{An illustration showing that the Proposition \ref{Proposition:embedding} allows us to treat the process network $\mathcal{B}=( E :=\lbrace e_1,e_2, e_3,e_4,e_5 \rbrace, B:= \lbrace b \rbrace, \lbrace l_i \rbrace_{2})$ as same to the the process network $\mathcal{B}'=( E':=\lbrace e_1,e_2, e_3,e_4,e_5 \rbrace, B':= \lbrace b \rbrace, \lbrace l'_i\rbrace_{4})$ such that $l'_1=l_1$, $l'_2=l_2$, $l'_3=0$ and $l'_4=0$. }
	\label{fig:embedding}
\end{figure}


\begin{remark}[Relation to Petri net with link]
In \cite{compositionality:13637}, the authors defined a \textit{Petri net with link} as a functor $P \colon {\rm{Sch(LPetri)}} \ra {\rm{\textbf{Set}}}$, where ${\rm{Sch(LPetri)}}$ is the category freely generated by the following morphisms:
\[
\begin{tikzcd}
  & \text{Input arc} \arrow[ld, "{\rm{src}}_{I}"'] \arrow[rd, "{\rm{tgt}}_{I}"] &   \\ S & \text{Output arc} \arrow[l, "{\rm{src}}_{O}"'] \arrow[r, "{\rm{tgt}}_{O}"]   & T \\
  & \text{Link} \arrow[lu, "{\rm{src}}_{L}"] \arrow[ru, "{\rm{tgt}}_{L}"'] &  
\end{tikzcd}
\] 

Now, let us consider a subclass of Petri nets with links $P \colon {\rm{Sch(LPetri)}} \ra {\rm{\textbf{Set}}}$ which satisfy the following condition: For each $b \in P(T)$, the sets $P({\rm{src}}_I) \big( P({\rm{tgt}}_I)^{-1}(b) \big) $, $P({\rm{src}}_L) \big( P({\rm{tgt}}_L)^{-1}(b) \big) $, and $P({\rm{src}}_{O}) \big( P({\rm{tgt}}_{O})^{-1}(b) \big)$ are finite sets, and the functions $P({\rm{src}}_{I}) \colon P(\text{Input arc}) \ra P(S)$, $P({\rm{src}}_{L}) \colon P(\text{Link}) \ra P(S)$ and $P({\rm{src}}_{O}) \colon P(\text{Input arc}) \ra P(S)$ are injective functions. Then, it is a lengthy but straightforward to show a one-one correspondence between such a subclass of Petri nets with links and the set of process networks with three legs.
\end{remark}

\subsection{Process networks with interfaces as structured cospans}\label{Subsection:Process networks with interfaces as structured cospans}
To introduce the notion of a process network with an interface, we will borrow the theory of structured cospans as developed in \cite{baez2020structuredcospans}. To be more precise, we will use the theory of structured cospans to construct a symmetric monoidal double-category whose horizontal 1-morphisms are process networks with interfaces. By an interface of a process network, we mean a set of entities (see Definition \ref{ARDefinition: Biochemical process networks}) associated to the process network through which the process network connects with other process networks. In mathematical terms, we will represent process networks with interfaces as structured cospans.
We begin our treatment by recalling the definition of a structured cospan.

\begin{definition}[{\textbf{Section 2}, \cite{baez2020structuredcospans}}]\label{Definition:structured-cospans}
	Let $C$ and $X$ be a pair of categories. A \textit{structured cospan} consists of a functor $L \colon C \ra X$ along with a cospan in the category $X$ of the following form: 
	\[
	\begin{tikzcd}
		& x &                    \\
		L(a) \arrow[ru, "\mc{I}"] &   & L(b) \arrow[lu, "\mc{O}"']
	\end{tikzcd}
 \]
	where $a,b$ are objects of $C$ and $x$ is an object of $X$. Morphisms $\mc{I} \colon L(a) \ra x$ and $\mc{O} \colon L(b) \ra x$ are called the \textit{left leg} and the \textit{right leg} of the  above structured cospan.
 \end{definition} 
 
To obtain a notion of an interface of a process network using structured cospans, we in particular need to focus on a particular case, by constructing a functor $L \colon {\rm{\textbf{Set}}}\ra  {\rm{\textbf{Process}}}_{n}$ for each $n \in \mb{N}$. Using the notation $\emptyset$ for both the empty set and the unique function between two empty sets, we state the following lemma.
\begin{lemma}\label{ARLemma: Bio Sets}
	\begin{equation*}
		\begin{split}
			L \colon {\rm{\textbf{Set}}} & \ra {\rm{\textbf{Process}}}_{n}\\
			 E & \mapsto \big(E, \emptyset, \lbr \underbrace{ \emptyset_{\mathbb{B}[E]}, \emptyset_{\mathbb{B}[E]}, \emptyset_{\mathbb{B}[E]}, \emptyset_{\mathbb{B}[E]}}_{n-times} \rbr \big)\\
			 (\alpha \colon E \ra E') & \mapsto \big (\alpha, \emptyset)
		\end{split}
	\end{equation*}
	defines a functor, where $\emptyset_{\mathbb{B}[E]} \colon \emptyset \ra \mathbb{B}[\mb{E}]$ is the canonical unique map. 
\end{lemma}

\begin{definition}[Open Process network]\label{Definition:Open process networks}
	An \textit{open process network} or a \textit{process network with an interface} is  defined as a structured cospan given by the functor $L \colon {\rm{\textbf{Set}}} \ra {\rm{\textbf{Process}}}_{n}$ defined in Lemma \ref{ARLemma: Bio Sets} and a cospan in \textbf{Process}$_{n}$ of the following form
	\[
	\begin{tikzcd}
		& \mc{B} &                    \\
		L(X) \arrow[ru, "\mc{I}"] &   & L(Y) \arrow[lu, "\mc{O}"']
	\end{tikzcd}
    \]
	where $X$ and $Y$ are  sets, and $\mc{B}$ is a process network. We denote the above process network with an interface by the tuple $(\mc{B},X,Y, \mc{I}, \mc{O})$. We will call the left leg $\mc{I}$ and right leg $\mc{O}$ as the \textit{left interface} and the \textit{right interface of the process network} $\mc{B}$.
\end{definition}
\begin{example}\label{Example:open-proces-network}
In Figure \ref{fig:open-process-network}, we illustrate an open process network $(\mc{B}, \lbr 1\rbr, \lbr 2,3 \rbr, \mc{I}, \mc{O})$, where  $\mc{B}= (\lbr e^{b}_{11}, e^{b}_{12}, e^{b}_{21}, e^{b}_{31}, e^{b}_{32} \rbr, \lbr b \rbr, \lbr l_i \rbr_{3})$ is a process network with three legs $l_1, l_2,$ and $l_3$. Here, the functor $L$ is as defined in  Lemma \ref{ARLemma: Bio Sets}, and thus $L(\lbr 1 \rbr)$ and $L(\lbr 1,2 \rbr)$ represent the process networks $\big(\lbr 1 \rbr, \emptyset, \lbr \emptyset_{\mathbb{B}[\lbr 1 \rbr]}, \emptyset_{\mathbb{B}[\lbr 1 \rbr]}, \emptyset_{\mathbb{B}[\lbr 1 \rbr]} \rbr \big)$ and $\big(\lbr 1,2 \rbr, \emptyset, \lbr \emptyset_{\mathbb{B}[\lbr 1,2 \rbr]}, \emptyset_{\mathbb{B}[\lbr 1,2 \rbr]}, \emptyset_{\mathbb{B}[\lbr 1,2 \rbr]} \rbr \big)$  respectively. The left interface $\mc{I}$ specifies the entity $e^{b}_{11}$ and the right interface $\mc{O}$ specifies the entities $e^{b}_{21}$ and $e^{b}_{32}$, through which the open process network interconnects with other open process networks.
    
\end{example}
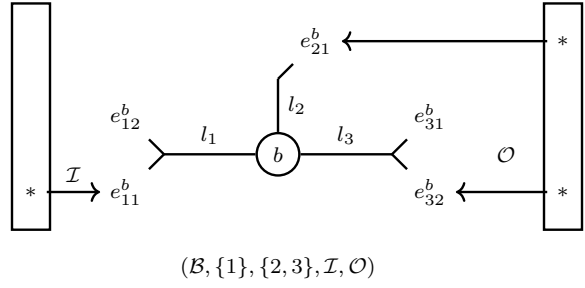
\begin{figure}[htb]
	\begin{center}
	\begin{tikzpicture}[font={\footnotesize}, line width = 0.9pt]
		\node at (3.5,0.5) [] {$(\mathcal{B},\{1\},\{2,3\},\mathcal{I},\mathcal{O})$};
		\node (b-l) at (3.5,2) [circle, draw] {$b$};
		\coordinate (l1-l) at (2,2);
		\node (e11-l) at (1.5,1.5) {$e^b_{11}$};
		\draw[-] (e11-l) -- (l1-l);
		\node (e12-l) at (1.5,2.5) {$e^b_{12}$};
		\draw[-] (l1-l) -- (b-l) node[pos=0.5, above] {$l_1$};
		\draw[-] (e12-l) -- (l1-l);
		\coordinate (l2-l) at (3.5,3);
		\draw[-] (b-l) -- (l2-l) node[pos=0.5, right] {$l_2$};
		\node (e21-l) at (4,3.5) {$e^b_{21}$};
		\draw[-] (e21-l) -- (l2-l);
		\coordinate (l3-l) at (5,2);
		\draw[-] (l3-l) -- (b-l) node[pos=0.5, above] {$l_3$};
		\node (e32-l) at (5.5,1.5) {$e^b_{32}$};
		\draw[-] (e32-l) -- (l3-l);
		\node (e31-l) at (5.5,2.5) {$e^b_{31}$};
		\draw[-] (e31-l) -- (l3-l);
		\draw[-] (0,1) -- (0,4) -- (0.5,4) -- (0.5,1) -- cycle;
		\node (Istar-ll) at (0.25,1.5) {$*$};
		\draw[->] (Istar-ll) -- (e11-l) node[pos=0.5, above] {$\mathcal{I}$};
		\draw[-] (7,1) -- (7,4) -- (7.5,4) -- (7.5,1) -- cycle;
		\node (Ostar-ll) at (7.25,1.5) {$*$};
		\draw[->] (Ostar-ll) -- (e32-l);
		\node (Ostar-lu) at (7.25,3.5) {$*$};
		\node at (6.5,2) {$\mathcal{O}$};
		\draw[->] (Ostar-lu) -- (e21-l);
	\end{tikzpicture}
	\end{center}
	\vspace{-0.3cm}
	\caption{An illustration of an open process network  discussed in Example \ref{Example:open-proces-network}.}
	\label{fig:open-process-network}
\end{figure}


\begin{remark}
To highlight the analogy of Definition \ref{Definition:Open process networks} with an \textit{open system}, a system that interacts with other systems, we call a process network with an interface also an open process network. These kinds of terminologies are common in Applied Category Theory literature; for example, in \cite{MR4085076}, authors have called Petri nets with interfaces as \textit{open Petri nets}, in \cite{MR3694082} and \cite{MR3478745} respectively, authors called the reaction network with an interface and Markov process with an interface as an \textit{open reaction network} and \textit{open Markov process}.
\end{remark}
Now, we need a technical lemma before we can construct our desired symmetric monoidal double category whose horizontal 1-morphisms are open process networks. 
\begin{lemma}\label{ARLemma Left adjoint bio}
	For each $n \in \mb{N}$, the functor $L \colon {\rm{\textbf{Set}}} \ra {\rm{\textbf{Process}}}_{n}$ has a right adjoint.
\end{lemma}
\begin{proof}
	Let us define 	
	\begin{equation*}
		\begin{split}
			R \colon & {\rm{\textbf{Process}}}_{n} \ra {\rm{\textbf{Set}}}\\
			& (E, B, \lbr l_i \rbr_{n} ) \ra E\\
			& \Big((\alpha, \beta) \colon (E, B, \lbr l_i \rbr_{n} ) \ra (E', B', \lbr l_i' \rbr_{n} ) \Big) \mapsto \big( \alpha \colon E \ra E' \big).
		\end{split}
	\end{equation*}
	Observe that from the definition itself, it is clear that $R$ is well-defined and is a functor. Now, we have the following natural isomorphisms: 
	\begin{equation*}
		\begin{split}
			&\hom_{{\rm{\textbf{Set}}}} \big(L(X), (E, B, \lbr l_i \rbr_{n} ) \big)\\
			&\cong \hom_{{\rm{\textbf{Set}}}} \big((X, \emptyset, \lbr \underbrace{ \emptyset_{\mathbb{B}[X]}, \emptyset_{\mathbb{B}[X]}, \emptyset_{\mathbb{B}[X]}, \emptyset_{\mathbb{B}[X]}}_{n-times} \rbr ), (E, B, \lbr l_i \rbr_{n} ) \big)\\
			& \cong \hom_{{\rm{\textbf{Set}}}}(X, E) \,\, [\text{As}\,\, \emptyset \,\,  \text{is an initial object in \textbf{Set}}]\\
			& \cong \hom_{{\rm{\textbf{Set}}}}\Big(X, R\big( (E, B, \lbr l_i \rbr_{n} )  \big) \Big).
		\end{split}
	\end{equation*}
	Hence, $R$ is a left adjoint of $L$. 
\end{proof}

 We now have all the machinery to compose process networks with interfaces, as we see in our following result.

\begin{theorem}\label{ARMain Theorem1}
	Consider the functor $L \colon {\rm{\textbf{Set}}} \ra {\rm{\textbf{Process}}}_{n}$ defined in Lemma \ref{ARLemma: Bio Sets}. Then, there is a symmetric monoidal double category ${\rm{Open}}_{\rm{Double}}({\rm{\textbf{Process}}}_n)$ whose
	\begin{itemize}
		\item objects are sets,
		\item vertical 1-morphisms are functions,
		\item horizontal 1-cells from a set $X$ to a set $Y$ are open process networks
		\[
		\begin{tikzcd}
			& \mc{B} &                    \\
			L(X) \arrow[ru, "\mc{I}"] &   & L(Y) \arrow[lu, "\mc{O}"']
		\end{tikzcd}
  \]
		\item A 2-morphism 
  \begin{center}
		\begin{tikzcd}
			& \mc{B} &                    \\
			L(X) \arrow[ru, "\mc{I}"] &   & L(Y) \arrow[lu, "\mc{O}"']
		\end{tikzcd} $\Rightarrow$
		\begin{tikzcd}
			& \mc{B}' &                    \\
			L(X') \arrow[ru, "\mc{I}'"] &   & L(Y') \arrow[lu, "\mc{O}'"']
		\end{tikzcd}
        
  \end{center}
is given by a tuple $(f \colon X \ra X' , \eta \colon \mc{B} \ra \mc{B}', g \colon Y \ra Y')$, such that the following is a commutative diagram
		\[
		\begin{tikzcd}
			L(X) \arrow[r, "\mc{I}"] \arrow[d, "L(f)"'] & \mc{B} \arrow[d, "\eta"'] & L(Y) \arrow[l, "\mc{O}"'] \arrow[d, "L(g)"] \\
			L(X') \arrow[r, "\mc{I}'"]                 & \mc{B}'                 & L(Y') \arrow[l, "\mc{O}'"']               
		\end{tikzcd}
  \]	
		in \textbf{Process}$_{n}$.		 
		\item Composition of 1-morphisms is the standard composition of functions.
		
		\item Composition of horizontal 1-cells is given by the composition of cospans via pushout constructions.
		
		More precisely, suppose
  \begin{center}
		\begin{tikzcd}
			& \mc{B} &                    \\
			L(X) \arrow[ru, "\mc{I}_1"] &   & L(Y) \arrow[lu, "\mc{O}_1"']
		\end{tikzcd} and 
		\begin{tikzcd}
			& \mc{A} &                    \\
			L(Y) \arrow[ru, "\mc{I}_2"] &   & L(Z) \arrow[lu, "\mc{O}_2"']
		\end{tikzcd} 
        
  \end{center}
		be two process networks. Then, the composite open process network is given by the following cospan from $L(X)$ to $L(Z)$ in \textbf{Process}$_{n}$
		\[
		\begin{tikzcd}
			&                   & \mc{B} +_{L(Y)}\mc{A}                                  &                    &                    \\
			& \mc{B} \arrow[ru, ""] &                                    & \mc{A} \arrow[lu, ""'] &                    \\
			L(X) \arrow[ru, "\mc{I}_1"] &                   & L(Y) \arrow[lu, "\mc{O}_1"'] \arrow[ru, "\mc{I}_2"] &                    & L(Z) \arrow[lu, "\mc{O}_2"']
		\end{tikzcd}
  \]
		where $\mc{B} +_{L(Y)}\mc{A}$ is a chosen pushout square, and the unlabeled maps are the canonical maps to the pushout.
		\item the horizontal composition of 2-morphisms 
  
		\begin{tikzcd}
			L(X) \arrow[r, "\mc{I}_{1}"] \arrow[d, "L(f)"'] & \mc{B} \arrow[d, "\eta_1"'] & L(Y) \arrow[l, "\mc{O}_{1}"'] \arrow[d, "L(g)"] \\
			L(X') \arrow[r, "\mc{I}_{1}'"]                 & \mc{B}'                 & L(Y') \arrow[l, "\mc{O}_{1}'"']               
		\end{tikzcd} and 
		\begin{tikzcd}
			L(Y) \arrow[r, "\mc{I}_{2}"] \arrow[d, "L(g)"'] & \mc{A} \arrow[d, "\eta_{2}"'] & L(Z) \arrow[l, "\mc{O}_{2}"'] \arrow[d, "L(h)"] \\
			L(Y') \arrow[r, "\mc{I}_{2}'"]                 & \mc{A}'                 & L(Z') \arrow[l, "\mc{O}_{2}'"']               
		\end{tikzcd}
		is given as 
		\[
		\begin{tikzcd}
			L(X) \arrow[r, ""] \arrow[d, "L(f)"'] & \mc{B} +_{L(Y)}\mc{A} \arrow[d, "\eta_2 +_{L(g)} \eta_1"'] & L(Z) \arrow[l, ""'] \arrow[d, "L(g)"] \\
			L(X') \arrow[r, ""]                 & \mc{B}' +_{L(Y)}\mc{A}'                 & L(Z') \arrow[l, ""']               
		\end{tikzcd}
        \]
		where the unlabeled maps are the canonical composite maps to the pushouts,
		\item vertical composition of 2-morphisms 
  \begin{center}
  \begin{tikzcd}
			L(X) \arrow[r, "\mc{I}"] \arrow[d, "L(f)"'] & \mc{B} \arrow[d, "\eta"'] & L(Y) \arrow[l, "\mc{O}"'] \arrow[d, "L(g)"] \\
			L(X') \arrow[r, "\mc{I}'"]                 & \mc{B}'                 & L(Y') \arrow[l, "\mc{O}'"']               
		\end{tikzcd} 
  \end{center}
  and 
  \begin{center}
  \begin{tikzcd}
			L(X') \arrow[r, "\mc{I}'"] \arrow[d, "L(f')"'] & \mc{B}' \arrow[d, "\eta'"'] & L(Y') \arrow[l, "\mc{O}'"'] \arrow[d, "L(g')"] \\
			L(X'') \arrow[r, "\mc{I}''"]                 & \mc{B}''                 & L(Y'') \arrow[l, "\mc{O}''"']               
		\end{tikzcd}
           \end{center}
           
 is defined using the composition of functions

  \begin{center}
  \begin{tikzcd}
			L(X) \arrow[r, "\mc{I}"] \arrow[d, "L(f' \circ f)"'] & \mc{B} \arrow[d, "\eta' \circ \eta"'] & L(Y) \arrow[l, "\mc{O}"'] \arrow[d, "L(g' \circ g)"] \\
			L(X'') \arrow[r, "\mc{I}''"]                 & \mc{B}''                 & L(Y'') \arrow[l, "\mc{O}''"']               
		\end{tikzcd} 
  \end{center}
 
		\item The symmetric monoidal structure is derived from the coproducts in \textbf{Set} and \textbf{Process}$_{n}$. More precisely, the monoidal product is defined using chosen coproducts in \textbf{Set} and \textbf{Process}$_{n}$. Hence, 
  \begin{itemize}
  \item the monoidal product of two finite sets $X_1$ and $X_2$ is the disjoint union $X_1 + X_2$,
  \item  the monoidal product of two vertical 1-morphisms $f_1 \colon X_1 \ra Y_1$ and $f_2 \colon X_2 \ra Y_2$ is given by the natural map $f_1 +f_2 \colon X_1 +X_2 \ra Y_1 +Y_2$,
\item the monoidal product of horizontal 1-cells 

\begin{tikzcd}
			& \mc{B}_1 &                    \\
			L(X_1) \arrow[ru, "\mc{I}_1"] &   & L(Y_1) \arrow[lu, "\mc{O}_1"']
		\end{tikzcd} and 
		\begin{tikzcd}
			& B_2 &                    \\
			L(X_2) \arrow[ru, "\mc{I}_2"] &   & L(Y_2) \arrow[lu, "\mc{O}_2"']
		\end{tikzcd} 
  is given as 
  \begin{center}
  \begin{tikzcd}
			& \mc{B}_1 + \mc{B}_2 &                    \\
			L(X_1+X_2) \arrow[ru, "\mc{I}_1 + \mc{I}_2"] &   & L(Y_1+Y_2) \arrow[lu, "\mc{O}_1 +\mc{O}_2"']
		\end{tikzcd}
  \end{center}
  \item the monoidal product of two 2-morphisms

  \begin{tikzcd}
			L(X_1) \arrow[r, "\mc{I}_1"] \arrow[d, "L(f_1)"'] & \mc{B}_1 \arrow[d, "\eta_1"'] & L(Y_1) \arrow[l, "\mc{O}_1"'] \arrow[d, "L(g_1)"] \\
			L(X_1') \arrow[r, "\mc{I}_1'"]                 & \mc{B}_1'                 & L(Y_1') \arrow[l, "\mc{O}_1'"']               
		\end{tikzcd} and 
		\begin{tikzcd}
			L(X_2) \arrow[r, "\mc{I}_2"] \arrow[d, "L(f_2)"'] & \mc{B}_2 \arrow[d, "\eta_2"'] & L(Y_2) \arrow[l, "\mc{O}_2"'] \arrow[d, "L(g_2)"] \\
			L(X_2') \arrow[r, "\mc{I}_2'"]                 & \mc{B}_2'                 & L(Y_2') \arrow[l, "\mc{O}_2'"']               
		\end{tikzcd}
		is given as 
		\[
		\begin{tikzcd}
			L(X_1+X_2) \arrow[r, "\mc{I}_1+ \mc{I}_2"] \arrow[d, "L(f_1+f_2)"'] & \mc{B}_1+ \mc{B}_2 \arrow[d, "\eta_1+\eta_2"'] & L(Y_1+Y_2) \arrow[l, "\mc{O}_1 + \mc{O}_2"'] \arrow[d, "L(g_1+g_2)"] \\
			L(X_1' +X_2') \arrow[r, "\mc{I}_1'+ \mc{I}_2'"]                 & \mc{B}_1' + \mc{B}_2'                & L(Y_1'+Y_2') \arrow[l, "\mc{O}_1'+ \mc{O}_2'"']               
		\end{tikzcd}
        \]
We consider initial objects as units for these monoidal products, and the symmetry is defined using the canonical isomorphism $X+Y \cong Y+ X$.
  \end{itemize}

	\end{itemize}
\end{theorem}		
\begin{proof}
	Since the category \textbf{Process}$_{n}$ is finitely cocomplete (Proposition 
\ref{Proposition: Bio is finitely cocomplete}), and the functor $L \colon {\rm{\textbf{Set}}} \ra {\rm{\textbf{Process}}}_{n}$ has a right adjoint (Lemma \ref{ARLemma Left adjoint bio}), then the proof our theorem follows from the \textit{Lemma 14} of \cite{MR4085076}.
\end{proof}

The following lemma and the Lemma \ref{ARLemma:coproduct Biochemical process species} show respectively, how the composition and the monoidal product of open process networks look explicitly in the light of Theorem \ref{ARMain Theorem1}.

\begin{lemma}\label{ARLemma:Pushout Biochemical process species}
		Consider two process networks $\mc{B}_1=(E_1, B_1, \lbr l_{i_{1}} \rbr_{n} ) $ and  $\mc{B}_2 =(E_2, B_2, \lbr l_{i_{2}} \rbr_{n} )$ in \textbf{Process}$_{n}$. 
		Let $\mc{O}_1 \colon L(Y) \ra \mc{B}_1$ and $\mc{I}_2 \colon L(Y) \ra \mc{B}_2$ be two morphisms in the category \textbf{Process}$_{n}$, where $L \colon {\rm{\textbf{Set}}} \ra {\rm{\textbf{Process}}}_{n}$ is the functor as defined in Lemma \ref{ARLemma: Bio Sets}. Then, the pushout  process network $\mc{B}_1 +_{\mc{O}_1, L(Y), \mc{I}_2}\mc{B}_2= (E, B, \lbr l_i \rbr_{n} )$  exists, and the following defines a pushout up to an isomorphism:
		\begin{itemize}
			\item $E:= E_1 +_{O_1, Y, I_2} E_2$, where $O_1 \colon Y \ra E_1$ and $I_2 \colon Y \ra E_2$ are the underlying canonical functions associated to $\mc{O}_1$ and $\mc{I}_2$, respectively.
			\item $B:= B_1 \sqcup B_2$.
			\item For each i $\in \lbr 1,2, \ldots, n \rbr$, $l_i \colon B \ra  \mathbb{B}[E]$ is defined as 
			\begin{equation*}
				\begin{split}
					l_i \colon  B & \ra \mathbb{B}[E]\\
					 (b_1, B_1) & \mapsto \mathbb{B}[j_1]\big(l_{i {_1}}(b_1) \big)\\
					 (b_2, B_2) & \mapsto \mathbb{B}[j_2]\big(l_{i {_2}}(b_2) \big),
				\end{split}
			\end{equation*}
   \end{itemize}
		where $j_1 \colon E_1 \ra E, x \mapsto [(x, E_1)]$ and $j_2 \colon E_2 \ra E, x \mapsto [(x, E_2)]$, where $[\,\,]$ represents the equivalence class.
	\end{lemma}
\begin{proof}
		The existence of pushout is a direct consequence of Theorem \ref{Proposition: Bio is finitely cocomplete}. Checking $\mc{B}_1 +_{\mc{O}_1, L(Y), \mc{I}_2}\mc{B}_2$	is indeed a pushout in \textbf{Process}$_{n}$ is a routine verification.

	\end{proof}	
 
 \begin{example}[{Combining open process networks using the  composition laws of horizontal 1-morphisms}]\label{Example:ombining open process networks using the  composition laws of horizontal 1-morphisms} In Figure \ref{fig:composition-of-open-process-network}, we illustrate the composition of open process networks $(\mc{B}, \lbr 1\rbr, \lbr 2,3 \rbr, \mc{I}, \mc{O})$ and  $(\mc{B}', \lbr 2,3\rbr, \lbr 4 \rbr, \mc{I}', \mc{O}')$, where 
 $$
 \mc{B}= \big( \lbr e^{b}_{11}, e^{b}_{12}, e^{b}_{21}, e^{b}_{31}, e^{b}_{32} \rbr, \lbr b \rbr, \lbr l_i \rbr_{3}\big)
 $$ 
 and 
 $$
 \mc{B}'= \big( \lbr e^{b'}_{11}, e^{b'}_{12}, e^{b'}_{31}, e^{b'}_{32} \rbr, \lbr b' \rbr, \lbr l'_i \rbr_{3}\big)
 $$ 
 are process networks. Here, the functor $L$ is as defined in Lemma \ref{ARLemma: Bio Sets}. While composing using Theorem \ref{ARMain Theorem1}, the right interface $\mc{O}$ of $(\mc{B}, \lbr 1\rbr, \lbr 2,3 \rbr, \mc{I}, \mc{O})$ and the left interface $\mc{I}'$ of $(\mc{B}', \lbr 2,3\rbr, \lbr 4 \rbr, \mc{I}', \mc{O}')$ specify the entities to be identified. In particular, they identify $e^{b}_{21}$ and $e^{b'}_{12}$ as an entity $e$ and identify $e^{b}_{32}$ and $e^{b'}_{11}$ as an entity $d$. The process network $\mc{B} +_{\mc{O},L(\lbr 2,3 \rbr), \mc{I}'} \mc{B}'$ in the composite open process network $(\mc{B} +_{\mc{O},L(\lbr 2,3 \rbr), \mc{I}'} \mc{B}', \lbr 1 \rbr, \lbr 4 \rbr, \mc{I}, \mc{O}', \mc{B}')$ is computed using the specification given in Lemma \ref{ARLemma:Pushout Biochemical process species}. Explicitly, $\mc{B} +_{\mc{O},L(\lbr 2,3 \rbr), \mc{I}'} \mc{B}'= (\lbr e^{b}_{11}, e^{b}_{12}, e, e^{b}_{31}, d, e^{b'}_{31}, e^{b'}_{32}\rbr, \lbr b,b' \rbr, \lbr \bar{l_i} \rbr_{3} )$, where the functions $\bar{l_i}$ are evident in the Figure \ref{fig:composition-of-open-process-network}.
 \end{example}

 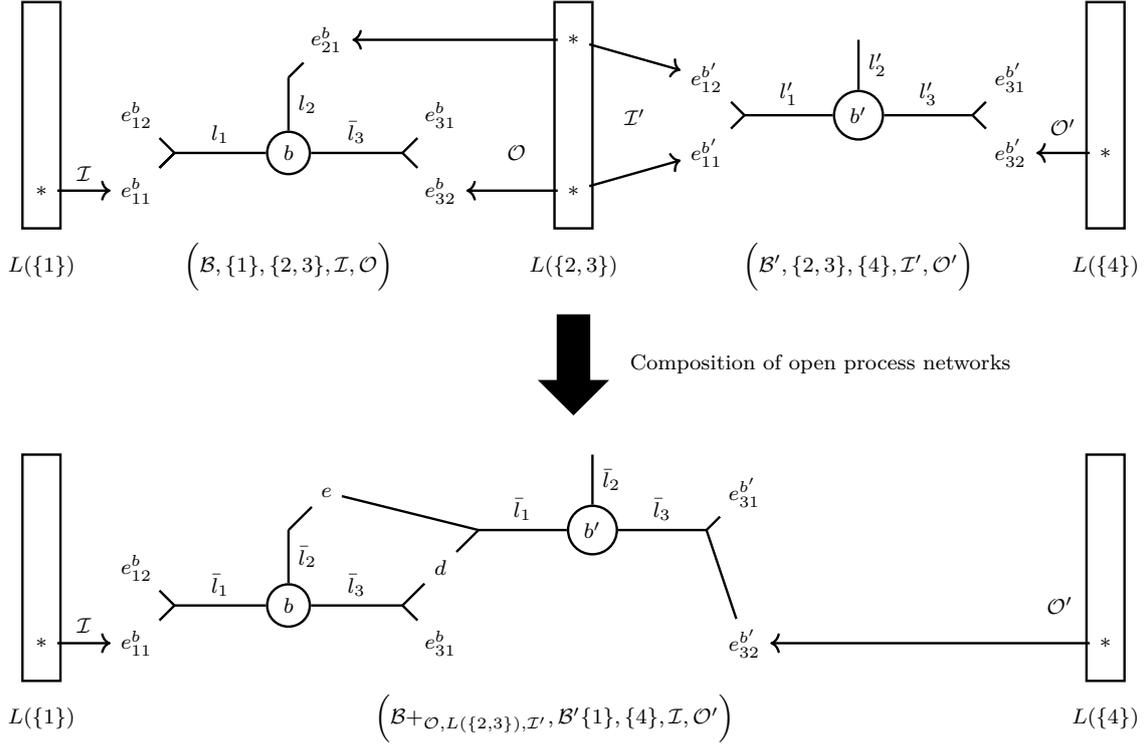
\begin{figure}[H]
	\begin{center}
	\begin{tikzpicture}[font={\footnotesize}, line width = 0.9pt]
		\node at (7,0.5) [] {$\Big(\mathcal{B}+_{\mathcal{O},L(\{2,3\}),\mathcal{I}'}, \mathcal{B}'\{1\},\{4\},\mathcal{I},\mathcal{O}'\Big)$};
		\node (b-l) at (3.5,2) [circle, draw] {$b$};
		\coordinate (l1-l) at (2,2);
		\node (e11-l) at (1.5,1.5) {$e^b_{11}$};
		\draw[-] (e11-l) -- (l1-l);
		\node (e12-l) at (1.5,2.5) {$e^b_{12}$};
		\draw[-] (l1-l) -- (b-l) node[pos=0.5, above] {$\bar{l}_1$};
		\draw[-] (e12-l) -- (l1-l);
		\coordinate (l2-l) at (3.5,3);
		\draw[-] (b-l) -- (l2-l) node[pos=0.5, right] {$\bar{l}_2$};
		\node (e-l) at (4,3.5) {$e$};
		\draw[-] (e-l) -- (l2-l);
		\coordinate (l3-l) at (5,2);
		\draw[-] (l3-l) -- (b-l) node[pos=0.5, above] {$\bar{l}_3$};
		\node (d-l) at (5.5,2.5) {$d$};
		\draw[-] (d-l) -- (l3-l);
		\node (e31-l) at (5.5,1.5) {$e^b_{31}$};
		\draw[-] (e31-l) -- (l3-l);
		\draw[-] (0,1) -- (0,4) -- (0.5,4) -- (0.5,1) -- cycle;
		\node (Istar-ll) at (0.25,1.5) {$*$};
		\draw[->] (Istar-ll) -- (e11-l) node[pos=0.5, above] {$\mathcal{I}$};
		\node () at (0.25,0.5) {$L(\{1\})$};
		\draw[-] (14,1) -- (14,4) -- (14.5,4) -- (14.5,1) -- cycle;
		\node (Ostar-ll) at (14.25,1.5) {$*$};
		\node at (13.65,2) {$\mathcal{O}'$};
		\node () at (14.25,0.5) {$L(\{4\})$};
		\coordinate (l1b-l) at (6,3);
		\draw[-] (e-l) -- (l1b-l);
		\draw[-] (d-l) -- (l1b-l);
		\node (bp-l) at (7.5,3) [circle, draw] {$b'$};
		\draw[-] (bp-l) -- (l1b-l) node[pos=0.5, above] {$\bar{l}_1$};
		\coordinate (l2b-l) at (7.5,4);
		\draw[-] (bp-l) -- (l2b-l) node[pos=0.5, right] {$\bar{l}_2$};
		\coordinate (l3b-l) at (9,3);
		\draw[-] (bp-l) -- (l3b-l) node[pos=0.5, above] {$\bar{l}_3$};
		\node (e31bp-l) at (9.5,3.5) {$e^{b'}_{31}$};
		\draw[-] (l3b-l) -- (e31bp-l);
		\node (e32bp-l) at (9.5,1.5) {$e^{b'}_{32}$};
		\draw[-] (l3b-l) -- (e32bp-l);
		\draw[->] (Ostar-ll) -- (e32bp-l);
		%
		%
		\draw[-] (0,7) -- (0,10) -- (0.5,10) -- (0.5,7) -- cycle;
		\node (Istar-ul) at (0.25,7.5) {$*$};
		\node () at (0.25,6.5) {$L(\{1\})$};
		\node (e11-ul) at (1.5,7.5) {$e^b_{11}$};
		\draw[->] (Istar-ul) -- (e11-ul) node[pos=0.5, above] {$\mathcal{I}$};
		\coordinate (l1-ul) at (2,8);
		\node (b-ul) at (3.5,8) [circle, draw] {$b$};
		\draw[-] (e11-ul) -- (l1-ul);
		\node (e12-ul) at (1.5,8.5) {$e^b_{12}$};
		\draw[-] (l1-ul) -- (b-ul) node[pos=0.5, above] {$l_1$};
		\draw[-] (e12-ul) -- (l1-ul);
		\draw[-] (e11-ul) -- (l1-ul);
		\coordinate (l2-ul) at (3.5,9);
		\draw[-] (b-ul) -- (l2-ul) node[pos=0.5, right] {$l_2$};
		\node (e-ul) at (4,9.5) {$e^b_{21}$};
		\draw[-] (e-ul) -- (l2-ul);
		\coordinate (l3-ul) at (5,8);
		\draw[-] (l3-ul) -- (b-ul) node[pos=0.5, above] {$\bar{l}_3$};
		\node (e32-ul) at (5.5,7.5) {$e^b_{32}$};
		\draw[-] (e32-ul) -- (l3-ul);
		\node (e31-ul) at (5.5,8.5) {$e^b_{31}$};
		\draw[-] (e31-ul) -- (l3-ul);
		\node at (3.5,6.5) [] {$\Big(\mathcal{B},\{1\},\{2,3\},\mathcal{I},\mathcal{O}\Big)$};
		\draw[-] (7,7) -- (7,10) -- (7.5,10) -- (7.5,7) -- cycle;
		\node () at (7.25,6.5) {$L(\{2,3\})$};
		\node (Ostar-uru) at (7.25,7.5) {$*$};
		\node (Ostar-url) at (7.25,9.5) {$*$};
		\draw[->] (Ostar-url) -- (e-ul);
		\node at (6.5,8) {$\mathcal{O}$};
		\draw[->] (Ostar-uru) -- (e32-ul);
		\node [fill=black, single arrow, draw=none, text=black, rotate=-90] at (7.25,5.3) {compo};
		\node at (10.5,5.2) [] {Composition of open process networks};
		\node (e11-ur) at (9,8) {$e^{b'}_{11}$};
		\draw[->] (Ostar-uru) -- (e11-ur) node[pos=0.5, above, yshift=5mm] {$\mathcal{I}'$};
		\node (e12-ur) at (9,9) {$e^{b'}_{12}$};
		\draw[->] (Ostar-url) -- (e12-ur);
		\coordinate (l1-ur) at (9.5,8.5);
		\draw[-] (l1-ur) -- (e11-ur);
		\draw[-] (l1-ur) -- (e12-ur);
		\node (bp-l) at (11,8.5) [circle, draw] {$b'$};
		\draw[-] (l1-ur) -- (bp-l) node[pos=0.5, above] {$l'_1$};
		\coordinate (l2b-ur) at (11,9.5);
		\draw[-] (bp-l) -- (l2b-ur) node[pos=0.5, right] {$l'_2$};
		\coordinate (l3p-ur) at (12.5,8.5);
		\draw[-] (bp-l) -- (l3p-ur) node[pos=0.5, above] {$l'_3$};
		\node at (11,6.5) [] {$\Big(\mathcal{B}',\{2,3\},\{4\},\mathcal{I}',\mathcal{O}'\Big)$};
		\node (e31p-ur) at (13,9) {$e^{b'}_{31}$};
		\draw[-] (l3p-ur) -- (e31p-ur);
		\node (e32p-ur) at (13,8) {$e^{b'}_{32}$};
		\draw[-] (l3p-ur) -- (e32p-ur);
		\draw[-] (14,7) -- (14,10) -- (14.5,10) -- (14.5,7) -- cycle;
		\node (Ost-uru) at (14.25,8) {$*$};
		\draw[->] (Ost-uru) -- (e32p-ur) node[pos=0.5, above, yshift=1mm] {$\mathcal{O}'$};
		\node () at (14.25,6.5) {$L(\{4\})$};
	\end{tikzpicture}
	\end{center}
	\vspace{-0.3cm}
	\caption{An illustration of composing open process networks using the composition law of horizontal 1-morphisms in Theorem \ref{ARMain Theorem1}, see Example \ref{Example:ombining open process networks using the  composition laws of horizontal 1-morphisms}.}
	\label{fig:composition-of-open-process-network}
\end{figure}


\begin{example}[{Combining open process networks using the monoidal product of horizontal 1-morphisms}]\label{Example:Combining open process networks using the monoidal product of horizontal 1-morphisms}

In Figure \ref{fig:monoidal-product-of-open-process-network}, we illustrate the monoidal product of open process networks $(\mc{B}, \lbr 1 \rbr, \lbr 4 \rbr, \mc{I}, \mc{O})$ and $(\mc{B}', \lbr 2,3 \rbr, \lbr 5 \rbr, \mc{I}', \mc{O}')$  using the monoidal product defined in Theorem \ref{ARMain Theorem1}, where $\mc{B}=( \lbr e^{b}_{11}, e^{b}_{12}, e^{b}_{21}, e^{b}_{31}, e^{b}_{32} \rbr, \lbr b \rbr, \lbr l_i \rbr_{3})$ and $\mc{B}'= (\lbr e^{b'}_{11}, e^{b'}_{12}, e^{b'}_{31}, e^{b'}_{32} \rbr, \lbr b' \rbr, \lbr l'_i \rbr_{3})$ Here, the functor $L$ is as defined in Lemma \ref{ARLemma: Bio Sets}. Observe that the  process network $\mc{B} + \mc{B}'$ in the composite open process network $(\mc{B} + \mc{B}', \lbr 1 \rbr \sqcup \lbr 2,3 \rbr, \lbr 4 \rbr \sqcup \lbr 5 \rbr, \mc{I} + \mc{I}', \mc{O} + \mc{O}')$ is computed using the specification given in Lemma \ref{ARLemma:coproduct Biochemical process species}. Explicitly, the process network $$\mc{B} + \mc{B}'=( \lbr e^{b}_{11}, e^{b}_{12}, e^{b}_{21}, e^{b}_{31}, e^{b}_{32}, e^{b'}_{11}, e^{b'}_{12}, e^{b'}_{31}, e^{b'}_{32}\rbr, \lbr b,b' \rbr, \lbr \bar{l}_{i}\rbr_{3} ),$$ where the functions $\bar{l_i}$ are evident in the Figure \ref{fig:monoidal-product-of-open-process-network}.
\end{example}

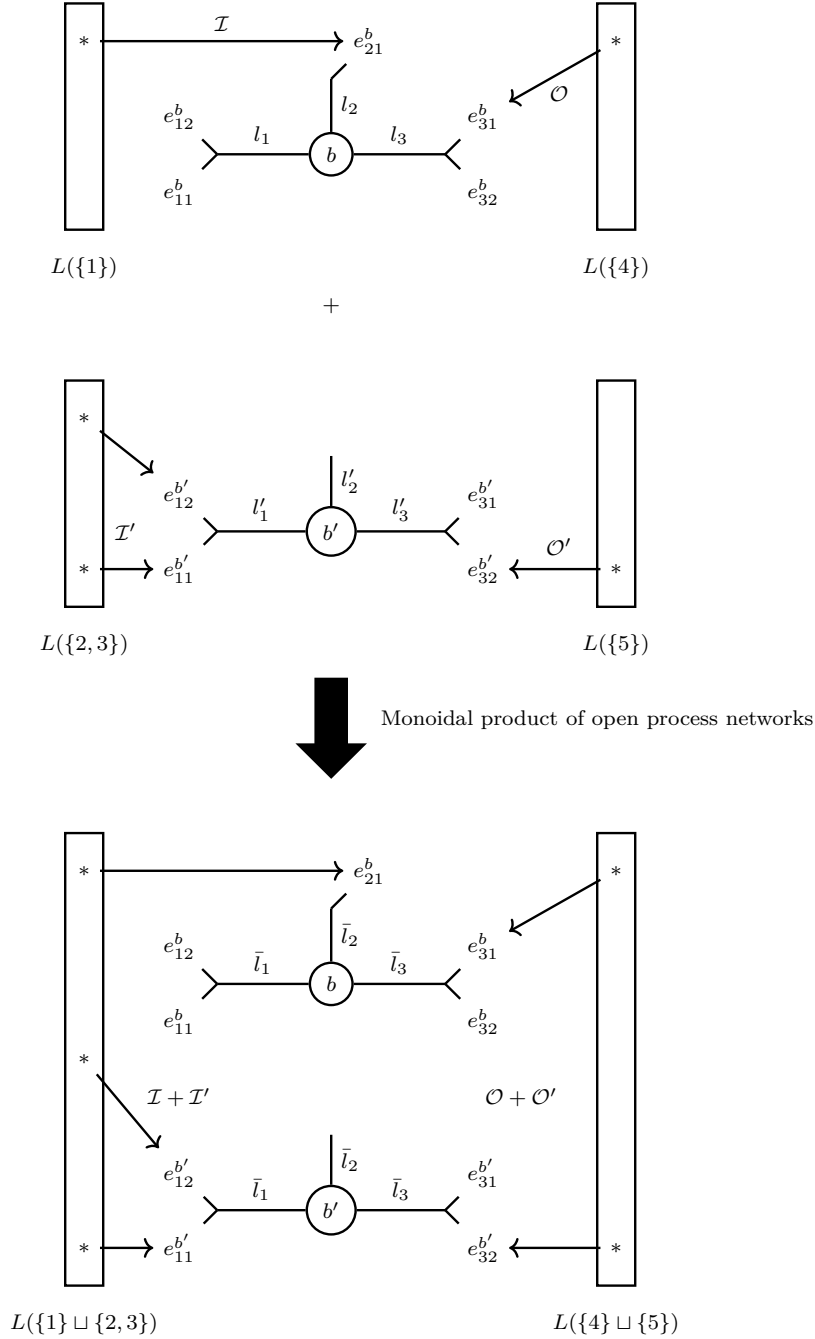
\begin{figure}[H]
	\begin{center}
	\begin{tikzpicture}[font={\footnotesize}, line width = 0.9pt]
		%
		\node (b-l) at (5,14) [circle, draw] {$b$};
		\coordinate (l1-l) at (3.5,14);
		\node (e11-l) at (3,13.5) {$e^b_{11}$};
		\draw[-] (e11-l) -- (l1-l);
		\node (e12-l) at (3,14.5) {$e^b_{12}$};
		\draw[-] (l1-l) -- (b-l) node[pos=0.5, above] {$l_1$};
		\draw[-] (e12-l) -- (l1-l);
		\coordinate (l2-l) at (5,15);
		\draw[-] (b-l) -- (l2-l) node[pos=0.5, right] {$l_2$};
		\node (e21-l) at (5.5,15.5) {$e^b_{21}$};
		\draw[-] (e21-l) -- (l2-l);
		\coordinate (l3-l) at (6.5,14);
		\draw[-] (l3-l) -- (b-l) node[pos=0.5, above] {$l_3$};
		\node (e32-l) at (7,13.5) {$e^b_{32}$};
		\draw[-] (e32-l) -- (l3-l);
		\node (e31-l) at (7,14.5) {$e^b_{31}$};
		\draw[-] (e31-l) -- (l3-l);
		\draw[-] (1.5,13) -- (1.5,16) -- (2,16) -- (2,13) -- cycle;
		\node (Istar-ll) at (1.75,15.5) {$*$};
		\draw[->] (Istar-ll) -- (e21-l) node[pos=0.5, above] {$\mathcal{I}$};
		\draw[-] (8.5,13) -- (8.5,16) -- (9,16) -- (9,13) -- cycle;
		\node (Ostar-lu) at (8.75,15.5) {$*$};
		\node at (8,14.8) {$\mathcal{O}$};
		\draw[->] (Ostar-lu) -- (e31-l);
  \node at (1.75,12.5) {$L(\lbrace 1 \rbrace)$};
\node at (8.75,12.5) {$L(\lbrace 4 \rbrace)$};

\node (b-l) at (5,9) [circle, draw] {$b'$};
		\coordinate (l1-l) at (3.5,9);
		\node (e11-l) at (3,8.5) {$e^{b'}_{11}$};
		\draw[-] (e11-l) -- (l1-l);
		\node (e12-l) at (3,9.5) {$e^{b'}_{12}$};
		\draw[-] (l1-l) -- (b-l) node[pos=0.5, above] {$l'_1$};
		\draw[-] (e12-l) -- (l1-l);
		\coordinate (l2-l) at (5,10);
		\draw[-] (b-l) -- (l2-l) node[pos=0.5, right] {$l'_2$};
		%
		\coordinate (l3-l) at (6.5,9);
		\draw[-] (l3-l) -- (b-l) node[pos=0.5, above] {$l'_3$};
		\node (e32-l) at (7,8.5) {$e^{b'}_{32}$};
		\draw[-] (e32-l) -- (l3-l);
		\node (e31-l) at (7,9.5) {$e^{b'}_{31}$};
		\draw[-] (e31-l) -- (l3-l);
		\draw[-] (1.5,8) -- (1.5,11) -- (2,11) -- (2,8) -- cycle;
		\node (Istar-upper) at (1.75,10.5) {$*$};
		\node (Istar-lower) at (1.75,8.5) {$*$};
		\draw[-] (8.5,8) -- (8.5,11) -- (9,11) -- (9,8) -- cycle;
		\node (Ostar-lu) at (8.75,8.5) {$*$};
		\node at (8,8.8) {$\mathcal{O}'$};
		\draw[->] (Ostar-lu) -- (e32-l);
  \node at (1.75,7.5) {$L(\lbrace 2,3 \rbrace)$};
\node at (8.75,7.5) {$L(\lbrace 5 \rbrace)$};
\node at (5,12) {$+$};
\node at (2.30,9) {$\mathcal{I}'$};
\draw[->] (Istar-upper) -- (e12-l);
\draw[->] (Istar-lower) -- (e11-l);

 \node [fill=black, single arrow, draw=none, text=black, rotate=-90] at (5,6.5) {compo};
		\node at (8.5, 6.5) [] {Monoidal product of open process networks};


	\draw[-] (1.5,-1) -- (1.5,5) -- (2,5) -- (2,-1) -- cycle;
 
 \draw[-] (8.5,-1) -- (8.5,5) -- (9,5) -- (9,-1) -- cycle;

 \node (b-l) at (5,3) [circle, draw] {$b$};
		\coordinate (l1-l) at (3.5,3);
		\node (e11-l) at (3,2.5) {$e^b_{11}$};
		\draw[-] (e11-l) -- (l1-l);
		\node (be12-l) at (3,3.5) {$e^b_{12}$};
		\draw[-] (l1-l) -- (b-l) node[pos=0.5, above] {$\bar{l}_1$};
		\draw[-] (be12-l) -- (l1-l);
		\coordinate (l2-l) at (5,4);
		\draw[-] (b-l) -- (l2-l) node[pos=0.5, right] {$\bar{l}_2$};
		\node (be21-l) at (5.5,4.5) {$e^b_{21}$};
		\draw[-] (be21-l) -- (l2-l);
		\coordinate (l3-l) at (6.5,3);
		\draw[-] (l3-l) -- (b-l) node[pos=0.5, above] {$\bar{l}_3$};
		\node (e32-l) at (7,2.5) {$e^b_{32}$};
		\draw[-] (e32-l) -- (l3-l);
		\node (be31-l) at (7,3.5) {$e^b_{31}$};
		\draw[-] (be31-l) -- (l3-l);
  \node (b-l) at (5,0) [circle, draw] {$b'$};
		\coordinate (l1-l) at (3.5,0);
		\node (bbe11-l) at (3,-0.5) {$e^{b'}_{11}$};
		\draw[-] (bbe11-l) -- (l1-l);
		\node (bbe12-l) at (3,0.5) {$e^{b'}_{12}$};
		\draw[-] (l1-l) -- (b-l) node[pos=0.5, above] {$\bar{l}_1$};
		\draw[-] (bbe12-l) -- (l1-l);
		\coordinate (l2-l) at (5,1);
		\draw[-] (b-l) -- (l2-l) node[pos=0.5, right] {$\bar{l}_2$};
		%
		\coordinate (l3-l) at (6.5,0);
		\draw[-] (l3-l) -- (b-l) node[pos=0.5, above] {$\bar{l}_3$};
		\node (bbe32-l) at (7,-0.5) {$e^{b'}_{32}$};
		\draw[-] (bbe32-l) -- (l3-l);
		\node (e31-l) at (7,0.5) {$e^{b'}_{31}$};
          \draw[-] (e31-l) -- (l3-l);

  
  \node (bIstar-upper) at (1.75,4.5) {$*$};
  \node (bIstar-middle) at (1.75, 2) {$*$};

  \node (bIstar-lower) at (1.75,-0.5) {$*$};

  \node (bOstar-upper) at (8.75,4.5) {$*$};
  \node (bOstar-lower) at (8.75,-0.5) {$*$};
  
\draw[->] (bIstar-upper)-- (be21-l);
\draw[->] (bIstar-middle)-- (bbe12-l);
\draw[->] (bIstar-lower)-- (bbe11-l);
\draw[->] (bOstar-upper)-- (be31-l);
\draw[->] (bOstar-lower)-- (bbe32-l);

\node at (1.75,-1.5) {$L(\lbrace 1 \rbrace \sqcup \lbrace 2,3 \rbrace)$};

\node at (8.75,-1.5) {$L(\lbrace 4 \rbrace \sqcup \lbrace 5 \rbrace)$};

\node at (3, 1.5) {$\mc{I} + \mc{I}'$};
\node at (7.5, 1.5) {$\mc{O} + \mc{O}'$};

	\end{tikzpicture}
	\end{center}
	\vspace{-0.3cm}
	\caption{An illustration of combining open process networks  using the monoidal product defined in Theorem \ref{ARMain Theorem1}, see Example \ref{Example:Combining open process networks using the monoidal product of horizontal 1-morphisms}. We combine the open process networks $(\mc{B}, \lbr 1 \rbr, \lbr 4 \rbr, \mc{I}, \mc{O})$ (top) and $(\mc{B}', \lbr 2,3 \rbr, \lbr 5 \rbr, \mc{I}', \mc{O}')$ (top) to obtain the open process network $(\mc{B} + \mc{B}', \lbr 1 \rbr \sqcup \lbr 2,3 \rbr, \lbr 4 \rbr \sqcup \lbr 5 \rbr, \mc{I} + \mc{I}', \mc{O} + \mc{O}')$ (below). }
	\label{fig:monoidal-product-of-open-process-network}
\end{figure}

\begin{example}[{2-morphism between two open process networks}]\label{Example:2-morphism-between two-open-process networks}
In Figure \ref{fig:2morphism}, we illustrate a 2-morphism $({\rm{id}_{\lbr 1 \rbr}} \colon \lbr 1 \rbr \ra \lbr 1\rbr, \eta \colon \mc{B} \ra \mc{B}', {\rm{id}_{\lbr 2,3 \rbr}} \colon \lbr 2,3 \rbr \ra \lbr 2,3 \rbr  )$ from the open process network $(\mc{B}, \lbr 1 \rbr, \lbr 2,3 \rbr, \mc{I}, \mc{O})$ to the open process network $(\mc{B}', \lbr 1 \rbr, \lbr 2,3 \rbr, \mc{I}', \mc{O}')$, where $\mc{B}=(E=\lbr e^{b}_{11}, e^{b}_{12}, e^{b}_{21}, e^{b}_{31}, e^{b}_{32} \rbr, B= \lbr b \rbr, \lbr l_i \rbr_{3} )$ and $\mc{B}'=( E'=\lbr e^{b}_{11}, e^{b}_{21}, e^{b}_{31} \rbr, B'=\lbr b \rbr, \lbr l'_i \rbr_{3} )$ are process networks. Here, $\eta: =(\alpha \colon E \ra E', \beta \colon B \ra B')$ is a morphism of process networks, where $\beta$ is the identity map on the set $\lbr b \rbr$, and $\alpha$ maps $e^b_{11} \mapsto e^{b}_{11}, e^b_{12} \mapsto e^b_{11}, e^b_{21} \mapsto e^b_{21}, e^b_{31} \mapsto e^b_{32}, e^b_{32} \mapsto e^b_{32}$, and hence we have $\mc{I}' \circ L({\rm{id}}_{\lbr 1 \rbr})= \eta \circ \mc{I}$ and $\mc{O}' \circ L({\rm{id}}_{\lbr 2,3 \rbr})= \eta \circ \mc{O}$.
\end{example}
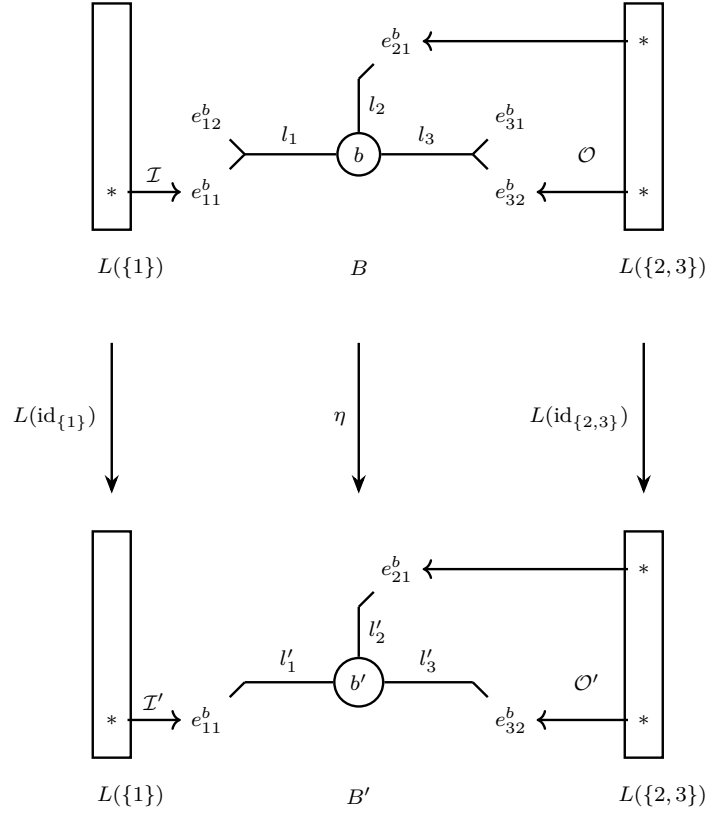
\begin{figure}[H]
	\begin{center}
	\begin{tikzpicture}[font={\footnotesize}, line width = 0.9pt]
		\node at (3.5,0.5)  {$B$};
        \node (L1) at (0.5,0.5) {$L(\lbrace 1 \rbrace)$};
		\node (b-l) at (3.5,2) [circle, draw] {$b$};
  \node at (7.5,0.5) {$L(\lbrace 2,3 \rbrace)$};
		\coordinate (l1-l) at (2,2);
		\node (e11-l) at (1.5,1.5) {$e^b_{11}$};
		\draw[-] (e11-l) -- (l1-l);
		\node (e12-l) at (1.5,2.5) {$e^b_{12}$};
		\draw[-] (l1-l) -- (b-l) node[pos=0.5, above] {$l_1$};
		\draw[-] (e12-l) -- (l1-l);
		\coordinate (l2-l) at (3.5,3);
		\draw[-] (b-l) -- (l2-l) node[pos=0.5, right] {$l_2$};
		\node (e21-l) at (4,3.5) {$e^b_{21}$};
		\draw[-] (e21-l) -- (l2-l);
		\coordinate (l3-l) at (5,2);
		\draw[-] (l3-l) -- (b-l) node[pos=0.5, above] {$l_3$};
		\node (e32-l) at (5.5,1.5) {$e^b_{32}$};
		\draw[-] (e32-l) -- (l3-l);
		\node (e31-l) at (5.5,2.5) {$e^b_{31}$};
		\draw[-] (e31-l) -- (l3-l);
		\draw[-] (0,1) -- (0,4) -- (0.5,4) -- (0.5,1) -- cycle;
		\node (Istar-ll) at (0.25,1.5) {$*$};
		\draw[->] (Istar-ll) -- (e11-l) node[pos=0.5, above] {$\mathcal{I}$};
		\draw[-] (7,1) -- (7,4) -- (7.5,4) -- (7.5,1) -- cycle;
		\node (Ostar-ll) at (7.25,1.5) {$*$};
		\draw[->] (Ostar-ll) -- (e32-l);
		\node (Ostar-lu) at (7.25,3.5) {$*$};
		\node at (6.5,2) {$\mathcal{O}$};
		\draw[->] (Ostar-lu) -- (e21-l);

  \node (b-l) at (3.5,-5) [circle, draw] {$b'$};

 \coordinate (l1-l) at (2,-5);
		\node (e11-l) at (1.5,-5.5) {$e^b_{11}$};
		\draw[-] (e11-l) -- (l1-l);
		\draw[-] (l1-l) -- (b-l) node[pos=0.5, above] {$l'_1$};
		%
		\coordinate (l2-l) at (3.5,-4);
		\draw[-] (b-l) -- (l2-l) node[pos=0.5, right] {$l'_2$};
		\node (e21-l) at (4,-3.5) {$e^b_{21}$};
		\draw[-] (e21-l) -- (l2-l);
		\coordinate (l3-l) at (5,-5);
		\draw[-] (l3-l) -- (b-l) node[pos=0.5, above] {$l'_3$};
		\node (e32-l) at (5.5,-5.5) {$e^b_{32}$};
		\draw[-] (e32-l) -- (l3-l);
		%
		\draw[-] (0,-6) -- (0,-3) -- (0.5,-3) -- (0.5,-6) -- cycle;
		\node (Istar-ll) at (0.25,-5.5) {$*$};
		\draw[->] (Istar-ll) -- (e11-l) node[pos=0.5, above] {$\mathcal{I}'$};
		\draw[-] (7,-6) -- (7,-3) -- (7.5,-3) -- (7.5,-6) -- cycle;
		\node (Ostar-ll) at (7.25,-5.5) {$*$};
		\draw[->] (Ostar-ll) -- (e32-l);
		\node (Ostar-lu) at (7.25,-3.5) {$*$};
		\node at (6.5,-5) {$\mathcal{O}'$};
		\draw[->] (Ostar-lu) -- (e21-l);
      \node at (7.5,-6.5) {$L(\lbrace 2,3 \rbrace)$};
       \node at (0.5,-6.5) {$L(\lbrace 1 \rbrace)$};
       \node at (3.5,-6.5)  {$B'$};
       \draw[-Stealth] (.25,-.5) -- (0.25,-2.5);
       \draw[-Stealth] (3.50,-.5) -- (3.50,-2.5);
        \draw[-Stealth] (7.25,-.5) -- (7.25,-2.5);
        \node at (-0.5,-1.5)  {$L({\rm{id}}_{\lbrace 1 \rbrace})$};
        \node at (3.25,-1.5)  {$\eta$};
         \node at (6.40,-1.5)  {$L({\rm{id}}_{\lbrace 2,3 \rbrace})$};
	\end{tikzpicture}
	\end{center}
	\vspace{-0.3cm}
	\caption{An illustration of 2-morphism between two open process networks, see Example \ref{Example:2-morphism-between two-open-process networks}.}
	\label{fig:2morphism}
\end{figure}

\begin{example}[{Horizontal composition of 2-morphisms}]\label{Example:horizontal-composition-of-2-morphisms}
In Figure \ref{fig:horizontal-composition-of-2-morphisms} and Figure \ref{fig:part2composition-of-open-process-network}, we illustrate the horizontal composition of the 2-morphism $({\rm{id}}_{\lbr 2,3 \rbr},\eta_2, {\rm{id}}_{\lbr 4 \rbr} )$ from the open process network $(\mc{B}_2, \lbr 2,3 \rbr, \lbr 4 \rbr, \mc{I}_2, \mc{O}_2)$ to the open process network  $(\mc{B}'_2, \lbr 2,3 \rbr, \lbr 4 \rbr, \mc{I}'_2, \mc{O}'_2)$, and the 2-morphism $({\rm{id}}_{\lbr 1 \rbr},\eta_1, {\rm{id}}_{\lbr 2,3 \rbr} )$  from the open process network $(\mc{B}_1, \lbr 1 \rbr, \lbr 2,3 \rbr, \mc{I}_1, \mc{O}_1)$ to the open process network $(\mc{B}_1, \lbr 1 \rbr, \lbr 2,3 \rbr, \mc{I}'_1, \mc{O}'_1)$. 2-morphisms $\eta_2$ and $\eta_1$ are shown in Figure \ref{fig:horizontal-composition-of-2-morphisms}, and  their horizontal composition $\eta_2 +_{L({\rm{id}}_{\lbr 2,3 \rbr})} \eta_1$ is shown in Figure \ref{fig:part2composition-of-open-process-network}. Note that here the horizontal composition of open process networks identifies the entities $e^{b_1}_{21}$ and $e^{b_2}_{21}$ into the entity $e$, and identifies the entities $e^{b_1}_{31}$ and $e^{b_2}_{11}$ into $d$.
\end{example}

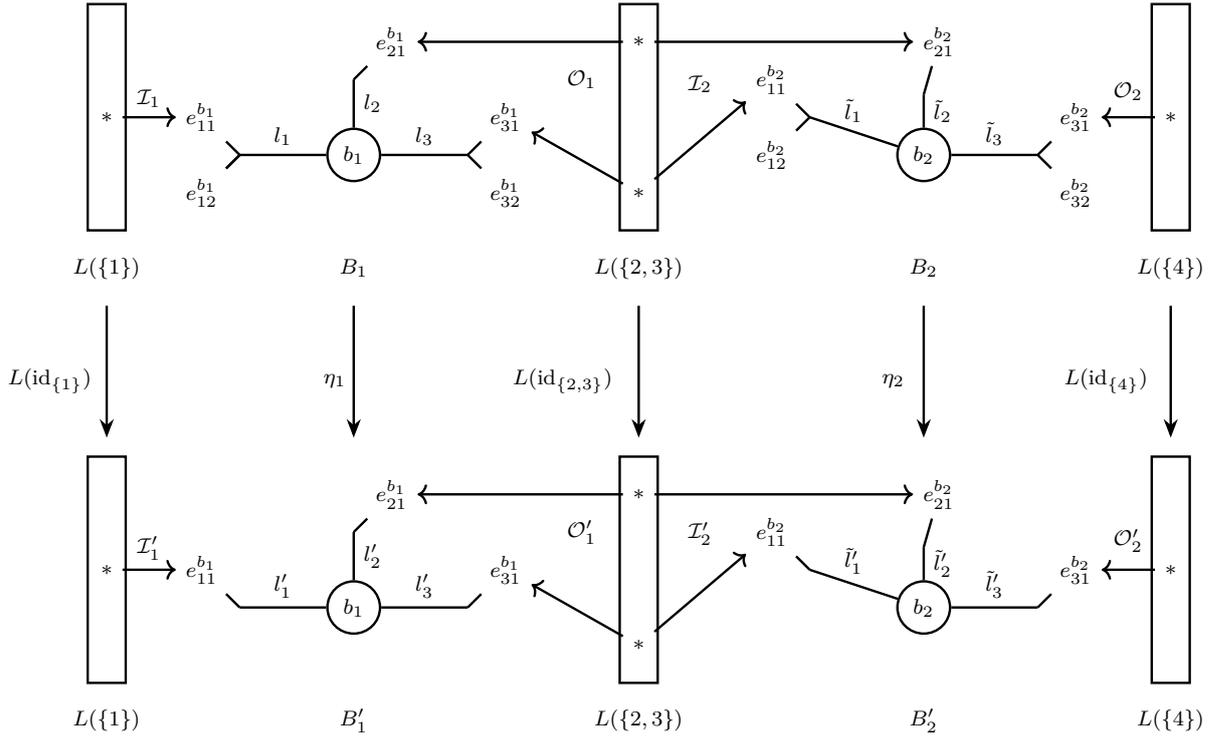
\begin{figure}[H]
	\begin{center}
	\begin{tikzpicture}[font={\footnotesize}, line width = 0.9pt]
		%
		%
		\draw[-] (0,7) -- (0,10) -- (0.5,10) -- (0.5,7) -- cycle;
		\node (Istar-ul) at (0.25,8.5) {$*$};
		\node () at (0.25,6.5) {$L(\{1\})$};
		\node (e11-ul) at (1.5,7.5) {$e^{b_1}_{12}$};
		\coordinate (l1-ul) at (2,8);
		\node (b-ul) at (3.5,8) [circle, draw] {$b_1$};
		\draw[-] (e11-ul) -- (l1-ul);
		\node (e12-ul) at (1.5,8.5) {$e^{b_1}_{11}$};
  \draw[->] (Istar-ul) -- (e12-ul) node[pos=0.5, above] {$\mathcal{I}_1$};
		\draw[-] (l1-ul) -- (b-ul) node[pos=0.5, above] {$l_1$};
		\draw[-] (e12-ul) -- (l1-ul);
		\draw[-] (e11-ul) -- (l1-ul);
		\coordinate (l2-ul) at (3.5,9);
		\draw[-] (b-ul) -- (l2-ul) node[pos=0.5, right] {$l_2$};
		\node (e-ul) at (4,9.5) {$e^{b_1}_{21}$};
		\draw[-] (e-ul) -- (l2-ul);
		\coordinate (l3-ul) at (5,8);
		\draw[-] (l3-ul) -- (b-ul) node[pos=0.5, above] {$l_3$};
		\node (e32-ul) at (5.5,7.5) {$e^{b_{1}}_{32}$};
		\draw[-] (e32-ul) -- (l3-ul);
		\node (e31-ul) at (5.5,8.5) {$e^{b_{1}}_{31}$};
		\draw[-] (e31-ul) -- (l3-ul);
		%
		%
		\draw[-] (7,7) -- (7,10) -- (7.5,10) -- (7.5,7) -- cycle;
		\node () at (7.25,6.5) {$L(\{2,3\})$};
		\node (Ostar-uru) at (7.25,7.5) {$*$};
		\node (Ostar-url) at (7.25,9.5) {$*$};
		\draw[->] (Ostar-url) -- (e-ul);
		\node at (6.5,9) {$\mathcal{O}_1$};
		\draw[->] (Ostar-uru) -- (e31-ul);
		\node (e11-ur) at (9,8) {$e^{b_2}_{12}$};
  \node (e12-ur) at (9,9) {$e^{b_2}_{11}$};
		\draw[->] (Ostar-uru) -- (e12-ur) node[pos=0.5, above, yshift=5mm] {$\mathcal{I}_2$};
  \node (e21-ul) at (11.2,9.5) {$e^{b_2}_{21}$};
		\draw[->] (Ostar-url) -- (e21-ul);
		\coordinate (l1-ur) at (9.5,8.5);
		\draw[-] (l1-ur) -- (e11-ur);
		\draw[-] (l1-ur) -- (e12-ur);
		\node (bp-l) at (11,8) [circle, draw] {$b_2$};
		\draw[-] (l1-ur) -- (bp-l) node[pos=0.5, above] {$\tilde{l}_1$};
		\coordinate (l2b-ur) at (11,8.8);
		\draw[-] (bp-l) -- (l2b-ur) node[pos=0.5, right] {$\tilde{l}_2$};
		\coordinate (l3p-ur) at (12.5,8);
		\draw[-] (bp-l) -- (l3p-ur) node[pos=0.5, above] {$\tilde{l}_3$};
		%
		%
		\node (e31p-ur) at (13,8.5) {$e^{b_2}_{31}$};
		\draw[-] (l3p-ur) -- (e31p-ur);
		\node (e32p-ur) at (13,7.5) {$e^{b_2}_{32}$};
		\draw[-] (l3p-ur) -- (e32p-ur);
		\draw[-] (14,7) -- (14,10) -- (14.5,10) -- (14.5,7) -- cycle;
		\node (Ost-uru) at (14.25,8.5) {$*$};
		\draw[->] (Ost-uru) -- (e31p-ur) node[pos=0.5, above, yshift=1mm] {$\mathcal{O}_2$};
		\node  at (14.25,6.5) {$L(\{4\})$};
		\draw[-] (e21-ul) -- (l2b-ur);

\draw[-] (0,1) -- (0,4) -- (0.5,4) -- (0.5,1) -- cycle;
		\node (Istar-ul) at (0.25,2.5) {$*$};
		\node () at (0.25,0.5) {$L(\{1\})$};
		%
		\coordinate (l1-ul) at (2,2);
		\node (b-ul) at (3.5,2) [circle, draw] {$b_1$};
		\node (e012-ul) at (1.5,2.5) {$e^{b_1}_{11}$};
		\draw[-] (l1-ul) -- (b-ul) node[pos=0.5, above] {$l'_1$};
		\draw[-] (e012-ul) -- (l1-ul);
  \draw[->] (Istar-ul) -- (e012-ul) node[pos=0.5, above] {$\mathcal{I}'_1$};
		\coordinate (l2-ul) at (3.5,3);
		\draw[-] (b-ul) -- (l2-ul) node[pos=0.5, right] {$l'_2$};
		\node (e-ul) at (4,3.5) {$e^{b_1}_{21}$};
		\draw[-] (e-ul) -- (l2-ul);
		\coordinate (l3-ul) at (5,2);
		\draw[-] (l3-ul) -- (b-ul) node[pos=0.5, above] {$l'_3$};
		\node (e31-ul) at (5.5,2.5) {$e^{b_{1}}_{31}$};
		\draw[-] (e31-ul) -- (l3-ul);
		%
		%
		\draw[-] (7,1) -- (7,4) -- (7.5,4) -- (7.5,1) -- cycle;
		\node () at (7.25,0.5) {$L(\{2,3\})$};
		\node (Ostar-uru) at (7.25,1.5) {$*$};
		\node (Ostar-url) at (7.25,3.5) {$*$};
		\draw[->] (Ostar-url) -- (e-ul);
		\node at (6.5,3) {$\mathcal{O}'_1$};
		\draw[->] (Ostar-uru) -- (e31-ul);
		%
		%
  \node (e12-ur) at (9,3) {$e^{b_2}_{11}$};
		\draw[->] (Ostar-uru) -- (e12-ur) node[pos=0.5, above, yshift=5mm] {$\mathcal{I}'_2$};
  \node (e21-ul) at (11.2,3.5) {$e^{b_2}_{21}$};
		\draw[->] (Ostar-url) -- (e21-ul);
		\coordinate (l1-ur) at (9.5,2.5);
		\draw[-] (l1-ur) -- (e12-ur);
		\node (bp-l) at (11,2) [circle, draw] {$b_2$};
		\draw[-] (l1-ur) -- (bp-l) node[pos=0.5, above] {$\tilde{l}'_1$};
		\coordinate (l2b-ur) at (11,2.8);
		\draw[-] (bp-l) -- (l2b-ur) node[pos=0.5, right] {$\tilde{l}'_2$};
		\coordinate (l3p-ur) at (12.5,2);
		\draw[-] (bp-l) -- (l3p-ur) node[pos=0.5, above] {$\tilde{l}'_3$};
		%
		%
		\node (e31p-ur) at (13,2.5) {$e^{b_2}_{31}$};
		\draw[-] (l3p-ur) -- (e31p-ur);
		%
		\draw[-] (14,1) -- (14,4) -- (14.5,4) -- (14.5,1) -- cycle;
		\node (Ost-uru) at (14.25,2.5) {$*$};
		\draw[->] (Ost-uru) -- (e31p-ur) node[pos=0.5, above, yshift=1mm] {$\mathcal{O}'_2$};
		\node  at (14.25,0.5) {$L(\{4\})$};
		\draw[-] (e21-ul) -- (l2b-ur);
   \draw[-Stealth] (.25, 6) -- (0.25, 4.25);
   \draw[-Stealth] (7.25, 6) -- (7.25, 4.25);
   \draw[-Stealth] (14.25, 6) -- (14.25, 4.25);
   \node at (3.5,0.5) {$B'_1$};
   \node at (3.5,6.5) {$B_1$};
   \node at (11,0.5) {$B'_2$};
    \draw[-Stealth] (3.5, 6) -- (3.5, 4.25);
    \node at (11,6.5) {$B_2$};
    \draw[-Stealth] (11, 6) -- (11, 4.25);
    \node at (-.5, 5) {$L({\rm{id}_{ \lbrace 1 \rbrace}})$};
    \node at (3.25,5) {$\eta_1$};
    \node at (6.25,5) {$L({\rm{id}_{ \lbrace 2,3 \rbrace}})$};
    \node at (10.6,5) {$\eta_2$};
     \node at (13.4,5) {$L({\rm{id}_{ \lbrace 4\rbrace}})$};

	\end{tikzpicture}
	\end{center}
	\vspace{-0.3cm}
	\caption{An illustration showing 2-moprhisms between composable horizontal 1-morphisms considered in Example \ref{Example:horizontal-composition-of-2-morphisms}.}
	\label{fig:horizontal-composition-of-2-morphisms}
\end{figure}

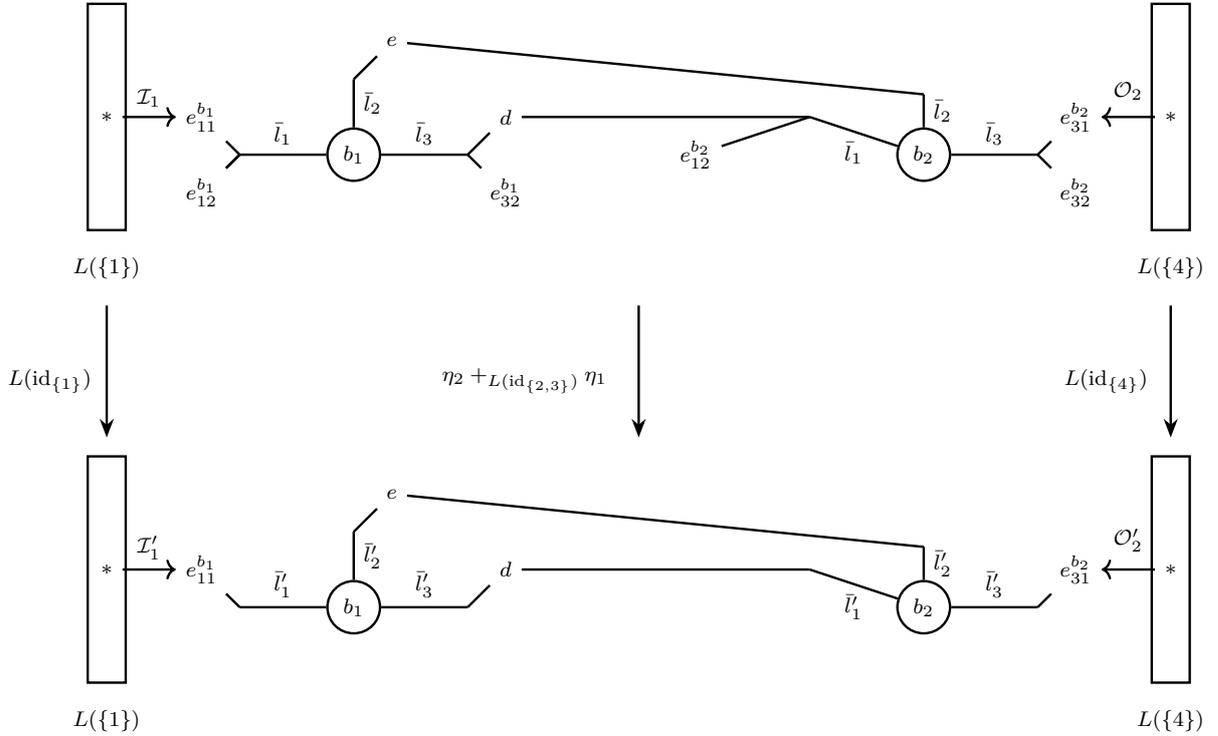
\begin{figure}[H]
	\begin{center}
	\begin{tikzpicture}[font={\footnotesize}, line width = 0.9pt]
		%
		%
		\draw[-] (0,7) -- (0,10) -- (0.5,10) -- (0.5,7) -- cycle;
		\node (Istar-ul) at (0.25,8.5) {$*$};
		\node () at (0.25,6.5) {$L(\{1\})$};
		\node (e11-ul) at (1.5,7.5) {$e^{b_1}_{12}$};
		\coordinate (l1-ul) at (2,8);
		\node (b-ul) at (3.5,8) [circle, draw] {$b_1$};
		\draw[-] (e11-ul) -- (l1-ul);
		\node (e12-ul) at (1.5,8.5) {$e^{b_1}_{11}$};
  \draw[->] (Istar-ul) -- (e12-ul) node[pos=0.5, above] {$\mathcal{I}_1$};
		\draw[-] (l1-ul) -- (b-ul) node[pos=0.5, above] {$\bar{l}_1$};
		\draw[-] (e12-ul) -- (l1-ul);
		\draw[-] (e11-ul) -- (l1-ul);
		\coordinate (l2-ul) at (3.5,9);
		\draw[-] (b-ul) -- (l2-ul) node[pos=0.5, right] {$\bar{l}_2$};
		\node (e-ul) at (4,9.5) {$e$};
		\draw[-] (e-ul) -- (l2-ul);
		\coordinate (l3-ul) at (5,8);
		\draw[-] (l3-ul) -- (b-ul) node[pos=0.5, above] {$\bar{l}_3$};
		\node (e32-ul) at (5.5,7.5) {$e^{b_{1}}_{32}$};
		\draw[-] (e32-ul) -- (l3-ul);
		\node (e31-ul) at (5.5,8.5) {$d$};
		\draw[-] (e31-ul) -- (l3-ul);
		%
		%
		%
		%
		\node (e11-ur) at (8,8) {$e^{b_2}_{12}$};
		%
		\coordinate (l1-ur) at (9.5,8.5);
		\draw[-] (l1-ur) -- (e11-ur);
		\draw[-] (l1-ur) -- (e31-ul);
		\node (bp-l) at (11,8) [circle, draw] {$b_2$};
		\draw[-] (l1-ur) -- (bp-l) node[pos=0.5, below] {$\bar{l}_1$};
		\coordinate (l2b-ur) at (11,8.8);
		\draw[-] (bp-l) -- (l2b-ur) node[pos=0.5, right] {$\bar{l}_2$};
		\coordinate (l3p-ur) at (12.5,8);
		\draw[-] (bp-l) -- (l3p-ur) node[pos=0.5, above] {$\bar{l}_3$};
		\node (e31p-ur) at (13,8.5) {$e^{b_2}_{31}$};
		\draw[-] (l3p-ur) -- (e31p-ur);
		\node (e32p-ur) at (13,7.5) {$e^{b_2}_{32}$};
		\draw[-] (l3p-ur) -- (e32p-ur);
		\draw[-] (14,7) -- (14,10) -- (14.5,10) -- (14.5,7) -- cycle;
		\node (Ost-uru) at (14.25,8.5) {$*$};
		\draw[->] (Ost-uru) -- (e31p-ur) node[pos=0.5, above, yshift=1mm] {$\mathcal{O}_2$};
		\node  at (14.25,6.5) {$L(\{4\})$};
		\draw[-] (e-ul) -- (l2b-ur);

	\draw[-] (0,1) -- (0,4) -- (0.5,4) -- (0.5,1) -- cycle;
		\node (Istar-ul) at (0.25,2.5) {$*$};
		\node () at (0.25,0.5) {$L(\{1\})$};
		%
		\coordinate (l1-ul) at (2,2);
		\node (b-ul) at (3.5,2) [circle, draw] {$b_1$};
		\node (e12-ul) at (1.5,2.5) {$e^{b_1}_{11}$};
  \draw[->] (Istar-ul) -- (e12-ul) node[pos=0.5, above] {$\mathcal{I}'_1$};
		\draw[-] (l1-ul) -- (b-ul) node[pos=0.5, above] {$\bar{l}'_1$};
		\draw[-] (e12-ul) -- (l1-ul);
		%
		\coordinate (l2-ul) at (3.5,3);
		\draw[-] (b-ul) -- (l2-ul) node[pos=0.5, right] {$\bar{l}'_2$};
		\node (e-ul) at (4,3.5) {$e$};
		\draw[-] (e-ul) -- (l2-ul);
		\coordinate (l3-ul) at (5,2);
		\draw[-] (l3-ul) -- (b-ul) node[pos=0.5, above] {$\bar{l}'_3$};
		\node (e31-ul) at (5.5,2.5) {$d$};
		\draw[-] (e31-ul) -- (l3-ul);
		%
		%
		%
		%
		%
		\coordinate (l1-ur) at (9.5,2.5);
		\draw[-] (l1-ur) -- (e31-ul);
		\node (bp-l) at (11,2) [circle, draw] {$b_2$};
		\draw[-] (l1-ur) -- (bp-l) node[pos=0.5, below] {$\bar{l}'_1$};
		\coordinate (l2b-ur) at (11,2.8);
		\draw[-] (bp-l) -- (l2b-ur) node[pos=0.5, right] {$\bar{l}'_2$};
		\coordinate (l3p-ur) at (12.5,2);
		\draw[-] (bp-l) -- (l3p-ur) node[pos=0.5, above] {$\bar{l}'_3$};
		\node (e31p-ur) at (13,2.5) {$e^{b_2}_{31}$};
		\draw[-] (l3p-ur) -- (e31p-ur);
		%
		\draw[-] (14,1) -- (14,4) -- (14.5,4) -- (14.5,1) -- cycle;
		\node (Ost-uru) at (14.25,2.5) {$*$};
		\draw[->] (Ost-uru) -- (e31p-ur) node[pos=0.5, above, yshift=1mm] {$\mathcal{O}'_2$};
		\node  at (14.25,0.5) {$L(\{4\})$};
		\draw[-] (e-ul) -- (l2b-ur);


   \draw[-Stealth] (.25, 6) -- (0.25, 4.25);
   \draw[-Stealth] (7.25, 6) -- (7.25, 4.25);
   \draw[-Stealth] (14.25, 6) -- (14.25, 4.25);
    \node at (-.5, 5) {$L({\rm{id}_{ \lbrace 1 \rbrace}})$};
    \node at (5.75,5) {$\eta_2 +_{L({\rm{id}}_{\lbrace 2,3 \rbrace})} \eta_1$};
     \node at (13.4,5) {$L({\rm{id}_{ \lbrace 4\rbrace}})$};


	\end{tikzpicture}
	\end{center}
	\vspace{-0.3cm}
	\caption{An illustration of horizontal composition of the 2-morphisms as shown in Example \ref{Example:horizontal-composition-of-2-morphisms}.}
	\label{fig:part2composition-of-open-process-network}
\end{figure}


\begin{example}[{Compatibility  between compositionality and zooming-out process}]\label{Example:Compatibility  between compositionality and zooming-out process}
 Using SBGN-PD visualisations, in Figure \ref{fig:composition-and zooming}, we demonstrate how the horizontal composition law of 2-morphisms as in Theorem \ref{ARMain Theorem1} (see Example \ref{Example:horizontal-composition-of-2-morphisms}) induces a compatibility condition between compositionality and zooming-out process in a system of biochemical reaction networks. In particular, we start with two biochemical reactions numbered (1) and (2). Then, using the 2-morphisms as shown in Figure \ref{fig:horizontal-composition-of-2-morphisms},  we zoom-out (shown with the thin black arrows) by forgetting ADP's and ATP's from the reactions, and we obtain the biochemical reactions (3) from (1), and (4) from (2). Then, using the horizontal composition of 1-morphisms (shown with dotted black lines) as in Figure \ref{fig:part2composition-of-open-process-network}, we combine (3) and (4)  to obtain the reaction network (6), and the reactions (1) and (2) to get reaction network (5). The horizontal composition of two 2-morphisms (as in Figure \ref{fig:part2composition-of-open-process-network}) provides us with a canonical way of combining two zooming-out procedures $\big($(1) to (3) and (2) to (4)$ \big)$ such that the combined zoom-out preocedure behaves well with the composition of process networks. To be more precise, the combined zoom-out process takes the reaction network (5) to the reaction network (6). 
\end{example}

\begin{figure}[H]
\begin{center}
\includegraphics[width=15cm]{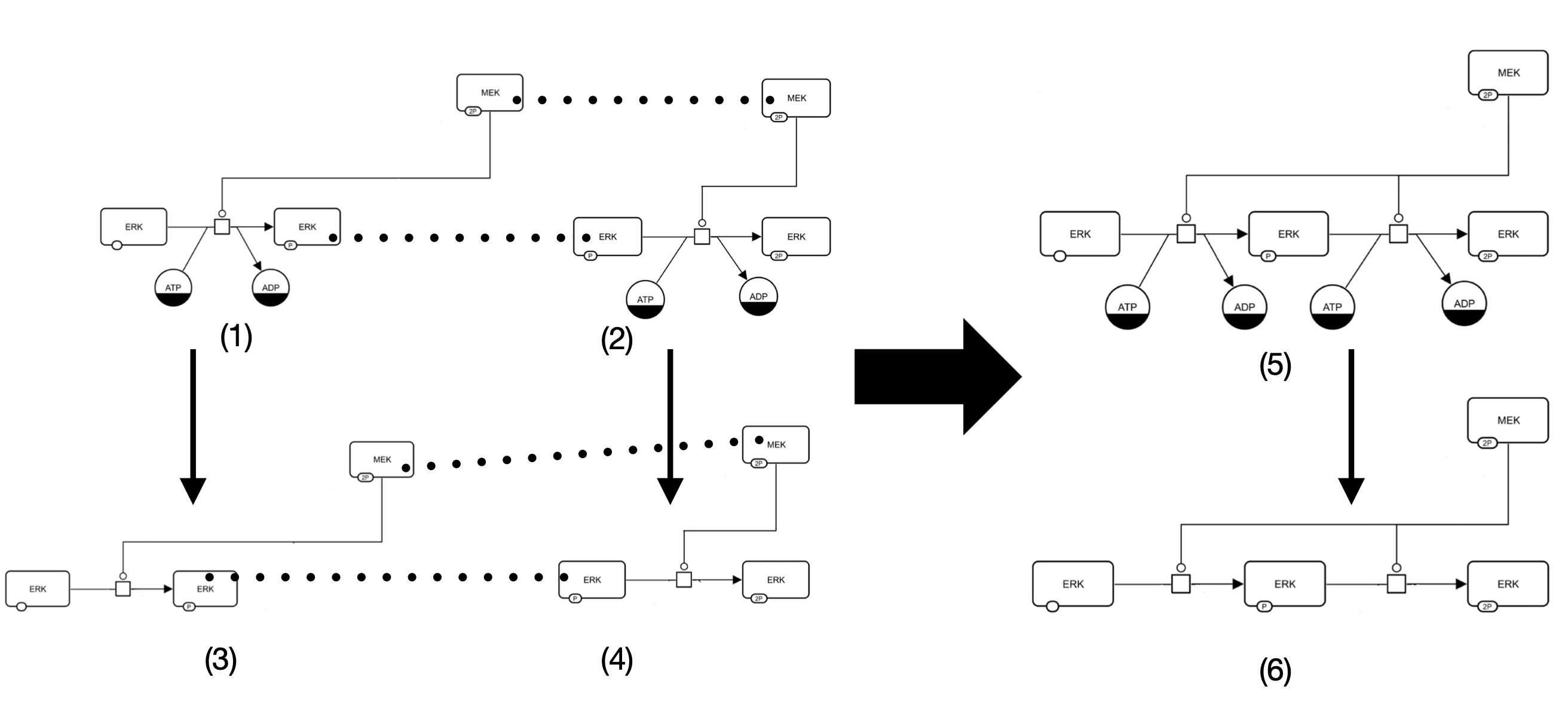}
\caption{An illustration of a compatibility condition between compositionality and zooming-out process in a system of biochemical reaction networks viusalised in SBGN-PD (see Example \ref{Example:Compatibility  between compositionality and zooming-out process}). SBGN images are derived from the MAPK cascade example on Page 65, \cite{rougny_systems_2019}.}
\label{fig:composition-and zooming}
\end{center}
\end{figure}

In the Example \ref{Example:horizontal-composition-of-2-morphisms}, we explained how one can interpret the horizontal composition of 2-morphisms as a compatiblity condition between the zooming-out procedures and composition laws of process networks. In the same spirit, one can interpret the monoidal product of 2-morphisms as a compatiblity condition between the zooming-out procedures and the monoidal product of process networks.

 \subsection{Macroscope of a process network}\label{Subsection:Macroscope of a process network: A tool to study the influence of a process network on its environment and vice versa}
This subsection builds a mathematical gadget by the name \textit{macroscope}, which provides us with a formal language to describe the influence of a particular portion of a biochemical network on the remaining portion and vice versa. To achieve our said purpose, we introduce a few technical notions below. 
\begin{definition}[Process subnetwork of a process network]\label{Defn:Processsubnetwork}
A \textit{process subnetwork of a process network} $\mc{B}= (E, B, \lbr l_i \rbr_{n})$ is a process network $\mc{B}'= (E', B', \lbr l'_i \rbr_{n})$ such that $E' \subseteq E$ and $B' \subseteq B$ and $l'_i(b)=l_i(b)$ for all $b \in B'$ and $i \in \lbr 1,2, \ldots, n \rbr$.
\end{definition}

\begin{definition}\label{Defn:Env}[Environment of a process subnetwork]
Let $\mc{B}'= (E', B', \lbr l'_i \rbr_{n})$ be a process subnetwork of a process network $\mc{B}= (E, B, \lbr l_i \rbr_{n})$. Then the \textit{environment of $\mc{B}'$ with respect to $\mc{B}$} is defined as the subset $B-B'$.

\end{definition}

\begin{remark}
Recall, in Section \ref{Sec:1}, we discussed   how the SBGN-PD visualisation of MAPK cascade can be composed with two other biochemical pathways to build IGF signalling pathways (Figure \ref{fig:construction-of-insulinsignalling}), and how to hide the details of MAPK cascade by using an encapsulation node called the submap in the IGF signalling SBGN-PD (Figure \ref{fig:cascade-as-submap}). In this context, our process subnetwork of a process network (Definition \ref{ARDefinition: Biochemical process networks}) is similar to how a submap of a SBGN-PD is to the SBGN-PD. With this analogy, reactions in the two  biochemical pathways in Figure \ref{fig:construction-of-insulinsignalling} (marked in blue and grey) form the environment (Definition \ref{Defn:Env}) of the MAPK cascade SBGN-PD (treated as a subprocess network) of the process network description of the IGF signalling pathway.

\end{remark}

\begin{definition}\label{Defn:M-function-of-a-process-network}
Let $\mc{B}'= (E', B', \lbr l'_i \rbr_{n})$ be a process subnetwork of a process network $\mc{B}= (E, B, \lbr l_i \rbr_{n})$. Then for each $i \in \lbr 1,2, \ldots, n-1 \rbr$, using the notations as in the  Remark \ref{ARRemark: Relevant entities}, we define the \textit{$M^{\mc{B}, \mc{B}'}_{i}$-function of $\mc{B}'$} as follows:
\begin{alignat*}{3}
M^{\mc{B}, \mc{B}'}_{i} \colon & E' \ra &&\mathbb{B}[E']\\
& e \mapsto e, \, \, &&\text{if there exists}\,\,  b\in B-B' \,\, \text{and}\,\, b' \in B' \text{such that}\,\, e=e^{b}_{nj}= e^{b'}_{ij'}\,\\
&\,&&  \text{for some}\,\, j \in \lbr 1,2, \ldots, m_{b}^{n} \rbr \,\, \text{and}\,\,
j' \in \lbr 1,2, \ldots, m_{b'}^{i} \rbr,\\
& e \mapsto 0, \,\, &&\text{otherwise.}
\end{alignat*}
and, when $i=n$, we have 
\begin{alignat*}{3}
M^{\mc{B},\mc{B}'}_{n} \colon & E' \ra &&\mathbb{B}[E']\\ 
& e \mapsto e, \, \, &&\text{if there exists}\,\, b \in B-B' \,\, \text{and}\,\, b' \in B' \text{such that} \,\, e=e^{b}_{kj}= e^{b'}_{nj'},\,\\
&\,&& \text{for some}\,\, k \in \lbr 1,2, \ldots, n-1 \rbr, j \in \lbr 1,2, \ldots, m_{b}^{k} \rbr, j' \in \lbr 1,2, \ldots, m_{b'}^{n} \rbr,\\
& e \mapsto &&0, \, \, \text{otherwise.}
\end{alignat*}
\end{definition}

\begin{remark}\label{rmk-macro-novel}

To make sense of the definition of $M^{\mc{B}, \mc{B}'}_{i}$ function (Definition \ref{Defn:M-function-of-a-process-network}) in the context of SBGN-PD, let us consider the case of $n=4$. Let $\mc{B}=(E, B, \lbr l_{i} \rbr_{4})$ be a process network with 4 legs.  Let $\mc{B}'=(E', B', \lbr l'_{i} \rbr_{4})$ be a sub-process network of $\mc{B}$. 

Let 
\begin{itemize}
\item $l_1$ denotes the \textit{input leg}, i.e. for each $b \in B$, $l_1(b) \in \mathbb{B}[E]$ contains the information of the substrates that goes into the process species (reaction) $b$.
\item $l_2$ denotes the \textit{activation leg}, i.e. for each $b \in B$, $l_2(b) \in \mathbb{B}[E]$ contains the information of the biomolecules  that activate (stimulate) the reaction $b$.
\item $l_3$ denotes the \textit{inhibition leg}, i.e. for each $b \in B$, $l_3(b) \in \mathbb{B}[E]$ contains the information of the biomolecules  that inhibit the reaction $b$.
\item $l_4$ denotes the \textit{production leg}, i.e. for each $b \in B$, $l_4(b) \in \mathbb{B}[E]$ contains the information of the metabolites produced in the reaction $b$.
\end{itemize}
Then, from the Definition \ref{Defn:M-function-of-a-process-network} itself, it is evident that
\begin{itemize}
\item $M^{\mc{B}, \mc{B}'}_{1}$ picks out the biomolecules produced by some reactions in the environment of $B'$ such that they act as substrates to some reactions in $B'$,
\item $M^{\mc{B}, \mc{B}'}_{2}$ picks out the biomolecules produced by some reactions in the environment of $B'$ such that they act as activators to some reactions in $B'$,
\item $M^{\mc{B}, \mc{B}'}_{3}$ picks out the biomolecules produced by some reactions in the environment of $B'$ such that they act as inhibitors to some reactions in $B'$,
\item $M^{\mc{B}, \mc{B}'}_{4}$ picks out the biomolecules produced by some reactions in $B'$ such that they \textit{act} (eg. as substrates or/and activators or/and inhibitors) on some reactions in the environment of $B'$.

\end{itemize}
\end{remark}


Given a process network $\mc{B}= (E, B, \lbr l_i \rbr_{n})$, for each $i \in \lbr 1,2, \ldots, ,n-1,n \rbr$, one can extend the function $M^{\mc{B}, \mc{B}'}_{i} \colon E' \ra \mathbb{B}[E']$ to a function 
\begin{equation*}
\begin{split}
\overline{M^{\mc{B},\mc{B}'}_{i}} \colon \mathbb{B}[E']  & \ra \mathbb{B}[E']\\
 (\alpha_1e_1 + \alpha_2e_2 + \ldots + \alpha_ne_n) & \mapsto \big(\alpha_1 M^{\mc{B}, \mc{B}'}_{i}(e_1) + \alpha_2M^{\mc{B}, \mc{B}'}_{i}(e_2)+ \ldots + \alpha_nM^{\mc{B}, \mc{B}'}_{i}(e_n)\big).
\end{split}
\end{equation*}

Our next definition provides us with a way to collect entities in a process network which performs \textit{similar functions}.

\begin{definition}\label{Defn:Sigma-B-function}
    Given a process network $\mc{B}= (E, B, \lbr l_i \rbr_{n})$, we define the \textit{$\Sigma_{\mc{B}}$-function of the process network $\mc{B}$} as follows:
    \begin{equation*}
    \begin{split}
    \Sigma_{\mc{B}} \colon & \mc{L} \ra \mathbb{B}[E]\\
    & l_i \mapsto \sum_{b \in B} \big(l_{i}(b) \big), \,\, \text{for each}\,\,i \in \lbr1,2,\ldots, n \rbr,
    \end{split} 
\end{equation*}
where $\mc{L}$ is the set $\lbr l_1, l_2, \ldots, l_n \rbr$.
\end{definition}

\begin{remark}

For the case of $n=4$ (as considered in Remark \ref{rmk-macro-novel}), Definition \ref{Defn:Sigma-B-function}  provides us with a way to collect entities which performs similar functions (eg. acting as substrates or metabolites or activators or inhibitors) to some reactions in the process network $\mc{B}= (E, B, \lbr l_i \rbr_{4})$.
\end{remark}

We are now ready to introduce the main notion of this section, that we call the \textit{macroscope of a process network}.
\begin{definition}[Macroscope of a process subnetwork]\label{Defn:macroscope}
Let $\mc{B}'= (E', B', \lbr l'_i \rbr_{n})$ be a process subnetwork of a process network $\mc{B}= (E, B, \lbr l_i \rbr_{n})$. Then, the \textit{macroscope of $\mc{B}'$ with respect to $\mc{B}$} is defined as the $n$-tuple $M(\mc{B}, \mc{B}'):= (L^{\mc{B}, \mc{B}'}_1, L^{\mc{B}, \mc{B}'}_2, \ldots, L^{\mc{B}, \mc{B}'}_n) \in \bar{\mathbb{Z}}_2[E'] \times \bar{\mathbb{Z}}_2[E'] \times \cdots\times \bar{\mathbb{Z}}_2[E']$, where $$L^{\mc{B}, \mc{B}'}_i:=\overline{M^{\mc{B}, \mc{B}'}_{i}} \big(\Sigma_{\mc{B}'}(l_i)\big),\,\, \forall \, \, i \in \lbr 1,2, \ldots, n \rbr.$$ For each $i \in \lbr 1,2, \ldots, n-1 \rbr$, the $i$-th coordinate $L^{\mc{B}, \mc{B}'}_i$  is called the $l_i$-\textit{influence of the environment on $\mc{B}'$}, and the $n$-th coordinate $L^{\mc{B}, \mc{B}'}_n$ is called the \textit{influence of $\mc{B}'$ on its environment}.
\end{definition}
\begin{remark}\label{rmk:Novelty-of-macroscope}
Observe that our notion of macroscope $M(\mc{B, \mc{B}'})$ of a process subnetwork $\mc{B}'$ of a process network $\mc{B}$ captures more information than just identifying the interface of $\mc{B}'$ with respect to $\mc{B}$. By the virtue of Definition \ref{Defn:M-function-of-a-process-network} (as explained in Remark \ref{rmk-macro-novel}), 
\begin{itemize}
\item the  macroscope $M(\mc{B, \mc{B}'})$ explicitly specifies the entities produced by some  reactions in $\mc{B}'$ which act (eg. substrates, activators or inhibitors for the case of $n=4$) on the reactions present outside (environment) of $\mc{B}'$,
\item the  macroscope $M(\mc{B, \mc{B}'})$ explicitly specifies the entities produced outside (environment) of $\mc{B}'$ and \textit{the way they are acting} (eg. substrates, activators or inhibitors for the case of $n=4$) on  some  reactions in $\mc{B}'$.

\end{itemize}

\end{remark}

\begin{remark}
Using the notations as in Definition \ref{Defn:macroscope}, note that when $\mc{B}=\mc{B}'$, we have $M(\mc{B}, \mc{B}')=\underbrace{(0,0, \ldots, 0)}_{n-{\rm{times}}}$, which matches our intuition.
\end{remark}
\begin{example}\label{Example:Macroscope}
In Figure \ref{fig:macroscope}, we illustrate a process subnetwork $\mc{B}'$ (marked with a blue rectangle) of a process network $\mc{B}$ (marked with a black rectangle) with three legs. Note that the environment of $\mc{B}'$  with respect to $\mc{B}$  is the set  $\lbr b_1, b_2 ,b_3,b_4,b_5\rbr - \lbr b_1, b_3 \rbr= \lbr b_2, b_4, b_5 \rbr$. Observe that $M(\mc{B, \mc{B}'})= (e^{b_1}_{11} + e^{b_1}_{12}, e^{b_1}_{21}, e^{b_1}_{32}+ e^{b_3}_{31} )$. The first coordinate $e^{b_1}_{11} + e^{b_1}_{12}$ and the second coordinate $e^{b_1}_{21}$ of $M(\mc{B, \mc{B}'})$ describes the $l_1$-influence and $l_2$-influence of the environment on $\mc{B}'$, respectively, and the 3rd coordinate $e^{b_1}_{32}+ e^{b_3}_{31}$ describes the influence of $\mc{B}'$ on its environment.
\end{example}

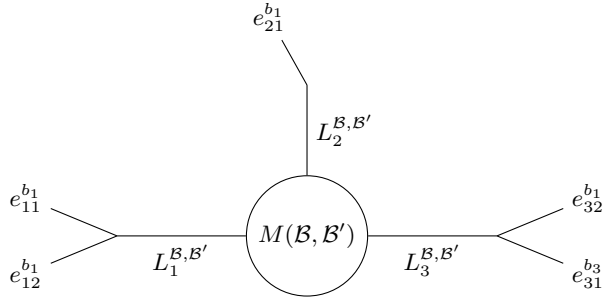
\begin{figure}[H]
\begin{center}
           \raisebox{+0.5cm}{\includegraphics[width=16cm]{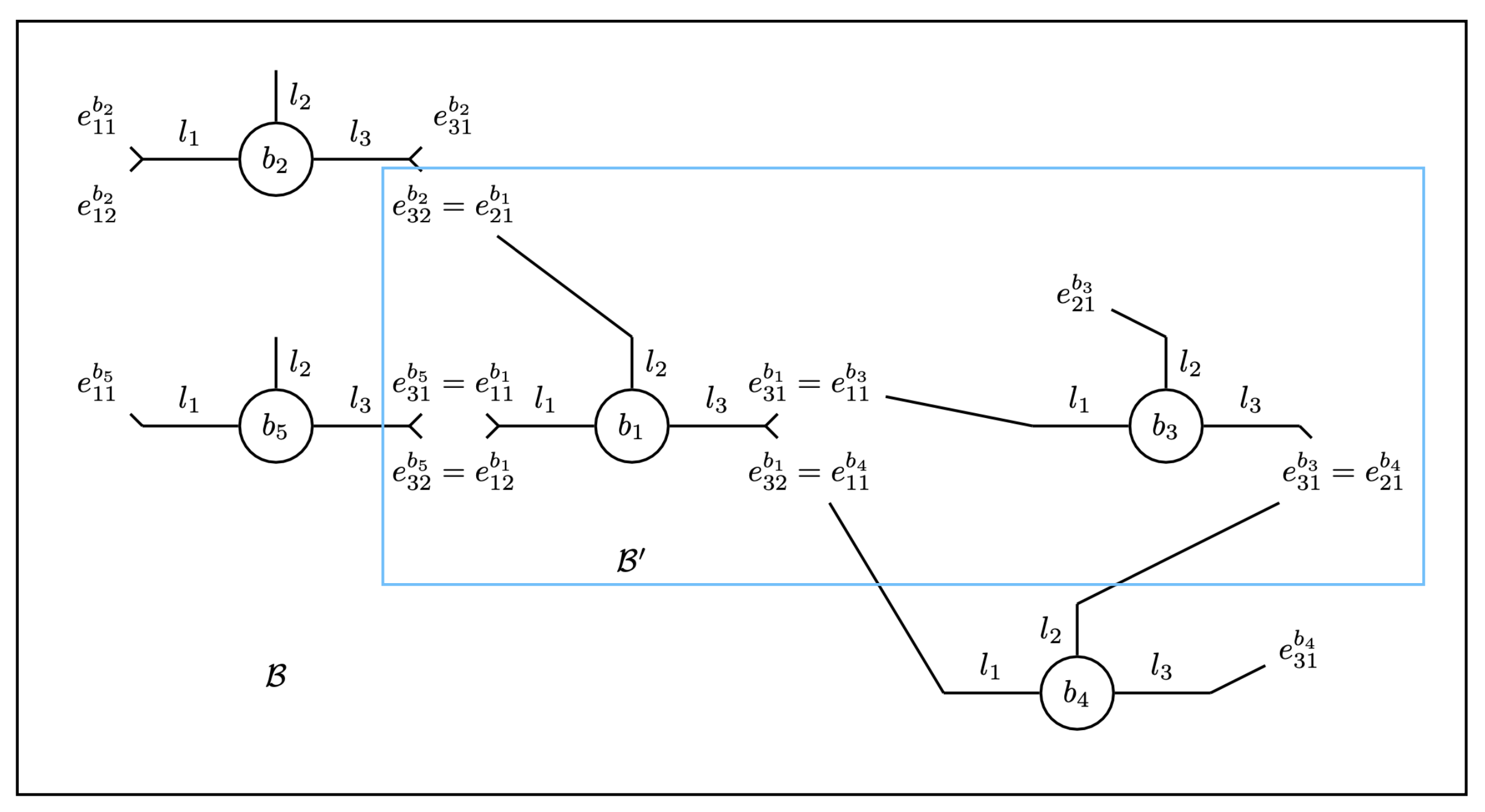}}        
              \hspace{0.5cm}
\begin{tikzpicture}[node distance = 1.5 cm, font={\small}, node font=\small]
                               \node [fill=black, single arrow, draw=none, text=black, rotate=-90] at (24.3,3.7) {compo};
		\node at (28, 3.6) [] {Applying macroscope of $\mc{B}'$ with respect to $\mc{B}$};
\end{tikzpicture}
\begin{tikzpicture}[node distance = 1.5 cm, font={\small}, node font=\small]
\node (b) at (21.7,-1) [circle, draw] {$M(\mc{B}, \mc{B}')$};
\node (e12) at (18,-0.5) {$e_{11}^{b_1}$};
\node (e11) at (18,-1.5) {$e_{12}^{b_1}$};
\node (e21) at (21.2,1.9) {$e^{b_1}_{21}$};
\node (e31) at (25.4,-0.5) {$e^{b_1}_{32}$};
\node (e32) at (25.4,-1.5) {$e_{31}^{b_3}$};
\coordinate (l1) at (19.2,-1);
\coordinate (l2) at (21.7,1);
\coordinate (l3) at (24.2,-1);
\draw[-] (e12) -- (l1);
\draw[-] (e11) -- (l1);
\draw[-] (l1) -- (b) node[pos=0.5, below] {$L^{\mc{B}, \mc{B}'}_1$};
\draw[-] (e21) -- (l2);
\draw[-] (l2) -- (b) node[pos=0.5, right] {$L^{\mc{B}, \mc{B}'}_2$};
\draw[-] (e32) -- (l3);
\draw[-] (e31) -- (l3);
\draw[-] (b) -- (l3) node[pos=0.5, below] {$L^{\mc{B}, \mc{B}'}_3$};
\end{tikzpicture}
\end{center}
        \caption{Visualisation of the macroscope discussed in Example \ref{Example:Macroscope}.}\label{fig:macroscope}
\end{figure}

\begin{example}\label{Example:SBGN-example-macroscope}
In Figure \ref{fig:MAPKmacroscope}, we illustrate Definition \ref{Defn:macroscope} on the SBGN-PD visualisation of the MAPK cascade. Here, we mark $\mc{B}'$ with blue and $\mc{B}$ with black. The macroscope $M(\mc{B}, \mc{B}')$ of $\mc{B}'$ with respect to $\mc{B}$ is $\big(0, \text{(RAF-P, macromolecule), (MEK-2P, macromolecule)} \big)$. Notations MEK-2P and RAF-P describe a double phosphorylated MEK and phosphorylated RAF, respectively. While representing entities (Definition \ref{ARDefinition: Biochemical process networks}) in the form $(a,b)$ such as (RAF-P, macromolecule) and (MEK-2P, macromolecule), the first coordinate $a$ represents the molecule, and the second coordinate $b$ represents its type. It says that the biochemical reaction network marked in blue modulates its environment via (MEK-2P, macromolecule); on the other hand, the environment modulates it via (RAF-P, macromolecule).
\end{example}
\begin{figure}[H]
        \begin{center}
           \raisebox{+0.5cm}{\includegraphics[width=8cm]{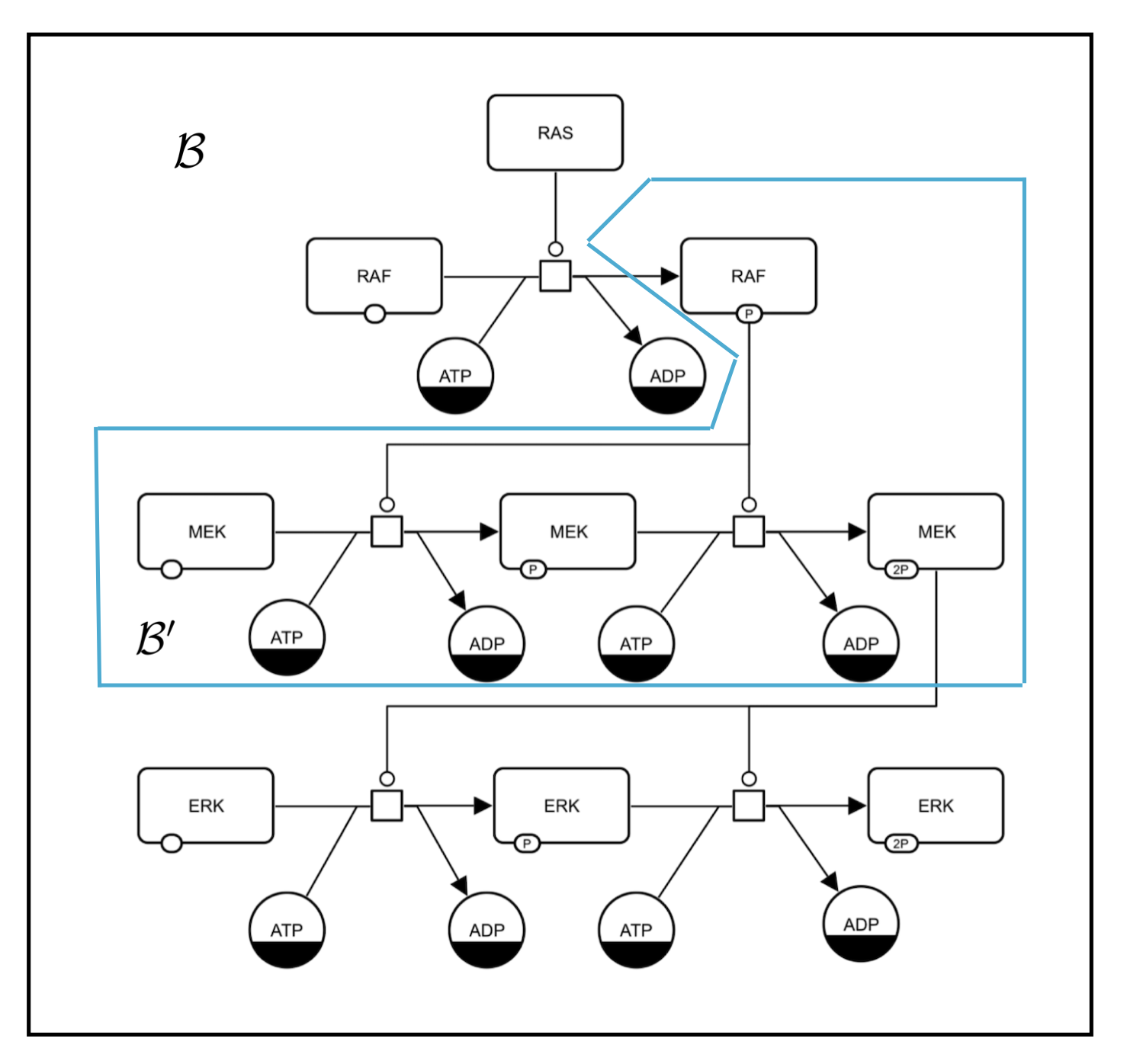}}        
              \hspace{0.5cm}
                \begin{tikzpicture}[node distance = 1.5 cm, font={\small}, node font=\small]
                               \node [fill=black, single arrow, draw=none, text=black, rotate=-90] at (24.3,3.7) {compo};
		\node at (28, 3.6) [] {Applying macroscope of $\mc{B}'$ with respect to $\mc{B}$};
\node (b) at (25.7,-1) [circle, draw] {$M(\mc{B}, \mc{B}')$};
\node (e21) at (25.2,2) {\text{(RAF-P, macromolecule)}};
\node (e31) at (30.2,-0.5) {\text{(MEK-2P, macromolecule)}};
%
\coordinate (l1) at (23.2,-1);
\coordinate (l2) at (25.7,1);
\coordinate (l3) at (28.2,-1);
%
\draw[-] (l1) -- (b) node[pos=0.5, below] {$L^{\mc{B}, \mc{B}'}_1$};
\draw[-] (e21) -- (l2);
\draw[-] (l2) -- (b) node[pos=0.5, right] {$L^{\mc{B}, \mc{B}'}_2$};
%
\draw[-] (e31) -- (l3);
\draw[-] (b) -- (l3) node[pos=0.5, below] {$L^{\mc{B}, \mc{B}'}_3$};
\end{tikzpicture}
        \end{center}
        \caption{Visualisation of a macroscope (Definition \ref{Defn:macroscope}) on the SBGN-PD visualisation of the MAPK cascade as discussed in Example \ref{Example:SBGN-example-macroscope}. The SBGN image is taken from the MAPK cascade example on Page 65 in \cite{rougny_systems_2019}. In SBGN-PD language, various geometric shapes like rectangles with rounded corners, circles, etc., represent types }\label{fig:MAPKmacroscope}
\end{figure}





\section{SBGN-PD as process networks}\label{Section5}

 The goal of this section is to sketch how one can express SBGN-PD as process networks. Currently, our treatment is not complete and is discussed only via examples. A complete and fully rigorous dictionary between the language of SBGN-PD and our process networks is left for future research. To keep our treatment self contained, we will briefly describe the vocabulary of SBGN-PD. For a detailed treatment readers are referred to \cite{rougny_systems_2019}.

\begin{figure}[htb]
\begin{center}
\includegraphics[scale=0.2]{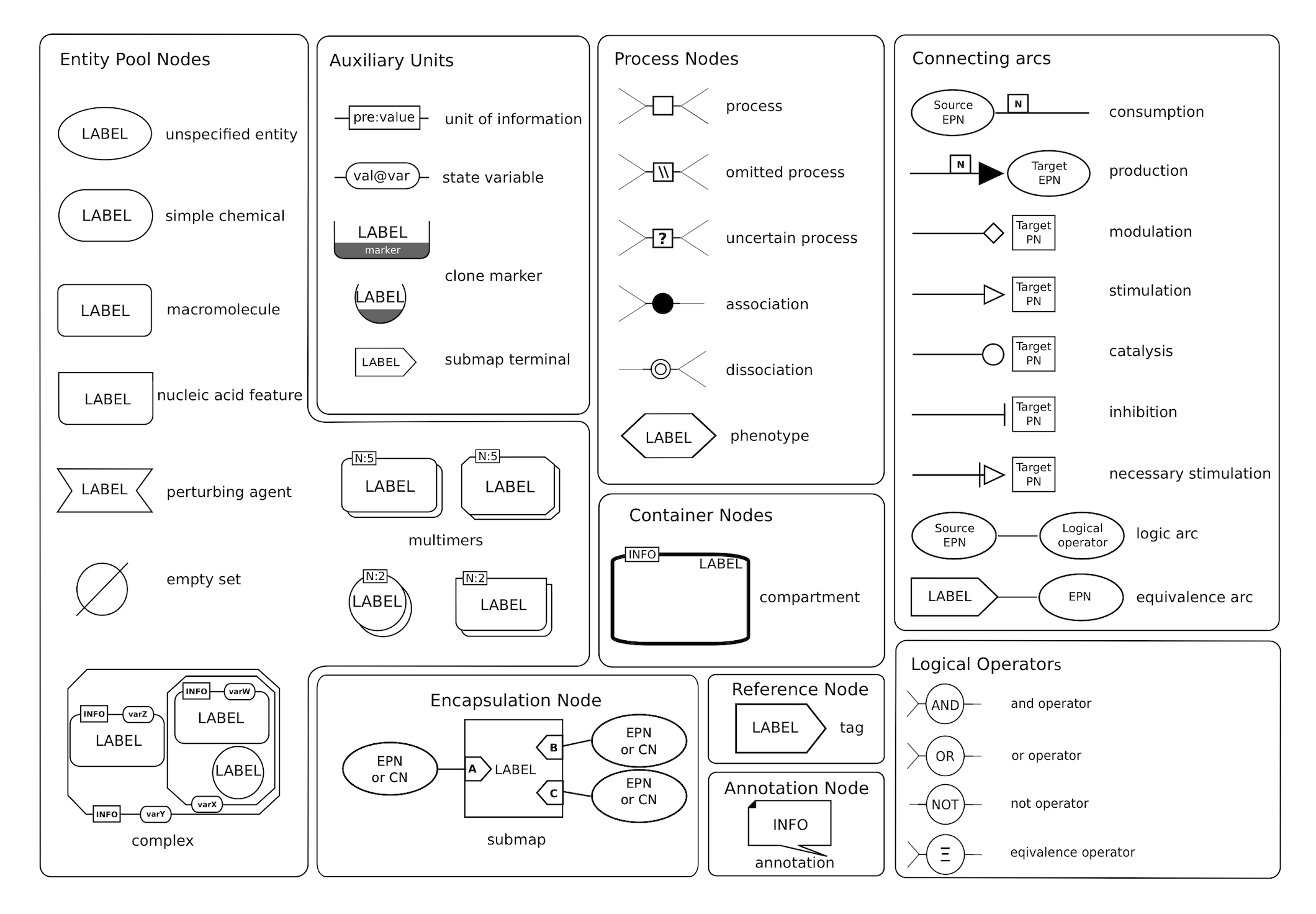}
\caption{Reference card for SBGN PD language Level 1 Version 2.0. The image is taken from Page 75, \cite{rougny_systems_2019}. 
}
\label{fig:SBGNPD-reference-card}
\end{center}
\end{figure}
\subsection{SBGN-Process Description (SBGN-PD)}

Figure \ref{fig:SBGNPD-reference-card} describes various symbols and glyphs used in SBGN-PD language. An \textbf{entity pool node} represents a population of indistinguishable biological entities in an SBGN PD. It may correspond to different levels of granularity, such as all proteins, all instances of a specific protein, or particular forms of a protein. SBGN-PD defines six glyphs for material entity pools: \textit{unspecified entity}, \textit{simple chemical}, \textit{macromolecule}, \textit{nucleic acid feature}, \textit{multimer}, and \textit{complex}. In addition, two conceptual entity pools are defined: the \textit{empty set} and the \textit{perturbing agent}. Both material and conceptual entity pools may optionally carry \textbf{auxiliary units}, which are  glyphs used to decorate other glyphs in order to convey additional information. Their presence modifies the interpretation of the associated glyph or provides supplementary descriptive context. Auxiliary units may be used to encode annotations (\textit{units of information}), represent state information (\textit{state variables}), indicate duplicated entity pool nodes (\textit{clone markers}), describe specific glyph structures (e.g., \textit{subunits of a complex}), or \textit{provide references to elements external to a portion of SBGN-PD} (e.g., submap terminals). \textbf{Process nodes} represent transformations of one or more entity pools into one or more entity pools. SBGN Process Description defines a generic process and five specific types: \textit{omitted}, \textit{uncertain}, \textit{association}, \textit{dissociation}, and \textit{phenotype}. SBGN-PD also has \textbf{container nodes} called \textit{compartments}, which are logical or physical structures that contain entity pool nodes. Each entity pool node belongs to only one compartment; consequently, the same biochemical species located in different compartments is represented as distinct entity pools.
For example, calcium ions in the endoplasmic reticulum and in the cytosol are represented as separate entity pool nodes. A \textbf{connecting arc} between an entity pool node and the process node can be of several types, depending on the kind of role the entity pool node play on the process node. In SBGN-PD, we have 
\begin{itemize}
\item \textit{consumption arcs}- the entity pool node is consumed by the process,
\item \textit{production arcs}- the entity pool node is produced by the the process,
\item \textit{modulation arcs}- the entity pool nodes modulates the process but the exact nature of the modulation is not specified or is unknown,
\item \textit{stimulation arcs}- the entity pool node affects positively the flux of the process,
\item \textit{catalysis arcs}- a particular case of stimulation, where the entity pool node catalyses the process,
\item inhibition arc- the entity pool node affects negatively the flux of the process,
\item necessary stimulation arc- the entity pool node is necessary for the process to take place.
\end{itemize}
Apart from the above 7 types of connecting arcs there are two more types viz.  \textit{logical arcs} and \textit{equivalence arcs}, which we will not discuss in this article. The \textbf{encapsulation node} called the \textit{submap} encapsulates an entire portion of SBGN-PD, including all node and edge types, within a single glyph. It is therefore distinct from an omitted process. The internal contents of the submap are hidden from the viewer, with only the submap terminals displayed.

\subsection{From a SBGN-PD diagram to a process network}
In this subsection, we demonstrate  how one can  convert a SBGN-PD diagram into a process network through two examples. 
\begin{example}\label{Example:musclesignalling}
We consider the SBGN-PD visualisation of  neuronal/muscle signalling   (Figure \ref{originalfig:musclesignalling}).
\begin{figure}[H]
\begin{center}
\includegraphics[width=10cm]{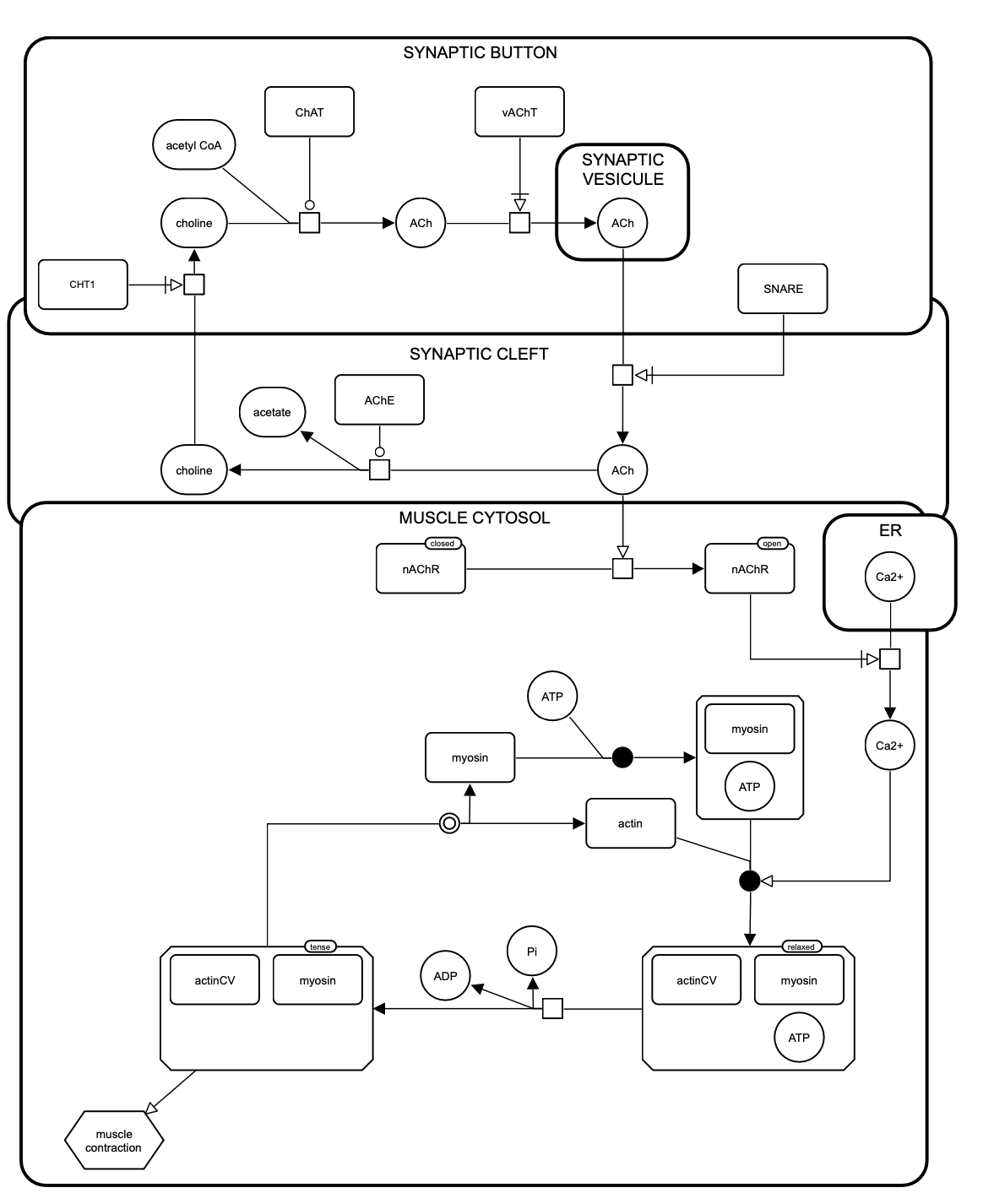}
\caption{An illustration of the SBGN-PD visualisation of the neuronal/muscle signalling, describing inter-cellular signalling. The SBGN image is taken from the neuronal/muscle signalling example on Page 66, \cite{rougny_systems_2019}.}
\label{originalfig:musclesignalling}
\end{center}
\end{figure}
\subsubsection*{Step 1:}
Without loss of generality, we label the entity pool nodes  and process nodes present in the Figure \ref{originalfig:musclesignalling}, respectively, as $e_1, e_2, \ldots, e_{23}, e_{24}$ and $b_1,b_2, \ldots, b_{12}$, see Figure \ref{fig:musclesignalling}.
\subsubsection*{Step 2:}
Without loss of generality, we enumerate the compartments SYNAPTIC BUTTON, SYNAPTIC VESICULE, SYNAPTIC CLEFT, MUSCLE CYTOSOL AND ER  present in Figure \ref{fig:musclesignalling} as $c_1, c_2, c_3, c_4$, and $c_5$ respectively. 

\subsubsection*{Step 3:}

We denote the `types' of entity pool nodes: unspecified entity, simple chemical, macromolecule, nucleic acid feature, perturbing agent, empty set, multimer and complex by $t_1,t_2, t_3, t_4,t_5,t_6, t_7, t_8$,  respectively.
\subsubsection*{Step 4:}
Let 
\begin{equation}\nonumber
T= \lbrace t_1,t_2, \cdots, t_8 \rbrace
\end{equation}
\begin{equation}\nonumber
C=\lbrace c_1,c_2, \cdots, c_5 \rbrace
\end{equation}
\begin{equation}\nonumber
B= \lbrace b_1,b_2, \cdots, b_{12} \rbrace
\end{equation}
\begin{equation}\nonumber
\mathcal{E} = \lbrace e_1,e_2, \cdots, e_{24} \rbrace
\end{equation}

\subsubsection*{Step 5:}

Let 

\begin{itemize}
\item $E= \mathcal{E} \times C \times  T$
\item \begin{equation*}
    \begin{split}
    l_1 \colon & B \ra \mathbb{B}[E]\\
    & b_1 \mapsto (e_1, c_1, t_2) + (e_2, c_1, t_{2}) \\
    & b_2 \mapsto (e_4, c_1, t_2)\\
    & b_3 \mapsto (e_6, c_2, t_2)\\
    & b_4 \mapsto (e_8, c_3, t_2)\\
    & b_5 \mapsto (e_{11}, c_3, t_2)\\
    & b_6 \mapsto (e_{14}, c_4, t_3)\\
    & b_7 \mapsto (e_{15}, c_5, t_2 )\\
    & b_8 \mapsto (e_{17}, c_4, t_{8}) + (e_{20}, c_4, t_3)\\
    & b_9 \mapsto (e_{19}, c_{4},t_3) +(e_{18}, c_4, t_2)\\
    & b_{10} \mapsto (e_{24}, c_{4}, t_{8})\\
    & b_{11} \mapsto (e_{21}, c_{4}, t_{8})\\
    & b_{12} \mapsto 0
    \end{split} 
\end{equation*}
\begin{equation*}
    \begin{split}
    l_2 \colon & B \ra \mathbb{B}[E]\\
    & b_1 \mapsto 0 \\
    & b_2 \mapsto 0 \\
    & b_3 \mapsto 0 \\
    & b_4 \mapsto 0 \\
    & b_5 \mapsto 0 \\
    & b_6 \mapsto (e_{8}, c_3, t_2)\\
    & b_7 \mapsto 0 \\
    & b_8 \mapsto (e_{16}, c_4, t_{2})\\
    & b_9 \mapsto 0 \\
    & b_{10} \mapsto 0 \\
    & b_{11} \mapsto 0 \\
    & b_{12} \mapsto (e_{24}, c_{4}, t_{8})
    \end{split} 
\end{equation*}
\begin{equation*}
    \begin{split}
    l_3 \colon & B \ra \mathbb{B}[E]\\
    & b_1 \mapsto (e_3, c_1, t_3) \\
    & b_2 \mapsto 0 \\
    & b_3 \mapsto 0 \\
    & b_4 \mapsto (e_9, c_3, t_3)\\
    & b_5 \mapsto 0\\
    & b_6 \mapsto 0\\
    & b_7 \mapsto 0 \\
    & b_8 \mapsto 0 \\
    & b_9 \mapsto 0 \\
    & b_{10} \mapsto 0 \\
    & b_{11} \mapsto 0 \\
    & b_{12} \mapsto 0
    \end{split} 
\end{equation*}
\begin{equation*}
    \begin{split}
    l_4 \colon & B \ra \mathbb{B}[E]\\
    & b_1 \mapsto 0 \\
    & b_2 \mapsto (e_5, c_1, t_3)\\
    & b_3 \mapsto (e_7, c_1, t_3)\\
    & b_4 \mapsto 0 \\
    & b_5 \mapsto (e_{12}, c_1, t_3)\\
    & b_6 \mapsto 0\\
    & b_7 \mapsto (e_{13}, c_4, t_3 )\\
    & b_8 \mapsto 0 \\
    & b_9 \mapsto 0 \\
    & b_{10} \mapsto 0 \\
    & b_{11} \mapsto 0 \\
    & b_{12} \mapsto 0
    \end{split} 
\end{equation*}
\begin{equation*}
    \begin{split}
    l_5 \colon & B \ra \mathbb{B}[E]\\
    & b_1 \mapsto (e_{4}, c_
    1, t_{2})\\
    & b_2 \mapsto (e_6, c_2, t_2)\\
    & b_3 \mapsto (e_8, c_3, t_2)\\
    & b_4 \mapsto (e_{11}, c_3, t_2) + (e_{10}, c_{3}, t_{2})\\
    & b_5 \mapsto (e_{2}, c_1, t_2)\\
    & b_6 \mapsto (e_{13}, c_4, t_3)\\
    & b_7 \mapsto (e_{16}, c_4, t_2 )\\
    & b_8 \mapsto (e_{21}, c_{4}, t_{8})\\
    & b_9 \mapsto (e_{17}, c_{4},t_8)\\
    & b_{10} \mapsto (e_{19}, c_{4}, t_{3}) + (e_{20}, c_{4}, t_{3})\\
    & b_{11} \mapsto (e_{24}, c_{4}, t_{8}) + (e_{23}, c_{4}, t_{2}) + (e_{22}, c_4, t_{2})\\
    & b_{12} \mapsto 0
    \end{split} 
\end{equation*}

\end{itemize}

Thus, we showed that the SBGN-PD in Figure \ref{originalfig:musclesignalling} can be expressed as the process network $(E, B, \lbr l_{i}\rbr_{5})$ with 5 legs. Let $\mathcal{L}= \lbr l_1, l_2, \cdots, l_{5}\rbr$. Now, if we apply our $\Sigma_{B}$-function $\Sigma_{B} \colon \mathcal{L} \to \mathbb{B}[E]$ (Definition \ref{Defn:Sigma-B-function}), then 
\begin{itemize}
\item $\Sigma_{B}(l_{1})$ collects the entity pool nodes associated to consumption arcs present in the SBGN-PD visualisation in Figure \ref{originalfig:musclesignalling},
\item $\Sigma_{B}(l_{2})$ collects the entity pool nodes associated to stimulation arcs present in the SBGN-PD visualisation in Figure \ref{originalfig:musclesignalling},
\item $\Sigma_{B}(l_{3})$ collects the entity pool nodes associated to catalysis arcs present in the SBGN-PD visualisation in Figure \ref{originalfig:musclesignalling},
\item $\Sigma_{B}(l_{4})$ collects the entity pool nodes associated to necessary stimulation arcs present in the SBGN-PD visualisation in Figure, \ref{originalfig:musclesignalling}
\item $\Sigma_{B}(l_{5})$ collects the entity pool nodes associated to production arcs present in the SBGN-PD visualisation in Figure \ref{originalfig:musclesignalling}. 
\end{itemize}
\begin{figure}[H]
\begin{center}
\includegraphics[width=7cm]{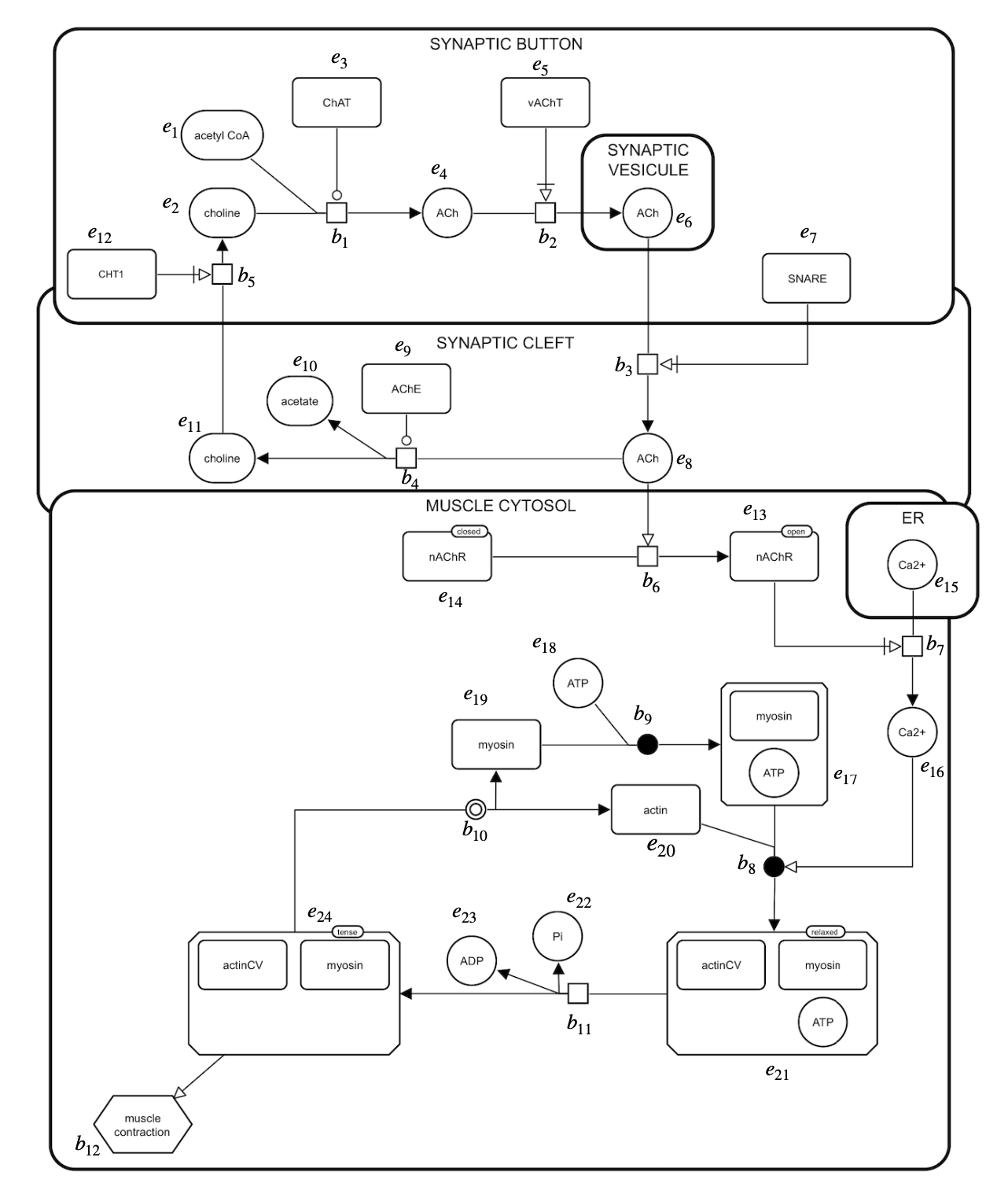}
\caption{We label the entity pool nodes and the process nodes  of Figure \ref{originalfig:musclesignalling} respectively, with the elements of the sets $\lbr e_1,e_2, \ldots, e_{24} \rbr$ and $\lbr b_1,b_2, \ldots, b_{12} \rbr$.}
\label{fig:musclesignalling}
\end{center}
\end{figure}

\end{example}

\begin{example}\label{Example:MAPK}
We consider the SBGN-PD visualisation of MAPK-cascade (Figure \ref{fig:MAPKexample}).

\subsubsection*{Step 1:}
Without loss of generality, we label the entity pool nodes  and process nodes present in the Figure \ref{originalfig:musclesignalling}, respectively, as $e_1, e_2, \ldots, e_{18}, e_{19}$ and $b_1,b_2, \ldots, b_{5}$, see Figure \ref{fig:MAPKexamplemarked}. Observe that although clone markers such as ATP and ADP are allowed to repeat according to SBGN-PD visualization rules, in our treatment we have distinguished clone markers associated to  different process nodes so that we are allowed to legitimately apply our compositionality results (Lemma \ref{ARLemma:Pushout Biochemical process species}) to decompose Figure \ref{fig:MAPKexample} into modules. 

\subsubsection*{Step 2:}
Observe that all the entity pool nodes in the SBGN-PD visualization in the Figure \ref{fig:MAPKexample} lies in the same compartment. We denote it by $c$.

\subsubsection*{Step 3:}
We denote the `types' of entity pool nodes: unspecified entity, simple chemical, macromolecule, nucleic acid feature, perturbing agent, empty set, multimer and complex by $t_1,t_2, t_3, t_4,t_5,t_6, t_7, t_8$,  respectively.

\subsubsection*{Step 4:}
Let 
\begin{equation}\nonumber
\begin{aligned}
T &= \lbrace t_1,t_2, \cdots, t_8 \rbrace \\
C &= \lbrace c \rbrace \\
B &= \lbrace b_1,b_2, \cdots, b_{5} \rbrace \\
\mathcal{E} &= \lbrace e_1,e_2, \cdots, e_{19} \rbrace
\end{aligned}
\end{equation}
\subsubsection*{Step 5:}
Let 
\begin{itemize}
\item $E= \mathcal{E} \times C \times T$.
\item 
\begin{equation*}
    \begin{split}
    l_1 \colon & B \ra \mathbb{B}[E]\\
    & b_1 \mapsto (e_1, c, t_3) + (e_{10}, c, t_{2}) \\
    & b_2 \mapsto (e_4, c, t_3) + (e_{12}, c, t_2)\\
    & b_3 \mapsto (e_5, c, t_3) + (e_{14}, c, t_2)\\
    & b_4 \mapsto(e_7, c, t_3) + (e_{16}, c, t_2)\\
    & b_5 \mapsto (e_8, c, t_3) + (e_{18}, c, t_2)
    \end{split} 
\end{equation*}
\begin{equation*}
    \begin{split}
    l_2 \colon & B \ra \mathbb{B}[E]\\
    & b_1 \mapsto (e_2, c, t_3) \\
    & b_2 \mapsto (e_3, c, t_3) \\
    & b_3 \mapsto (e_3, c, t_3)\\
    & b_4 \mapsto(e_6, c, t_3)\\
    & b_5 \mapsto (e_6, c, t_3)
    \end{split} 
\end{equation*}
\begin{equation*}
    \begin{split}
    l_3 \colon & B \ra \mathbb{B}[E]\\
    & b_1 \mapsto (e_3, c, t_3) + (e_{11}, c, t_{2}) \\
    & b_2 \mapsto (e_5, c, t_3) + (e_{13}, c, t_2)\\
    & b_3 \mapsto (e_6, c, t_3) + (e_{15}, c, t_2)\\
    & b_4 \mapsto(e_8, c, t_3) + (e_{17}, c, t_2)\\
    & b_5 \mapsto (e_9, c, t_3) + (e_{19}, c, t_2)
    \end{split} 
\end{equation*}
Hence, the SBGN-PD in Figure \ref{fig:MAPKexample} can be expressed as the process network $(E, B, \lbr l_{i} \rbr_{3})$ with 3 legs. In a similar way as in the Example \ref{Example:musclesignalling},  using the $\Sigma_{B}$ function of Definition \ref{Defn:Sigma-B-function}, we can collect the entity pool nodes associated to consumption arcs, catalysis arcs and production arcs present in the SBGN-PD visual representation of the MAPK-cascade. Furthermore, note that since all entity pool nodes belong to the same compartment in the SBGN-PD visualisation of MAPK cascade, the information about compartments can be implicitly captured without using the second cordinate in the three coordinate representation of entity pool nodes. We took this approach in Figure \ref{ARDefinition: Biochemical process networks}, Subsection \ref{subsection:Example-of-process-networks} and in the Example \ref{Example:SBGN-example-macroscope}.

\end{itemize}

\end{example}

\begin{remark}\label{Remark:Limitationsoftranslations}
Observe that in our above method of translating the  SBGN-PD diagram into a process network, we represent the entities (as defined in Definition \ref{Defn:Processsubnetwork}) as a triple $(x,y,z)$, where $x$ denotes an  entity pool node, $y$ denotes its type and $z$ denotes the compartment it belongs.
Thus, our translation methods could capture  (i) the `type information' of entity pool nodes, (ii) the information of the compartment to which an entity pool node belongs,  (iii) the `type information' of connecting arcs and (iv) the overall `connectivity structure' of the underlying biochemical reaction network.
However, we have not captured the auxillary units, the component entity pool nodes of a complex and  the `type information' of process nodes.
\end{remark}
\begin{remark}
Observe that in our definition of a process network (Definition \ref{ARDefinition: Biochemical process networks}) $n$ legs  define $n$ different abstract types of connecting arcs, each connecting a process node with a set of entity pool nodes performing the same type of function on the process node. However, there are only $9$ connecting arcs according  to the SBGN-PD language level 1 version 2.0 (see Figure \ref{fig:SBGNPD-reference-card}). Now, SBGN is an evolving visual notational system and our treatment was based on the current verion: SBGN-PD language level 1, version 2.0, and thus there may be a possibility that in the future levels and versions of SBGN-PD notational systems, new types of connecting arcs may appear according to the newer data available to biologists. Thus, we kept our framework more flexible and at the same time simple enough to demonstrate that one can suitably adapt the ideas of Baez-Masters’ Open Petri nets \cite{MR4085076} in the context of composing SBGN-PD diagrams. Furthermore, note that in Example \ref{Example:musclesignalling} and Example \ref{Example:MAPK},  respectively, SBGN-PD diagrams are realized as  process networks with 5 legs and 3 legs. Although by the virtue of Proposition \ref{Proposition:embedding}, we can actually consider them as process networks with $9$ legs, we feel doing so  is unnecessary and cumbersome.

\end{remark}
\begin{remark}
A main motivation for our efforts on translating a SBGN-PD diagram into a  process network is to use our Theorem \ref{ARMain Theorem1} to build a SBGN-PD by composing smaller SBGN-PD's using the composition laws of the symmetric monoidal double category of process networks. Now, in order to use Theorem \ref{ARMain Theorem1} legitimately after translating the SBGN-PD into a process network (like the way we did), we need to make sure that the entity pool nodes which we are identifying via pushouts are of same types and they lie in the same compartment. From the point of biology, we believe this restriction still makes our model effective in our effort on building a formal composable framework to study SBGN-PD's.  
\end{remark}
\begin{figure}[H]
\begin{center}
\includegraphics[width=8cm]{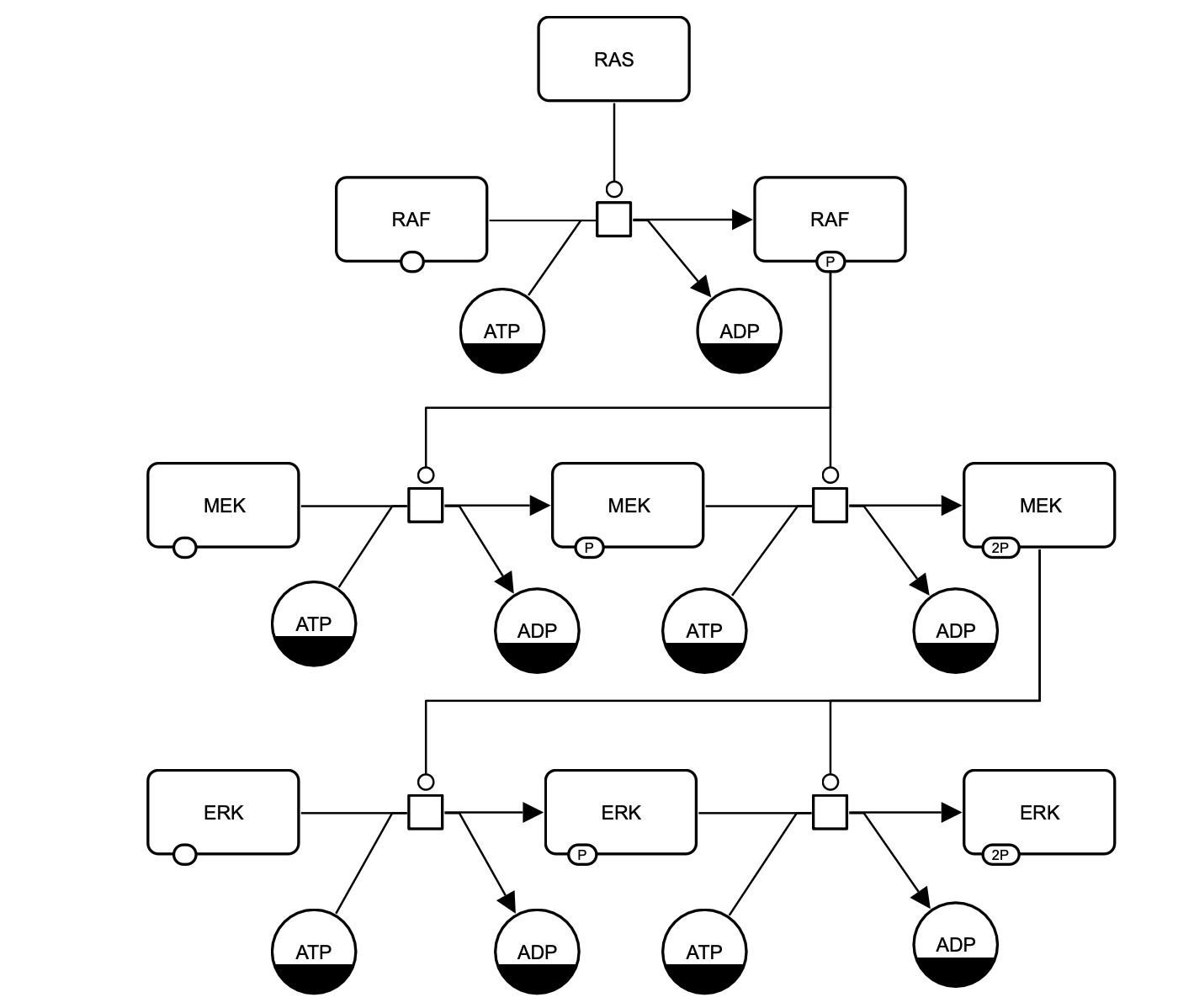}
\caption{SBGN-PD visualisation of MAPK cascade. The SBGN image is taken from the Page 65, \cite{rougny_systems_2019}.}
\label{fig:MAPKexample}
\end{center}
\end{figure}

\begin{figure}[H]
\begin{center}
\includegraphics[width=8cm]{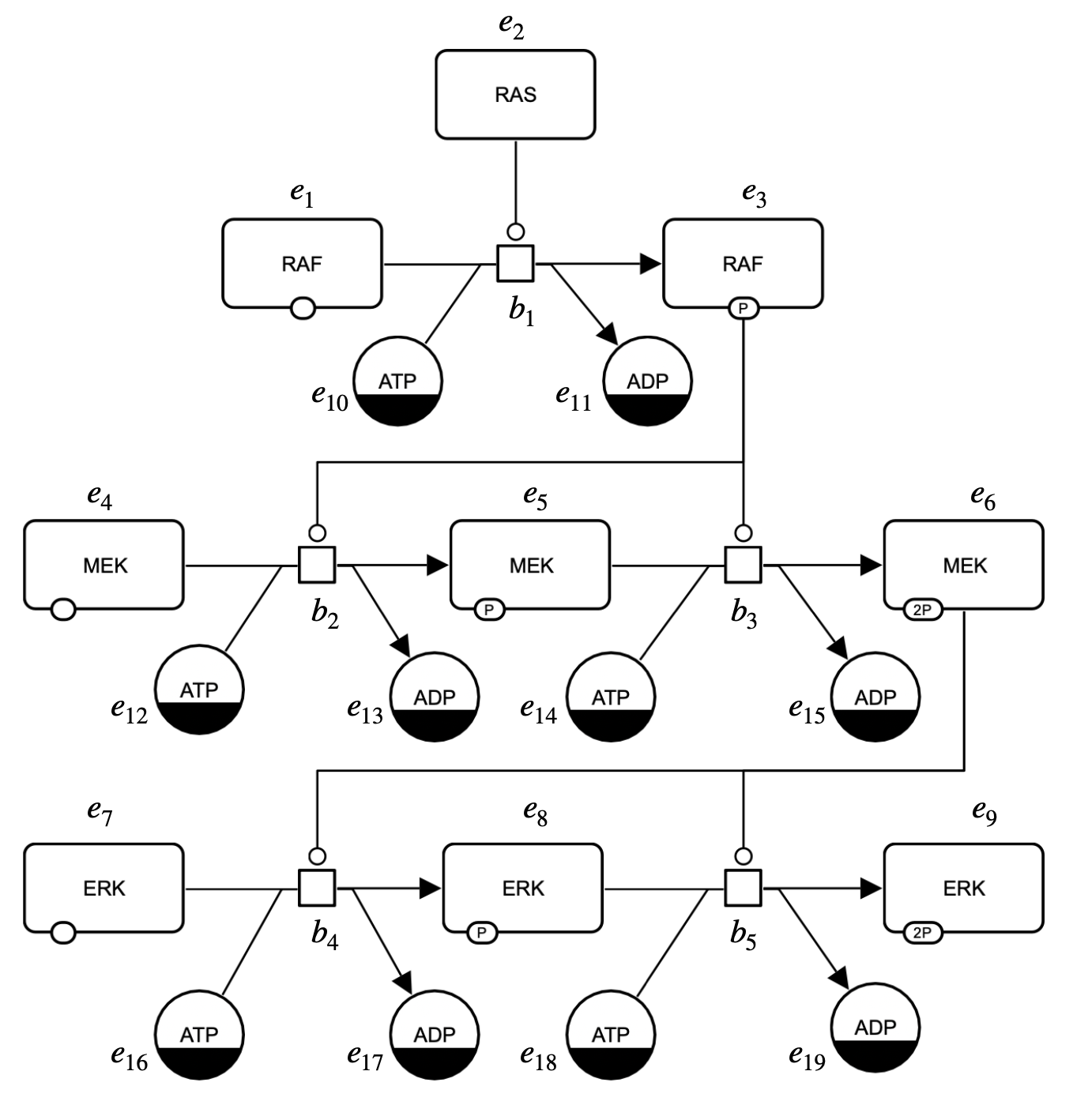}
\caption{ We label the entity pool nodes and the process nodes of Figure \ref{fig:MAPKexample} respectively, with the elements of the sets $\lbr e_1,e_2, \ldots, e_{19} \rbr$ and $\lbr b_1,b_2, \ldots, b_{5} \rbr$.}
\label{fig:MAPKexamplemarked}
\end{center}
\end{figure}

\section{Discussion and outlook}\label{Sec:discussion-and-conclusion}

\subsection{Achievements}
To the best of our knowledge, we believe we are the first ones to analyse SBGN diagrams in Systems Biology using category theory. Back in 2011, John C. Baez mentioned about such a possibility in his blog \cite{baez_network_theory_part1_2011}. In this regard, we wanted to use the existing results in Applied Category Theory to develop a compositional framework for SBGN-PD, which led us to the theory of \textit{open Petri nets} developed by  John C. Baez and Jade Master in \cite{MR4085076}. Although the definition of our process network is motivated from the Baez-Master's notion of Petri net, it differs in two crucial ways:
\begin{itemize}
\item[(i)] Instead of just a pair of maps (source and target) as in the case of Petri net, our process network has room for more maps $\lbr l_i\rbr_{n}$ (\textit{legs})  to accommodate entity pool nodes associated to connecting arcs such as modulation arc, stimulation arc, catalysis arc, inhibition arc and necessary stimulation arc present in the SBGN-PD formalism (see Figure \ref{fig:SBGNPD-reference-card}).
\item[(ii)]  the maps $\lbr l_i\rbr_{n}$ are valued in $\mathbb{B}[E]$ instead of $\mb{N}[E]$ to accommodate the fact that in a SBGN-PD  visualisation of a biochemical reaction network, the stoichiometry remains absent. Furthermore, the  additive monoid structure of $\mathbb{B}$ has been found to be necessary in introducing  our notion of `macroscope', especially while defining the $\Sigma_{B}$-function (Definition \ref{Defn:Sigma-B-function}).
\end{itemize}

 We were able to suitably adapt the ideas of Baez-Master's Open Petri nets to obtain a symmetrical monoidal double category whose horizontal 1-morphisms are our open process networks. However, the novelty of our compositional framework lies in the following aspects:
\begin{itemize}
\item Connecting the structured cospan-based compositional framework of open Petri nets to the world of SBGN in Systems Biology.
\item Observing how morphisms of process networks can be used as tools to zoom out details in SBGN-PD diagrams (Example \ref{Example:Existence-of-morphisms-of-process-networks}).
\item Observing how horizontal composition of 1-morphisms in ${\rm{Open}}_{\rm{Double}}({\rm{\textbf{Process}}}_n)$ model a compatibility condition between compositionality (via pushouts constructions) and zooming out process (via morphisms of process networks) (Example \ref{Example:Compatibility  between compositionality and zooming-out process}). 
\end{itemize}
Another important contribution of our framework lies in our notion of macroscope of a process network (section \ref{Subsection:Macroscope of a process network: A tool to study the influence of a process network on its environment and vice versa}). In the current era of data proliferation, biochemical pathways employed by biologists (for example, \cite{kanehisa_kegg_2024}) are often extremely large and structurally complex. The abundance of detailed information can obscure the key qualitative features and organizational principles underlying these pathways. This motivates the need for methodological tools that enable a meaningful “zooming out” of biochemical pathways, allowing essential structures and behaviors to be discerned without being overwhelmed by fine-grained detail. The notion of a macroscope introduced in this work contributes toward addressing this need. More concretely, let $\Gamma$ be an SBGN-PD visualization of a large biochemical reaction network. By applying our macroscope to a selected portion
$S$ of $\Gamma$, we can abstract away the internal reaction-level details of $S$ while preserving its functional interface with the remainder of the network. In particular, this construction captures how biochemical entities produced outside $S$ influence reactions within $S$, as well as how entities produced within $S$ affect reactions occurring outside $S$. From the nature of results obtained in this paper, we anticipate that our approach would cast a new light on the \textit{size issue} faced by biologists in studying the behavior of biochemical reaction networks. More elaborately, biologists may use our results 
\begin{itemize}
\item to study a large biochemical reaction network compositionally,
\item to zoom-out unnecessary details in large biochemical reaction pathways using our macroscope and our morphism of process networks.
\end{itemize}
However, in order to really implement our ideas in solving real biological problems, our results need to be translated in a computational environment. It would be interesting to translate our ideas into computational form
using the already existing category theory based computational platforms like AlgebraicJulia \cite{noauthor_algebraicjulia_nodate} and CatColab \cite{catcolab_help}. 

\subsection{Limitations}
Despite the achievements we discussed in the previous subsection, our present state of the work have some limitations which we are going to state next:
\begin{itemize}
\item[(i)] In Section \ref{Section5}, although we demonstrated conversion of SBGN-PD diagrams into our process networks via examples, a formal dictionary between an arbitrary SBGN-PD diagram and process networks is still missing.
\item[(ii)] In the conversion method discussed in Section \ref{Section5}, we could not capture the information of auxillary units, the component entity pool nodes of a complex and the `types' of process nodes (see Remark \ref{Remark:Limitationsoftranslations}). Also, we have not discussed about logic arcs and equivalence arcs.
\item[(iii)] As demonstrated in Example \ref{Example:musclesignalling} and Example \ref{Example:MAPK}, we can express a SBGN-PD as a process network such that we can capture the `type information' of entity pool nodes, the information of the compartment to which an entity pool node belongs,   the `type information' of connecting arcs and the overall `connectivity structure' of the underlying biochemical reaction network. However, the crucial information like compartment, type information of entity pool node and the component entity pool nodes of a complex are not captured in the definition of process network (Definition \ref{ARDefinition: Biochemical process networks}) itself.
\end{itemize}

In light of the above limitations, it is clear that further work is required to develop a fully effective ACT-based compositional theory for SBGN-PD. The contribution of this manuscript represents a step in this direction, demonstrating that ideas of symmetric monoidal double categories can nonetheless provide useful insights for the analysis of large-scale SBGN-PD diagrams. Our choice of using structured cospans based framework over other ACT based compositional framework is motivated by the existing results on Open Petri nets \cite{MR4085076}. 

\subsection{Future directions}

Among many possible  future directions, we mention a few here:
\begin{itemize}
\item To address the current limitations of the framework and to develop more comprehensive compositional tools for SBGN-PD. It may be be useful to see process networks as presheaves.
\end{itemize}

\begin{itemize}
\item Implement the ideas of this paper to develop category theory based computational tools using platforms like AlgebraicJulia \cite{noauthor_algebraicjulia_nodate} and CatColab \cite{catcolab_help} for real applications in studying biochemical reaction pathways compositionally. 
\item To develop formal translation methods (preferably functorial) of our process networks into exisitng formalisms for SBGN-PD like  asynchronous automata networks \cite{rougny2016qualitative}, Hybrid Functional Petri Net (HFPN) \cite{Loewe2011} and textual representations (SBGNtext) \cite{cherdal2018sbgn2hfpn}.
\item In standard practice, the modeling and simulation of biological networks are often treated as separate tasks. Since SBGN-PD provides a visual representation of biochemical reaction networks, the construction of a symmetric monoidal double functor from $${\rm{Open}}_{\rm{Double}}({\rm{\textbf{Process}}}_n)$$ (as in Theorem~\ref{ARMain Theorem1}), to a suitable symmetric monoidal double category whose horizontal 1-morphisms are `open' ODE-based dynamical systems (after assigning appropriate kinetic parameters) will provide a systematic link between modeling and simulation. This will be further strengthened by the fact there are already category theory based computational platforms like AlgebraicJulia \cite{noauthor_algebraicjulia_nodate} and CatColab \cite{catcolab_help} which can simulate structured cospans based models.



\item To relate our process network description of SBGN-PD with the  molecular systems biology inspired rule based-language for modeling interacting agents, called \textit{Kappa} \cite{kappalanguage_manualcite}.

\item Besides SBGN-PD, there are two other SBGN languages, viz. SBGN-AF and SBGN-ER. They are also commonly used for visualising biochemical networks at levels of granularity different from the one explored in this paper. It would be interesting to develop effective ACT-based compostional frameworks for  SBGN-AF and SBGN-ER, and investigate the existence of functorial interrelationships between three complementary SBGN languages viz. SBGN-PD, SBGN-AF and SBGN-ER. It is worth mentioning that recently, our first named author and John C. Baez developed a mathematical framework in \cite{baez2025graphspolarities}, which has been demonstrated in the same paper to be useful in studying SBGN-AF diagrams compositionally. More concretely, in \cite{baez2025graphspolarities},  a structured cospan-based compositional framework for monoid-labeled graphs was introduced and certain specific classes of SBGN-AF diagrams were realized as  special cases (See Example 3.5 in \cite{baez2025graphspolarities}).

\item In this paper, we have not explored how our macroscope behaves with the structured cospans-based compositional framework of process networks. It might be interesting to investigate in this direction.
\end{itemize}

We believe the future directions outlined above make it evident that the present work opens a broad and promising avenue for the development of deep and fruitful connections between SBGN and category theory, with substantial scope for further conceptual and technical advances.

\section{Acknowledgements}

The authors express their gratitude to Heike Siebert for discussions on a preliminary version of this manuscript. They also express their sincere thanks to John C. Baez for various helpful discussions on Applied Category Theory in Biology. They also gratefully acknowledge funding via the Pilot Study 12 — Networked Matter grant of the Life, Light \& Matter Interdisciplinary Faculty of the University of Rostock. The authors thank the anonymous referees for their careful reading of the manuscript and their constructive comments, which helped to improve the clarity and quality of this work.






\bibliographystyle{plainnat}
\bibliography{references}


\end{document}